\documentclass[10pt,a4paper]{article}
\usepackage{amsmath,amssymb,amsfonts,amsthm,mathtools}
\usepackage{dsfont} 
\usepackage{stmaryrd}
\usepackage{epsf}
\usepackage{amscd}
\usepackage{amsopn}

\usepackage[english]{babel}
\usepackage[T1]{fontenc}      
\usepackage[utf8]{inputenc}
\usepackage[colorlinks=true, allcolors=blue]{hyperref}
\usepackage{setspace}
\usepackage{mathrsfs}
\usepackage{natbib}
\usepackage{epsfig}
\usepackage{graphicx}
\usepackage{bm}
\usepackage{xcolor}
\usepackage{comment}
\usepackage{tabularx}
\usepackage{multirow}
\newtheorem{coro}{Corollary}[section]
\newtheorem{Theo}{Theorem}[section]
\newtheorem{assumption}[Theo]{Assumption}
\newtheorem{lemm}{Lemma}[section]
\newtheorem{rema}{Remark}[section]
\newtheorem{prop}{Proposition}[section]

\newcommand{\real}{\mathbb R}
\newcommand{\Rplus}{\mathbb{R}_{+}}

\newcommand{\dint}{\displaystyle\int}

\newcommand{\calC}{{\mathcal C}}

\newcommand{\calF}{{\mathcal F}}

\newcommand{\calL}{{\mathcal L}}

\newcommand{\bc}{\begin{center}}
	\newcommand{\ec}{\end{center}}
\newcommand{\benu}{\begin{enumerate}}
	\newcommand{\eenu}{\end{enumerate}}

\newcommand{\E}{\mathbb{E}}
\newcommand{\p}{\mathbb{P}}

\newcommand{\tto}{\longrightarrow}
\newcommand{\dd}{\,\mathrm{d}}

\usepackage{cases}
\usepackage{systeme}
\usepackage{enumerate}
\usepackage{cancel}
\usepackage{fancyhdr}
\usepackage{cleveref}
\usepackage{chngcntr}
\numberwithin{equation}{section}

\newcommand{\seclabel}[1] {\label{#1}}

\usepackage{xspace}
\usepackage{lmodern} 
\usepackage{fullpage}

\usepackage{lipsum}
\makeatletter
\def\blfootnote{\gdef\@thefnmark{}\@footnotetext}
\makeatother

\usepackage{geometry}
\newcommand{\ind}{\mathds{1}}
\usepackage{booktabs}
\usepackage{float}

\newcommand{\R}{\mathbb{R}}

\usepackage{algorithm}
\usepackage{algorithmic} 

\title{\bf Random Jump Intensities and Bernstein Density Estimation in Ergodic Basic Affine Jump-Diffusion Processes}

\author{Hamdi Fathallah\thanks{Universit\'e de Sousse, Laboratoire LAMMDA, \'Ecole Sup\'erieure des Sciences et de Technologie de Hammam Sousse, Rue Lamine El Abbessi 4011 Hammam Sousse, Tunisie. E-mail: \texttt{hamdi.fathallah@essths.u-sousse.tn}.}}
\date{}
\begin{document}
\maketitle
\begin{abstract}
This paper develops a random-effects extension of the ergodic basic affine
jump-diffusion (BAJD) model, originally proposed by Duffie and Gârleanu
\cite{DG2001} as a jump-type extension of the
Cox-Ingersoll-Ross (CIR) model. We consider a
population of independent BAJD trajectories whose jump intensities are
individual unobserved random effects, while the drift parameters, diffusion
coefficient, and jump-size distribution are common across the population. The
statistical objective is to estimate both the individual jump intensities and
their common density. For continuous-time observation, the jumps are
pathwise observable, and the extended maximum likelihood estimator (MLE) of each
random intensity is the empirical jump frequency. It coincides with the
ordinary MLE after the first observed jump. We establish its strong
consistency and stable mixed-normal limit as the observation horizon tends to
infinity. Using the estimated intensities as pseudo-observations, we then
construct a Bernstein-polynomial estimator of the random-effect density and
derive its bias, variance, mean integrated squared error, and pointwise
asymptotic normality under a sequential asymptotic framework. The general
theory is developed for an arbitrary density of positive jump-sizes within a
finite-activity jump framework and is then specialized to the
Gamma\((k,\lambda)\) family, including
the exponential and Erlang models. In this specialization, we derive a
pooled MLE of the common rate parameter, establish its consistency and
asymptotic normality, and prove its first-order conditional asymptotic
independence from the individual intensity estimators. We also construct a
consistent plug-in estimator of the population stationary mean. The
theoretical results are complemented by a controlled simulation study and a
reproducible empirical illustration based on daily financial
realized-volatility data. The empirical analysis applies the two-stage
procedure to Exchange-Traded Funds (ETFs) and large-cap equities and examines its sensitivity to the jump-detection threshold.
\end{abstract}

\noindent\textbf{Keywords:} Basic affine jump-diffusion process; jump-type CIR process; random effects; compound Poisson jumps; maximum likelihood estimation; nonparametric density estimation; Bernstein polynomials; exponential jumps; Gamma jumps; empirical financial data.

\noindent\textbf{Mathematics Subject Classification 2020:} 60H10; 60J75; 62F12; 62G07; 91G70.
\blfootnote{\sffamily This research is supported by the PHC-Utique project no.17G1505.}
\section{Introduction}

Stochastic differential equations (SDEs) provide a natural mathematical framework for modeling systems that combine a deterministic trend with random fluctuations. By incorporating a stochastic input, typically a Brownian motion, into the classical framework of ordinary differential equations, SDEs make it possible to account for random perturbations and exogenous shocks affecting the evolution of a system. Such models are widely used in physics, mathematical finance, biology, engineering and economics. Statistical studies of these so-called diffusion processes are by now a mature field, thoroughly documented in several monographs; see, e.g., 
\cite{BDF2023a,BDF2025,BDF2023b,BK2013,Bishwal2008,DMR2022,FH2013,FK2012,JKR2017,Kutoyants2004,Prakasa1999} 
and the references therein.
Classical SDE models implicitly assume a homogeneous dynamic shared by all observed units. In many applications, however clinical trials, longitudinal data, and panels of financial time series observations share a common dynamical structure while also displaying substantial variability from one subject, or series, to another.
Stochastic differential equations with random effects (SDEREs) were introduced to bridge this gap by letting random coefficients enter the drift and/or the diffusion coefficient; the approach was pioneered in population pharmacokinetics/pharmacodynamics (PK/PD) and in population dynamics, where it captures subject-specific unobserved heterogeneity while preserving a common underlying mechanism. Asymptotic inference for such continuous-time random-effects models has since developed considerably, moving from parametric diffusion estimation towards non-parametric, high-frequency \cite{CFS2025,CGCS2013,DGCS2013,DGC2016,Prakasa2021} 
and the references therein.

The Cox-Ingersoll-Ross (CIR) diffusion is one of the most classical models for non-negative stochastic dynamics. Introduced in financial mathematics to model short-term interest rates, it has since become a basic building block in stochastic volatility, credit risk, population dynamics, and other applied areas; see, for instance, \cite{CIR1985,F1951}. Its square-root diffusion coefficient guarantees non-negativity and captures state-dependent volatility. Statistical inference for CIR-type processes has generated a large literature, including maximum likelihood, least-squares, and high-frequency methods; see, among others, \cite{Alfonsi2005,BK2012,BK2013,Kutoyants2004,O1998,OR1997}.

Purely continuous diffusive dynamics are, however, not always sufficient: economic, financial, biological, and physical systems may be subject to sudden upward shocks. Jump-diffusion frameworks, which underlie the present work, have accordingly been the subject of sustained development over the last decade; see, e.g., \cite{BBKP2018,BECE2025b,FJKR2023,JRT2016a,SGKOU2002} and the references therein. A closely related and widely used class of models is that of compound Poisson processes, which arise naturally in queuing and insurance theory. Motivated by the fact that both the jump-size distribution and the Poisson intensity are typically unknown in practice, \cite{CCGC2014} construct parametric estimators for intensity functionals, along with an adaptive nonparametric estimator of the jump-size density.

These considerations motivate jump-type extensions of the CIR process. A basic affine jump-diffusion (BAJD), also called a jump-type CIR process, is obtained by superimposing a positive compound Poisson subordinator on the CIR drift-diffusion dynamics. In its fixed effect form, that is, with a common, non-random jump intensity shared by all trajectories, it satisfies
\begin{equation*}
	\dd Y_t=(a-bY_t)\dd t+\sigma\sqrt{Y_t}\dd W_t+\dd J_t,
	\quad t\geq0,
\end{equation*}
where $a\in (0,\infty)$, $b\in \real$, $\sigma \in (0,\infty)$, $W$ is a standard Brownian motion, and $J$ is a positive finite-activity jump process.

The BAJD model was introduced by Duffie and G\^arleanu \cite{DG2001} to describe the dynamics of default intensity, and was subsequently used as a short-rate model in \cite{Filipovic2001,KS2008}. Motivated by these financial applications, its long-time behavior has since been extensively studied: according to the main results of \cite{Keller2011,KS2008}, the BAJD process possesses a unique invariant probability measure, whose distributional properties were later investigated in \cite{JRT2016a,KM2012}. We note that the results of \cite{Keller2011,KM2012} are in fact very general, holding for a large class of affine processes with state space $[0,\infty)$. Existence results and approximations for the transition densities of the BAJD model can be found in \cite{FMS2013,JRT2016a}.

More broadly, the BAJD model belongs to the affine class and is closely connected to continuous-state branching processes with immigration; its ergodic properties, Laplace transforms, and statistical inference have been investigated in works such as \cite{BBKP2018,BBKP2019_Heston,BDF2023a,BDF2025,BDF2023b,JRT2016b,JRT2016a,LM2015}.

The present paper goes in a different direction. Instead of treating jump intensity as a fixed deterministic parameter common across all observations, we introduce population heterogeneity. We consider $n$ independent BAJD trajectories
$
Y^{(1)},\ldots,Y^{(n)},
$
where the $i$-th trajectory has its own jump intensity $C_i$. The random variables $C_1,\ldots,C_n$ are unobserved independent and identically distributed positive random effects with common density $f_{C}$. The other structural parameters are common across the population. This model is natural whenever different individuals share the same diffusion mechanism but differ in the frequency of shocks.
The interpretation is direct. In finance, $C_i$ may represent the jump-arrival intensity of the $i$-th asset, sector or portfolio. In insurance and credit risk, it may measure the frequency of abrupt losses or default-related shocks. In biomedical applications, it may model the rate of impulses, external interventions or abrupt changes affecting a subject. Thus the density $f_C$ describes the distribution of shock frequencies in the population and becomes an object of statistical interest.

The methodology follows a two-stage hybrid scheme, in the spirit of random-effects estimation for stochastic differential equations. First, each latent intensity $C_i$ is treated as an individual parameter and estimated from the corresponding trajectory by maximum likelihood. Second, the resulting estimates are used as pseudo-observations to construct a nonparametric estimator of the common density $f_C$. Since Bernstein estimators are well adapted to compact supports and exhibit good boundary behaviour, we use Bernstein polynomials to estimate the density of the random effects.

A key advantage of the present BAJD setting is that the random effect acts on
the total mass of the jump L\'evy measure. Under continuous-time observation,
the jump times are pathwise observable because the jumps of \(Y^{(i)}\)
coincide with the jumps of \(J^{(i)}\). Consequently,
$
N_T^{(i)}=\sum_{0<t\leq T}\ind_{\{\Delta Y_t^{(i)}>0\}}
$
is observed and, conditionally on \(C_i\), is Poisson with mean \(C_iT\). 
Therefore, as shown in Proposition~\ref{prop:mle-Ci}, the individual maximum
likelihood estimator is simply
$
\widehat C_{i,T}=\frac{N_T^{(i)}}{T}.
$
This explicit form makes the model analytically transparent and creates a
direct bridge between fixed-effect BAJD likelihood theory and nonparametric
density estimation for random effects.

After the general theory has been developed for an arbitrary jump-size density
\(q_\lambda\), we specialize it to the Gamma\((k,\lambda)\) family:
$
q_\lambda^{(k)}(z)
=
\frac{\lambda^k}{\Gamma(k)}
z^{k-1}e^{-\lambda z}\ind_{(0,\infty)}(z),
\; k\geq1,
$
where the classical exponential model corresponds to \(k=1\) and the
Erlang model of order two corresponds to \(k=2\). Importantly, the
estimator of \(C_i\) is unchanged by this specialization, because the likelihood
contribution of \(C_i\) depends only on the jump count.

The paper is organized into three substantive parts.
Section~\ref{sec:general-framework} develops the BAJD random-effect framework
for an arbitrary positive jump-size density \(q_\lambda\). It derives the
individual extended MLE, studies its asymptotic behavior, and introduces the
Bernstein estimator of the random-effect density together with its bias,
variance, MISE and pointwise limit distribution.
Section~\ref{sec:applications} specializes the framework to the
Gamma\((k,\lambda)\) family. It studies the pooled MLE of the common rate
parameter, its joint asymptotic behavior with the individual intensity
estimators, the resulting estimator of the population stationary mean, and
the exponential and Erlang special cases.
Section~\ref{sec:numerical-illustration} contains a controlled simulation
study and a reproducible empirical illustration based on financial
realized-volatility proxies. The empirical analysis includes a sensitivity
study for the jump-detection threshold. Finally, Appendix~\ref{app:proofs-auxiliary-lemmas} collects the proofs of
the four auxiliary lemmas used in the asymptotic analysis, while
Appendix~\ref{app:martingale-gclt} recalls the generalized central limit
theorem for martingales employed in the stable limit arguments.

\section{General BAJD random-effect model with arbitrary jump-size density}
\label{sec:general-framework}

\subsection{BAJD preliminaries}
Let $(\Omega,\calF,(\calF_t)_{t\geq0},\p)$ be a filtered probability space satisfying the usual conditions. We denote convergence in probability, in distribution, stable convergence in distribution and almost surely by $\stackrel{\mathbb{P}}{\longrightarrow}$, $\stackrel{\mathcal D}{\longrightarrow}$, $\xrightarrow[]{\mathcal D,\,\mathrm{st}}$ and $\stackrel{\mathrm{a.s.}}{\longrightarrow}$, respectively, and we write ``$:=$'' for ``is defined to equal''. We write
$\mathbb R_{--}:=(-\infty,0)$, 
$\mathbb R_{-}:=(-\infty,0]$,
$\mathbb R_{++}:=(0,+\infty)$, $\mathbb
R_+:=[0,+\infty)$ and $\mathbb N=\{0,1,2,\ldots\}$. Let $C^2([0,1])$ denote the space of twice continuously differentiable real-valued functions on $[0,1]$, and let $\Pi_{[0,1]}(x)$ denote the projection of $x$ onto $[0,1]$.

A BAJD process is the unique strong solution of
\begin{equation}
	Y_t=Y_0+\int_0^t(a-bY_s)\dd s+\sigma\int_0^t\sqrt{Y_s}\dd W_s+J_t,
	\quad t\geq0,
	\seclabel{eq:bajd-integral}
\end{equation}
where $Y_0=y_0>0$, $a \in \real_+$, $b\in\R$, $\sigma \in \real_{++}$, and $J$ is a compound Poisson subordinator independent of $W$ and $Y_0$.
In the ergodic case where $b \in \mathbb{R}_{++}$, the process admits a unique invariant probability distribution under standard integrability assumptions on the Lévy measure (see \cite[Theorem 2.4(i)]{BBKP2018}). Furthermore, if $a \in \mathbb{R}_{++}$ and suitable conditions on the Lévy measure of $(J_t)_{t \ge 0}$ hold, the process $(Y_t)_{t \ge 0}$ is exponentially ergodic (see \cite[Theorem 2.4(ii)]{BBKP2018}). For further stochastic properties within a more general framework, see \cite{FJKR2023, JRT2016b}.

Throughout the paper, the jump component is assumed to be of finite activity,
with L\'evy measure of the form
$
	m(\dd z)=c q_\lambda(z)\dd z,
	\; z>0,
	\label{eq:generic-levy-fixed}
$
where \(c\in \mathbb R_{++}\) is the total jump intensity and \(q_\lambda\) is a probability
density on $\mathbb R_{++}$ satisfying the finite first-moment condition
\begin{equation}
	\int_0^\infty zq_\lambda(z)\dd z<\infty.
	\label{eq:q-first-moment}
\end{equation}
Hence
$
J_t=\sum_{k=1}^{N_t}Z_k,
$
where $N$ is a Poisson process with intensity $c$ and $(Z_k)_{k\geq1}$ are  independent and identically distributed (i.i.d.) with density $q_\lambda$.
For a test function $f\in C_c^2(\Rplus,\R)$, the infinitesimal generator is
$
	\calL f(y)=(a-by)f'(y)+\frac{\sigma^2}{2}y f''(y)
	+\int_0^\infty \bigl(f(y+z)-f(y)\bigr)cq_\lambda(z)\dd z.
	$
The corresponding immigration mechanism is
$
	F(u)=au+c\int_0^\infty (e^{uz}-1)q_\lambda(z)\dd z,
	\; u\in\R \text{ such that } \int_0^\infty e^{uz}q_\lambda(z)\dd z<\infty,
	\label{eq:immigration-generic}
$
while the branching mechanism is
$
	R(u)=\frac{\sigma^2}{2}u^2-bu.
	\label{eq:branching}
$
Let
$
\mu_1(\lambda):=\int_0^\infty zq_\lambda(z)\dd z.
$
Then, conditionally on a fixed intensity $c$, one has
\begin{equation}
	\E(J_t)=ct\mu_1(\lambda),
	\label{eq:J-mean-generic}
\end{equation}
and, in the ergodic case $b>0$,
\begin{equation}
	\lim_{t\to\infty}\E(Y_t)=\frac{a+c\mu_1(\lambda)}{b}.
	\label{eq:stationary-mean-generic}
	\end{equation}

Being the strong solution of an SDE driven jointly by a Brownian motion and a compound Poisson process, $(Y_t)_{t\in\real_+}$ is a semimartingale on $(\Omega,\calF,(\calF_t)_{t\geq0},\p)$; see e.g. Jacod and Shiryaev \cite[Definition II.2.6 and Remark II.2.8]{JS2003}. Because the jump component has finite activity ($c<\infty$) and $\mu_1(\lambda)<\infty$, no truncation function is needed to compensate for small jumps: the process
$
M_t:=J_t-c\mu_1(\lambda)t,\; t\geq0,
$
is a centered compound Poisson $(\calF_t)$-martingale, and \eqref{eq:bajd-integral} takes the following Grigelionis (semimartingale) form:
\begin{equation}
	Y_t=Y_0+\int_0^t\bigl(a+c\mu_1(\lambda)-bY_u\bigr)\dd u+\sigma\int_0^t\sqrt{Y_u}\dd W_u+M_t,
	\quad t\geq0.
	\label{eq:grigelionis-form}
\end{equation}
This decomposition isolates the local martingale part $\bigl(\sigma\int_0^\cdot\sqrt{Y_u}\dd W_u+M_t\bigr)$ from the finite-variation drift part, and is consistent with \eqref{eq:J-mean-generic}--\eqref{eq:stationary-mean-generic}, obtained by taking expectations in \eqref{eq:grigelionis-form}.

\subsection{BAJD model with a random jump intensity}

We now introduce a population version of the BAJD model. For $i=1,\ldots,n$, the $i$-th individual trajectory is observed continuously on $[0,T]$ and satisfies
\begin{equation}
	\dd Y_t^{(i)}=(a-bY_t^{(i)})\dd t+
	\sigma\sqrt{Y_t^{(i)}}\dd W_t^{(i)}+
	\dd J_t^{(i)},
	\quad Y_0^{(i)}=y_0>0.
	\label{BAJD-Random}
\end{equation}

The parameters $a,b,\sigma,\lambda\in \mathbb R_{++}$ are common structural parameters. The individual heterogeneity enters only through the jump intensity. More precisely,
$
J_t^{(i)}=\sum_{k=1}^{N_t^{(i)}}Z_k^{(i)},
\; t\geq0,
$
where, conditionally on $C_i$, $N^{(i)}$ is a Poisson process with intensity $C_i$, and $(Z_k^{(i)})_{k\geq1}$ are i.i.d. positive jump-sizes with density $q_\lambda$.
The conditional L\'evy measure of the $i$-th trajectory is therefore
$
m_i(\dd z)=C_i q_\lambda(z)\dd z.
\label{eq:random-c-levy}
$
The random variable $C_i$ is the individual jump-frequency effect.

\begin{rema}
\begin{enumerate}
    \item The assumption that all trajectories start from the same deterministic initial
value \(Y_0^{(i)}=y_0>0\) is made only for notational simplicity. The estimation
of the random jump intensity \(C_i\) is based on the jump count
\(N_T^{(i)}\), and conditionally on \(C_i\) one has
$
N_T^{(i)}\sim \operatorname{Poisson}(C_iT),
$
independently of the initial value. Therefore, the consistency and stable
central limit theorem for
$
\widehat C_{i,T}=\frac{N_T^{(i)}}{T}
$
remain unchanged if \(Y_0^{(i)}\) is replaced by an i.i.d. positive random
initial condition, independent of \(C_i\) and of the driving Brownian and jump
components, with a distribution not depending on \(C_i\). We keep the
deterministic notation \(Y_0^{(i)}=y_0\) in order to avoid additional notation.
\item The condition \(b\in \mathbb R_{++}\) is retained throughout the paper in order to work within
the ergodic BAJD framework. It guarantees the existence of a conditional
invariant distribution for the trajectory \(Y^{(i)}\) given \(C_i=c\), and this
invariant distribution depends on the individual jump intensity \(c\). Hence
\(b\in \mathbb R_{++}\) is essential for the stationary interpretation of the model and for
long-run quantities such as the conditional stationary mean and the population
stationary mean.
However, the likelihood estimation of the random intensity \(C_i\) itself does
not use the invariant distribution of \(Y^{(i)}\). For continuous observation,
the jumps of \(Y^{(i)}\) are pathwise observable, and conditionally on \(C_i\)
the counting process \(N^{(i)}\) is Poisson with intensity \(C_i\). Therefore the
likelihood component depending on \(C_i\), the estimator
$
\widehat C_{i,T}=\frac{N_T^{(i)}}{T},
$
and its consistency and stable central limit theorem are driven by the counting
process rather than by the ergodic distribution of \(Y^{(i)}\). The assumption
\(b\in \mathbb R_{++}\) is thus structural for the ergodic interpretation of the BAJD model,
while the first-stage jump-intensity estimation relies on the Poisson jump-count
structure.
\end{enumerate}
\end{rema}
In what follows, we introduce two technical assumptions concerning the
population structure and the support of the density used for Bernstein
estimation.

\begin{assumption}
\label{ass:population}
The random variables \(C_1,\ldots,C_n\) are independent and identically
distributed with common density \(f_C\). They are independent of the Brownian
motions, the jump-size sequences and the Poisson mechanisms. Conditionally on
\(C_1,\ldots,C_n\), the trajectories
\(Y^{(1)},\ldots,Y^{(n)}\) are independent.
\end{assumption}

\begin{assumption}
\label{ass:support}
The common density \(f_C\) of the random effects is supported on \([0,1]\),
and
$
\mathbb P(C_i>0)=1.
$
Moreover, \(f_C\in C^2([0,1])\). In particular, each \(C_i\) is bounded and
integrable.
\end{assumption}
\begin{rema}\label{rem:support-normalization}
The normalization of the support to \([0,1]\) is not restrictive for the
Bernstein step. If the original random effect is supported on a compact interval
\([\alpha,\beta]\subset\mathbb R_{++}\), one may apply the affine transformation
$
U_i=\frac{C_i-\alpha}{\beta-\alpha}\in[0,1].
$
If the support is \(\mathbb R_{++}\), one may instead use the monotone
transformation
$
U_i=\frac{C_i}{1+C_i}\in(0,1).
$
The Bernstein estimator may then be applied to the density of the transformed
variable on \([0,1]\), and the density of the original random effect is recovered
by the corresponding change of variables.
\end{rema}

Since the continuous martingale part of $Y^{(i)}$ (driven by $W^{(i)}$) has no jumps, $Y^{(i)}$ jumps exactly when $J^{(i)}$ does, with $\Delta Y_t^{(i)}=\Delta J_t^{(i)}$. We denote by
$
\mu^{Y^{(i)}}(\dd t,\dd z):=\sum_{k\geq1}\epsilon_{(T_k^{(i)},Z_k^{(i)})}(\dd t,\dd z)
$
the associated jump measure on $\real_+\times \real_{++}$, where $\epsilon_{(t,z)}$ denotes the Dirac measure at $(t,z)$ and $(T_k^{(i)})_{k\geq1}$ are the successive jump times of $N^{(i)}$; equivalently, 
$\mu^{Y^{(i)}}(\dd t,\dd z)=\sum_{s>0}\mathds 1_{\{\Delta Y_s^{(i)}\neq0\}}\epsilon_{(s,\Delta Y_s^{(i)})}(\dd t,\dd z)$. Conditionally on $C_i$, $\mu^{Y^{(i)}}$ has compensator $C_iq_\lambda(z)\dd z\,\dd t$.
For each trajectory, define the observed jump count and cumulative jump-size by
\begin{equation}
N_T^{(i)}:=\int_0^T\int_0^\infty \mu^{Y^{(i)}}(\dd t,\dd z),
\quad
J_T^{(i)}:=\int_0^T\int_0^\infty z\,\mu^{Y^{(i)}}(\dd t,\dd z).
\label{eq:individual-count-size}
\end{equation}

$\Delta Y_t^{(i)}=\Delta J_t^{(i)}\in \mathbb R_{++}$ at jump times, both $N_T^{(i)}$ and $J_T^{(i)}$ are measurable functions of the continuously observed trajectory.

\subsection{Maximum likelihood estimation of the individual random intensities}
In this section, for each fixed $i=1,\ldots,n$, we construct the MLE of
the individual jump intensity $C_i$ from the continuous observation
$(Y_t^{(i)})_{t\in[0,T]}$.The common structural parameters $(a,b,\sigma,\lambda)$ are regarded as
fixed. Conditional on $C_i$, the L\'evy measure of the jump component of
the $i$-th trajectory is
$
    m_i(dz\mid C_i)
    =
    C_i q_\lambda(z)\,dz.
$
Since $q_\lambda$ is a common probability density, the total mass of this
conditional L\'evy measure is
$
    \int_0^\infty m_i(dz\mid C_i)
    =
    C_i\int_0^\infty q_\lambda(z)\,dz
    =
    C_i.
$
Thus, $C_i$ enters the model only through the total mass of the conditional
L\'evy measure, whereas the conditional distribution of the jump-sizes,
given that a jump occurs, remains $q_\lambda$. Consequently, the likelihood
component depending on $C_i$ is entirely determined by the observed jump
count $N_T^{(i)}$.
For the likelihood calculation, we condition on $C_i=c$ and regard $c$ as
an unknown deterministic parameter. The corresponding conditional L\'evy
measure is therefore
\begin{equation}
    m_c(dz)
    :=
    m_i(dz\mid C_i=c)
    =
    c q_\lambda(z)\,dz.
    \label{eq:condition-m}
\end{equation}
For $c\in\real_{++}$, let $\mathbb{P}_c^{(i)}$ denote the probability
measure induced by the BAJD trajectory $Y^{(i)}$ when its individual jump
intensity is equal to $c$, with all common structural parameters kept
fixed. Let $\mathbb{P}_{c,T}^{(i)}$ denote the restriction of this measure
to the observations

time $T$. The following proposition gives the
likelihood ratio between two possible values of the individual jump
intensity.

\begin{prop}
\seclabel{prop:likelihood-ratio-Ci}

Fix $i\in\{1,\ldots,n\}$. Let $a\in\real_{++}$, $b\in\real$,
$\sigma\in\real_{++}$, and $\lambda\in\real_{++}$, and let $q_\lambda$ be
a probability density on $\real_{++}$ satisfying
\eqref{eq:q-first-moment}. For each $c\in\real_{++}$, denote by
$\mathbb{P}_c^{(i)}$ the law of the unique strong solution of
\eqref{BAJD-Random} when the individual random effect is fixed at $C_i=c$,
so that its conditional jump L\'evy measure is $m_c$ as defined in
\eqref{eq:condition-m}. Denote by $\mathbb{P}_{c,T}^{(i)}$ the restriction
of this law to the $\sigma$-algebra generated by
$(Y_t^{(i)})_{0\leq t\leq T}$.
Then, for every $c,\widetilde c\in\real_{++}$ and $T\in\real_{++}$, the
measures $\mathbb{P}_{c,T}^{(i)}$ and
$\mathbb{P}_{\widetilde c,T}^{(i)}$ are mutually absolutely continuous,
and
\begin{equation}
    \log\left(
        \frac{\dd\mathbb{P}_{c,T}^{(i)}}
        {\dd\mathbb{P}_{\widetilde c,T}^{(i)}}
        \bigl(Y^{(i)}\bigr)
    \right)
    =
    N_T^{(i)}
    \log\left(\frac{c}{\widetilde c}\right)
    -
    (c-\widetilde c)T,
    \label{eq:likelihood-ratio-Ci}
\end{equation}
$\mathbb{P}_{\widetilde c,T}^{(i)}$-almost surely.
\end{prop}

\begin{proof}[\bf Proof of Proposition~\ref{prop:likelihood-ratio-Ci}]
Fix \(c,\widetilde c \in \real_{++}\) and use \(\mathbb P_{\widetilde c}^{(i)}\) as the
reference law.
Under this law, the jump measure $\mu^{Y^{(i)}}$ is a Poisson random
measure on $[0,T]\times\real_{++}$ with compensator
$
    \nu_{\widetilde c}^{(i)}(\dd t,\dd z)
    =
    \widetilde c\,q_\lambda(z)\dd z\,\dd t.
$
Its desired intensity under \(\mathbb P_c^{(i)}\) is
$
\nu_c^{(i)}(\dd t,\dd z)
=c q_\lambda(z)\dd z\,\dd t
=\frac{c}{\widetilde c}\,
  \nu_{\widetilde c}^{(i)}(\dd t,\dd z).
$
The likelihood ratio for these two finite-intensity Poisson random measures is
therefore

\begin{align}
L_{i,T}(c,\widetilde c)
&=
\exp\left\{
\int_0^T\int_0^\infty
\log\left(\frac{c}{\widetilde c}\right)
\mu^{Y^{(i)}}(\dd t,\dd z)
-\int_0^T\int_0^\infty
\bigl(\nu_c^{(i)}-\nu_{\widetilde c}^{(i)}\bigr)(\dd t,\dd z)
\right\}
\nonumber\\
&=
\exp\left\{
N_T^{(i)}\log\left(\frac{c}{\widetilde c}\right)
-(c-\widetilde c)T
\right\},
\label{eq:likelihood-ratio-Ci-intermediate}
\end{align}
where we used \eqref{eq:individual-count-size} and
\(\int_0^\infty q_\lambda(z)\dd z=1\).

This change of measure alters only the intensity of the Poisson jump
measure. The initial condition, the Brownian law, the coefficients
\((a,b,\sigma)\), and the conditional jump-size density \(q_\lambda\) are the
same under both measures. By strong existence and pathwise uniqueness, the
solution is a measurable functional of the common initial condition, the
Brownian motion and the jump measure. 
Since  $W^{(i)}$ and $\mu^{Y^{(i)}}$ are independent under both $\mathbb{P}^{(i)}_{c,T}$ and $\mathbb{P}^{(i)}_{\tilde c,T}$, and the law of $W^{(i)}$ is the same under both measures ($a,b,\sigma$ unchanged), the likelihood ratio of the joint law $(W^{(i)},\mu^{Y^{(i)}})$, and hence, by pathwise uniqueness, of $Y^{(i)}$ coincides with the likelihood ratio of $\mu^{Y^{(i)}}$ alone.
Consequently,
\eqref{eq:likelihood-ratio-Ci-intermediate} is also the likelihood ratio of the
induced trajectory laws. Conversely, every jump of \(Y^{(i)}\) is positive and
comes from its compound Poisson component, so \(\mu^{Y^{(i)}}\), and hence the
displayed likelihood, is measurable with respect to the continuously observed
trajectory. This proves \eqref{eq:likelihood-ratio-Ci}.

It remains to verify that the density is normalized. Under the reference
law $\mathbb P_{\widetilde c,T}^{(i)}$ introduced above,
$
    N_T^{(i)}
    \sim
    \operatorname{Poisson}(\widetilde cT).
$
Let $\mathbb E_{\widetilde c}^{(i)}$ denote expectation under
$\mathbb P_{\widetilde c,T}^{(i)}$. The probability-generating function
of the Poisson distribution gives
$
    \mathbb E_{\widetilde c}^{(i)}
    \left[
        L_{i,T}(c,\widetilde c)
    \right]
    =
    e^{-(c-\widetilde c)T}
    \mathbb E_{\widetilde c}^{(i)}
    \left[
        \left(
            \frac{c}{\widetilde c}
        \right)^{N_T^{(i)}}
    \right]
    =
    e^{-(c-\widetilde c)T}
    \exp\left\{
        \widetilde cT
        \left(
            \frac{c}{\widetilde c}-1
        \right)
    \right\}=1.
$
The density is strictly positive. Interchanging \(c\) and \(\widetilde c\)
gives its reciprocal and proves mutual absolute continuity.
\end{proof}

Next, fixing $\widetilde c\in\real_{++}$ as a reference value, define the
log-likelihood contrast
$
\Lambda_{i,T}(c,\widetilde c)
:=
N_T^{(i)}\log\left(\frac{c}{\widetilde c}\right)
-(c-\widetilde c)T.
$
The terms involving $\widetilde c$ do not depend on $c$, so maximizing
$\Lambda_{i,T}(c,\widetilde c)$ with respect to $c$ is equivalent to
maximizing
$
\ell_{i,T}^J(c)
=
N_T^{(i)}\log c-cT,
\; c\in \mathbb R_{++}.
$
In particular, the maximizer is independent of the arbitrary reference
parameter $\widetilde c$, exactly as in the usual likelihood-ratio
construction.
Because the model parameter space is $\real_{++}$, a minor boundary issue
arises on the event $\{N_T^{(i)}=0\}$. It is convenient to consider the
extended likelihood on $\Rplus$, with the conventions
$0\log 0:=0$ and $n\log 0:=-\infty$ for $n\geq1$.

In what follows, we establish the form of the MLE of the individual jump
intensity. For each \(i\in\{1,\ldots,n\}\), define the first jump time of
\(N^{(i)}\) by
$
\tau_1^{(i)}
:=
\inf\left\{
t>0:N_t^{(i)}>0
\right\},
$
with the convention \(\inf\varnothing=\infty\). This notation will be used to
distinguish whether at least one jump has been observed by time \(T\).

\begin{prop}
\seclabel{prop:mle-Ci}
Fix $i\in\{1,\ldots,n\}$ and $T\in\real_{++}$. For every
$c\in\real_{++}$, under $\mathbb P_c^{(i)}$,
$
    N_T^{(i)}
    \sim
    \operatorname{Poisson}(cT).
$
Define the extended conditional log-likelihood based on the observed jump
count by
\begin{equation*}
    \overline{\ell}_{i,T}^J(c)
    :=
    \begin{cases}
        N_T^{(i)}\log(cT)
        -
        cT
        -
        \log\bigl(N_T^{(i)}!\bigr),
        & c>0,\\[1mm]
        0,
        & c=0 \text{ and } N_T^{(i)}=0,\\
        -\infty,
        & c=0 \text{ and } N_T^{(i)}>0.
    \end{cases}
    \label{eq:extended-individual-c-loglik}
\end{equation*}
Then $\overline{\ell}_{i,T}^J$ admits a unique maximizer on $\Rplus$,
given by
\begin{equation}
    \widehat C_{i,T}
    =
    \frac{N_T^{(i)}}{T}.
    \label{eq:Ci-hat}
\end{equation}

On the event
$
    \left\{T\geq\tau_1^{(i)}\right\}
    =
    \left\{N_T^{(i)}>0\right\},
$
one has $\widehat C_{i,T}\in\real_{++}$, and
$\widehat C_{i,T}$ coincides with the ordinary MLE over
$\real_{++}$. Moreover, for every $c\in\real_{++}$,
$
    \mathbb P_c^{(i)}
    \left(
        \tau_1^{(i)}<\infty
    \right)
    =
    1.
$
On the complementary event
$
    \left\{T<\tau_1^{(i)}\right\}
    =
    \left\{N_T^{(i)}=0\right\},
$
the log-likelihood over $\real_{++}$ has no attained maximizer, whereas
the extended log-likelihood on $\Rplus$ has the unique maximizer $0$.
If the parameter space is restricted to $[0,1]$, the extended constrained MLE is
$
    \widehat C_{i,T}^{\mathcal C}
    =
    \min\left\{
        \frac{N_T^{(i)}}{T},\,1
    \right\}.
$
\end{prop}
\begin{proof}[\bf Proof of Proposition~\ref{prop:mle-Ci}]
Fix $i\in\{1,\ldots,n\}$ and $T\in\real_{++}$. Under
$\mathbb P_c^{(i)}$,
$
    N_T^{(i)}
    \sim
    \operatorname{Poisson}(cT).
$
Hence, for every $r\in\mathbb N$,
$
    \mathbb P_c^{(i)}
    \bigl(N_T^{(i)}=r\bigr)
    =
    e^{-cT}\frac{(cT)^r}{r!}.
$
Therefore, for $c\in\real_{++}$, the conditional log-likelihood based on
the observed jump count is
$
    \ell_{i,T}^J(c)
    =
    N_T^{(i)}\log(cT)
    -
    cT
    -
    \log\bigl(N_T^{(i)}!\bigr),
$
which agrees with the first case in the definition of
$\overline{\ell}_{i,T}^J$.
First, suppose that
$
    T\geq\tau_1^{(i)},
$
or, equivalently, that $N_T^{(i)}>0$. For every
$c\in\real_{++}$,
$
    \frac{\partial}{\partial c}
    \ell_{i,T}^J(c)
    =
    \frac{N_T^{(i)}}{c}-T
$
and
$
    \frac{\partial^2}{\partial c^2}
    \ell_{i,T}^J(c)
    =
    -\frac{N_T^{(i)}}{c^2}
    <
    0.
$
Thus, $\ell_{i,T}^J$ is strictly concave on $\real_{++}$. Its unique
critical point is determined by
$
    \frac{N_T^{(i)}}{c}-T=0,
$
and is therefore
$
    c
    =
    \frac{N_T^{(i)}}{T}.
$
Moreover,
$
    \ell_{i,T}^J(c)
    \xrightarrow[c\downarrow0]{}
    -\infty
    \;\text{and}\;
    \ell_{i,T}^J(c)
    \xrightarrow[c\to\infty]{}
    -\infty.
$
Hence,
$
    \widehat C_{i,T}
    =
    \frac{N_T^{(i)}}{T}
$
is the unique global maximizer on $\real_{++}$ and also the unique
maximizer of the extended log-likelihood on $\Rplus$.

Now suppose that
$
    T<\tau_1^{(i)},
$
or, equivalently, that $N_T^{(i)}=0$. Then
$
    \ell_{i,T}^J(c)
    =
    -cT,
    \qquad
    c\in\real_{++}.
$
This function is strictly decreasing, and
$
    \sup_{c>0}\ell_{i,T}^J(c)
    =
    0,
$
but this supremum is not attained on $\real_{++}$. By definition of the
extended log-likelihood,
$
    \overline{\ell}_{i,T}^J(0)
    =
    0.
$
Therefore, $c=0$ is the unique maximizer of
$\overline{\ell}_{i,T}^J$ on $\Rplus$. It follows that
$
    \widehat C_{i,T}
    =
    \frac{N_T^{(i)}}{T}
$
in both cases.

Furthermore, under $\mathbb P_c^{(i)}$, the first jump time has survival
function
$
    \mathbb P_c^{(i)}
    \left(
        \tau_1^{(i)}>t
    \right)
    =
    \mathbb P_c^{(i)}
    \left(
        N_t^{(i)}=0
    \right)
    =
    e^{-ct},
    \;
    t\geq0.
$
Since $c>0$, letting $t\to\infty$ gives
$
    \mathbb P_c^{(i)}
    \left(
        \tau_1^{(i)}=\infty
    \right)
    =
    0.
$
Consequently,
$
    \mathbb P_c^{(i)}
    \left(
        \tau_1^{(i)}<\infty
    \right)
    =
    1.
$
It remains to consider the constrained parameter space $[0,1]$. If
$
    0
    \leq
    \frac{N_T^{(i)}}{T}
    \leq
    1,
$
the unconstrained extended maximizer already belongs to $[0,1]$ and
therefore remains the unique constrained maximizer.

If
$
    \frac{N_T^{(i)}}{T}>1,
$
then, for every $c\in(0,1]$,
$
    \frac{\partial}{\partial c}
    \ell_{i,T}^J(c)
    =
    \frac{N_T^{(i)}}{c}-T
    \geq
    N_T^{(i)}-T
    >
    0.
$
Thus, the log-likelihood is strictly increasing on $(0,1]$, and its
unique constrained maximizer is the endpoint $c=1$. Combining the two
cases yields
$
    \widehat C_{i,T}^{\mathcal C}
    =
    \min\left\{
        \frac{N_T^{(i)}}{T},\,1
    \right\}.
$
This completes the proof.
\end{proof}
The following lemma establishes the stable limit theorem for a Poisson
process with random intensity that is used in the proof of
Theorem~\ref{consistency-normality}. Its proof is deferred to
Appendix~\ref{app:proofs-auxiliary-lemmas}.
Recall that,
under the BAJD-Random model, conditionally on $C_i$, the counting process
$N^{(i)}=(N_t^{(i)})_{t\geq0}$ is a Poisson process with intensity $C_i$.

\begin{lemm}
\seclabel{lem:stable-poisson-random-intensity}

Assume Assumptions~\ref{ass:population} and~\ref{ass:support}. Fix
$i\in\{1,\ldots,n\}$ and set
$\mathcal G_i:=\sigma(C_i).$
Then, stably with respect to $\mathcal G_i$,
\begin{equation}
    \frac{N_T^{(i)}-C_iT}{\sqrt T}
    \xrightarrow[T\to\infty]{\mathcal D,\mathrm{st}}
    \sqrt{C_i}\,Z_i,
    \label{eq:stable-poisson-random-intensity}
\end{equation}
where $Z_i\sim\mathcal N(0,1)$ is defined on an extension of the original
probability space and is independent of $\mathcal G_i$.
\end{lemm}
The following theorem establishes the strong consistency and the mixed-normal
asymptotic behavior of the estimator \(\widehat C_{i,T}\).

\begin{Theo}
\seclabel{consistency-normality}
\label{thm:Ci-mle-an}

Under Assumptions~\ref{ass:population} and~\ref{ass:support}, for each
fixed $i\in\{1,\ldots,n\}$,
\begin{equation}
    \widehat C_{i,T}
    \xrightarrow[T\to\infty]{\mathrm{a.s.}}
    C_i.
    \label{eq:Ci-consistency}
\end{equation}
Moreover, stably with respect to $\sigma(C_i)$,
\begin{equation}
    \sqrt T\bigl(\widehat C_{i,T}-C_i\bigr)
    \xrightarrow[T\to\infty]{\mathcal D,\mathrm{st}}
    \sqrt{C_i}\,Z_i,
    \label{eq:Ci-mixed-normal}
\end{equation}
where $Z_i\sim\mathcal N(0,1)$ is defined on an extension of the original
probability space and is independent of $\sigma(C_i)$.
Furthermore, for every $c\in(0,1]$, under the law
$\mathbb P_c^{(i)}$ defined in
Proposition~\ref{prop:likelihood-ratio-Ci},
\begin{equation}
    \sqrt T\bigl(\widehat C_{i,T}-c\bigr)
    \xrightarrow[T\to\infty]{\mathcal D}
    \mathcal N(0,c).
    \label{eq:Ci-clt-conditional}
\end{equation}
\end{Theo}

\begin{rema}
The stable central limit theorem in
Theorem~\ref{thm:Ci-mle-an} also yields a simple asymptotic confidence
interval for each jump intensity. Define
\[
    S_{i,T}
    :=
    \begin{cases}
        \displaystyle
        \frac{
            \sqrt T\bigl(\widehat C_{i,T}-C_i\bigr)
        }{
            \sqrt{\widehat C_{i,T}}
        },
        & \widehat C_{i,T}>0,\\
        0,
        & \widehat C_{i,T}=0.
    \end{cases}
\]
The ratio in the first case is well defined on the event
$
    \left\{T\geq\tau_1^{(i)}\right\}
    =
    \left\{N_T^{(i)}>0\right\}
    =
    \left\{\widehat C_{i,T}>0\right\}.
$
By Assumption~\ref{ass:support}, $C_i>0$ almost surely. Conditionally on
$C_i$, the first jump time $\tau_1^{(i)}$ is therefore finite almost
surely. Hence, almost surely, $S_{i,T}$ eventually coincides with the
ratio in the first case of its definition.
Since
$
    \widehat C_{i,T}
    \xrightarrow[T\to\infty]{\mathrm{a.s.}}
    C_i,
$
the stable version of Slutsky's theorem gives, stably with respect to
$\sigma(C_i)$,
$
    S_{i,T}
    \xrightarrow[T\to\infty]{\mathcal D,\mathrm{st}}
    Z_i,
    \;
    Z_i\sim\mathcal N(0,1).
$
Therefore, on
$\{T\geq\tau_1^{(i)}\}=\{N_T^{(i)}>0\}$, an asymptotic confidence interval
for $C_i$ with level $1-\alpha$ is
$
    \left[
        \widehat C_{i,T}
        -
        z_{1-\alpha/2}
        \sqrt{\frac{\widehat C_{i,T}}{T}},
        \;
        \widehat C_{i,T}
        +
        z_{1-\alpha/2}
        \sqrt{\frac{\widehat C_{i,T}}{T}}
    \right],
$
where $z_{1-\alpha/2}$ denotes the $(1-\alpha/2)$-quantile of the standard
normal distribution. More precisely, for every $c\in(0,1]$, under the law
$\mathbb P_c^{(i)}$,
$
    \lim_{T\to\infty}
    \mathbb P_c^{(i)}
    \left(
        c\in
        \left[
            \widehat C_{i,T}
            -
            z_{1-\alpha/2}
            \sqrt{\frac{\widehat C_{i,T}}{T}},
            \;
            \widehat C_{i,T}
            +
            z_{1-\alpha/2}
            \sqrt{\frac{\widehat C_{i,T}}{T}}
        \right]
    \right)
    =
    1-\alpha.
$
Under Assumption~\ref{ass:support}, since $C_i\in[0,1]$ almost surely, one
may report the support-truncated interval
$
    \left[
        \widehat C_{i,T}
        -
        z_{1-\alpha/2}
        \sqrt{\frac{\widehat C_{i,T}}{T}},
        \;
        \widehat C_{i,T}
        +
        z_{1-\alpha/2}
        \sqrt{\frac{\widehat C_{i,T}}{T}}
    \right]
    \cap[0,1].
$
Intersecting the interval with $[0,1]$ does not alter its asymptotic
coverage. This Wald-type interval is intended for large observation
horizons. For rare jumps or small observed counts, an exact Poisson
confidence interval may provide more reliable finite-sample coverage.

\end{rema}

\begin{proof}[\bf Proof of Theorem~\ref{consistency-normality}]

Fix $i\in\{1,\ldots,n\}$.
On the one hand, by Proposition~\ref{prop:mle-Ci},
$
    \widehat C_{i,T}
    =
    \frac{N_T^{(i)}}{T}.
$
Conditionally on $C_i$, the strong law of large numbers for Poisson
processes gives
$
    \mathbb P\left(
        \left.
        \lim_{T\to\infty}\frac{N_T^{(i)}}{T}
        =
        C_i
        \,\right|\,
        C_i
    \right)
    =
    1
    \;\text{a.s.}
$
Taking expectations and using the tower property, we obtain
$$
    \mathbb P\left(
        \lim_{T\to\infty}\frac{N_T^{(i)}}{T}
        =
        C_i
    \right)
    =
    \mathbb E\left[
        \mathbb P\left(
            \left.
            \lim_{T\to\infty}\frac{N_T^{(i)}}{T}
            =
            C_i
            \,\right|\,
            C_i
        \right)
    \right]
    =1.
$$
Therefore,
$
    \widehat C_{i,T}
    \xrightarrow[T\to\infty]{\mathrm{a.s.}}
    C_i,
$
which proves \eqref{eq:Ci-consistency}.
On the other hand, by \eqref{eq:Ci-hat},
$
    \sqrt T\bigl(\widehat C_{i,T}-C_i\bigr)
    =
    \frac{N_T^{(i)}-C_iT}{\sqrt T}.
$
Lemma~\ref{lem:stable-poisson-random-intensity} therefore yields, stably
with respect to $\sigma(C_i)$,
$
    \sqrt T\bigl(\widehat C_{i,T}-C_i\bigr)
    \xrightarrow[T\to\infty]{\mathcal D,\mathrm{st}}
    \sqrt{C_i}\,Z_i,
$
where $Z_i\sim\mathcal N(0,1)$ is independent of $\sigma(C_i)$. This
proves \eqref{eq:Ci-mixed-normal}.

Now, fix $c\in(0,1]$. Under the law $\mathbb P_c^{(i)}$ introduced in
Proposition~\ref{prop:likelihood-ratio-Ci},
$
    N_T^{(i)}
    \sim
    \operatorname{Poisson}(cT).
$
Let $\mathbb E_c^{(i)}$ denote expectation under
$\mathbb P_c^{(i)}$. Then, for every $u\in\mathbb R$,
\begin{align*}
    \mathbb E_c^{(i)}
    \left[
        \exp\left\{
            iu\sqrt T
            \bigl(\widehat C_{i,T}-c\bigr)
        \right\}
    \right]
    &=
    \exp\left\{
        cT\left(
            e^{iu/\sqrt T}
            -1
            -\frac{iu}{\sqrt T}
        \right)
    \right\}.
\end{align*}
Since
$
    e^{iu/\sqrt T}
    =
    1+\frac{iu}{\sqrt T}
    -\frac{u^2}{2T}
    +o(T^{-1}),
$
we have
$
    cT\left(
        e^{iu/\sqrt T}
        -1
        -\frac{iu}{\sqrt T}
    \right)
    \longrightarrow
    -\frac{cu^2}{2},\text{as } T\to\infty.
$
Consequently,
\[
    \mathbb E_c^{(i)}\left[
        \exp\left\{
            iu\sqrt T
            \bigl(\widehat C_{i,T}-c\bigr)
        \right\}
    \right]
    \longrightarrow
    \exp\left(-\frac{cu^2}{2}\right),
\; \text{as } T\to\infty. \]
The limiting function is the characteristic function of
$\mathcal N(0,c)$. Hence, by L\'evy's continuity theorem,
\[
    \sqrt T\bigl(\widehat C_{i,T}-c\bigr)
    \xrightarrow[T\to\infty]{\mathcal D}
    \mathcal N(0,c)
    \;\text{under }\mathbb P_c^{(i)}.
\]
This proves \eqref{eq:Ci-clt-conditional} and completes the proof.
\end{proof}
\subsection{Bernstein estimator of the random-effect density}
\label{sec:density-bernstein}
Under Assumption~\ref{ass:support}, the random effects are supported on
the canonical Bernstein interval $[0,1]$. Therefore, no preliminary
normalization is required. The Bernstein estimator can be applied directly
to the random intensities $C_i$ and, in the feasible construction, to
their truncated estimators $\widehat C_{i,T}^{\calC}$. If the true random effects $C_1,\ldots,C_n$ were observed, the ideal
Bernstein density estimator of order $m$ would be
\begin{equation}
	\widetilde f_{C,m,n}(c)
	=
	m\sum_{k=0}^{m-1}
	\left[
	F_{C,n}\!\left(\frac{k+1}{m}\right)
	-
	F_{C,n}\!\left(\frac{k}{m}\right)
	\right]
	p_k(m-1,c),
	\; c\in[0,1],
	\label{eq:ideal-bernstein}
\end{equation}
where
$
F_{C,n}(c)
=
\frac{1}{n}
\sum_{i=1}^n
\ind_{\{C_i\leq c\}},
\;
p_k(m,c)
=
\binom{m}{k}c^k(1-c)^{m-k}.
$
Since the random effects are unobserved, we replace $C_i$ by
$\widehat C_{i,T}^{\calC}$ and define the feasible Bernstein density
estimator by
\begin{equation}
	\widehat f_{C,m,n,T}(c)
	=
	m\sum_{k=0}^{m-1}
	\left[
	\widehat F_{C,n,T}\!\left(\frac{k+1}{m}\right)
	-
	\widehat F_{C,n,T}\!\left(\frac{k}{m}\right)
	\right]
	p_k(m-1,c),
	\; c\in[0,1],
	\label{eq:fC-bernstein-estimator}
\end{equation}
where
$
\widehat F_{C,n,T}(c)
=
\frac{1}{n}
\sum_{i=1}^n
\ind_{\{\widehat C_{i,T}^{\calC}\leq c\}},
\; c\in[0,1].
$
At a finite observation horizon, the projection
$\widehat C_{i,T}^{\mathcal C}=\min\{\widehat C_{i,T},1\}$ may create a
point mass at the upper endpoint $1$. This may induce a finite-horizon
boundary effect in the feasible Bernstein estimator. Observations
projected onto $1$ are included in the last subinterval of the Bernstein
partition, consistently with the semi-open partition used in
Subsection~\ref{subsec:numerical-design}. This phenomenon has no effect
on the asymptotic results: since $f_C$ is a density on $[0,1]$,
$\mathbb P(C_i=1)=0$, and the almost sure convergence
$\widehat C_{i,T}\to C_i$ implies that, for fixed $n$, the upper
truncation is eventually inactive almost surely.

In what follows, we study the asymptotic properties of the feasible
Bernstein density estimator under a sequential two-stage regime. In the
first stage, for fixed $n$ and $m$, the observation horizon $T$ tends to
infinity, so that the estimated random effects can be replaced
asymptotically by the true random effects. In the second stage,
$n\to\infty$ and $m=m_n\to\infty$, and the ideal Bernstein estimator is
used to recover the random-effect density. The following theorem
establishes the first-stage reduction.

\begin{Theo}
\seclabel{consistency-g-tilde}
\label{thm:first-stage-reduction}

Under Assumptions~\ref{ass:population} and~\ref{ass:support}, for every
fixed $n\geq1$,
\begin{equation}
    \max_{1\leq i\leq n}
    \left|
        \widehat C_{i,T}^{\calC}-C_i
    \right|
    \xrightarrow[T\to\infty]{\mathrm{a.s.}}
    0.
    \label{eq:max-first-stage}
\end{equation}
Moreover, on an event of probability one, for every continuity point
$u\in[0,1]$ of the empirical distribution function $F_{C,n}$,
\begin{equation}
    \widehat F_{C,n,T}(u)
    \xrightarrow[T\to\infty]{}
    F_{C,n}(u).
    \label{eq:cdf-first-stage}
\end{equation}
Finally, for every fixed $m,n\geq1$,
\begin{equation}
    \sup_{u\in[0,1]}
    \left|
        \widehat f_{C,m,n,T}(u)
        -
        \widetilde f_{C,m,n}(u)
    \right|
    \xrightarrow[T\to\infty]{\mathrm{a.s.}}
    0.
    \label{eq:bernstein-first-stage}
\end{equation}
In particular, for every $u\in[0,1]$,
$
    \widehat f_{C,m,n,T}(u)
    \xrightarrow[T\to\infty]{\mathrm{a.s.}}
    \widetilde f_{C,m,n}(u).
$
\end{Theo}

\begin{proof}[\bf Proof of Theorem~\ref{thm:first-stage-reduction}]
The proof proceeds in three steps.
\medskip
\noindent
\textbf{Step 1:}
We first prove the uniform first-stage consistency for fixed $n$.
Fix $n\geq1$. By Proposition~\ref{prop:mle-Ci},
$
    \widehat C_{i,T}^{\calC}
    =
    \min\left\{
        \widehat C_{i,T},1
    \right\},
    \; i=1,\ldots,n.
$
By Theorem~\ref{thm:Ci-mle-an},
$
    \widehat C_{i,T}
    \xrightarrow[T\to\infty]{\mathrm{a.s.}}
    C_i
$
for every fixed $i\in\{1,\ldots,n\}$. Since $C_i\in[0,1]$ almost surely,
$
    \left|
        \min\left\{\widehat C_{i,T},1\right\}-C_i
    \right|
    \leq
    \left|
        \widehat C_{i,T}-C_i
    \right|.
$
Consequently,
\begin{equation}
    \widehat C_{i,T}^{\calC}
    \xrightarrow[T\to\infty]{\mathrm{a.s.}}
    C_i,
    \; i=1,\ldots,n.
    \label{eq:Ci-proj-consistency}
\end{equation}

For every $i\in\{1,\ldots,n\}$, let $\Omega_i$ be an event of probability
one on which the convergence in \eqref{eq:Ci-proj-consistency} holds, and
set
$
    \Omega_n
    :=
    \bigcap_{i=1}^n\Omega_i.
$
Since $n$ is fixed and finite,
$
    \mathbb P(\Omega_n)=1.
$
On $\Omega_n$, all $n$ sequences converge simultaneously. Therefore,
$
    \max_{1\leq i\leq n}
    \left|
        \widehat C_{i,T}^{\calC}-C_i
    \right|
    \xrightarrow[T\to\infty]{}
    0.
$
This proves \eqref{eq:max-first-stage}.

\medskip
\noindent
\textbf{Step 2:}
We next establish the convergence of the empirical distribution function.
Fix $\omega\in\Omega_n$, and let $u\in[0,1]$ be a continuity point of
$F_{C,n}$. Fix $\omega\in\Omega_n$, and let $u\in[0,1]$ be a continuity point of
$F_{C,n}$. Since
$
    F_{C,n}(u)
    =
    \frac{1}{n}
    \sum_{i=1}^n
    \ind_{\{C_i\leq u\}},
$
continuity at $u$ is equivalent to
$
    u\notin\{C_1,\ldots,C_n\}.
$
For each $i\in\{1,\ldots,n\}$, either $C_i<u$ or $C_i>u$. If $C_i<u$,
the convergence $\widehat C_{i,T}^{\calC}\to C_i$ implies that
$\widehat C_{i,T}^{\calC}<u$ for all sufficiently large $T$. Similarly,
if $C_i>u$, then $\widehat C_{i,T}^{\calC}>u$ for all sufficiently large
$T$.
Thus, for every $i\in\{1,\ldots,n\}$, there exists $T_i<\infty$ such that
$
    \ind_{\{\widehat C_{i,T}^{\calC}\leq u\}}
    =
    \ind_{\{C_i\leq u\}},
    \; T\geq T_i.
$
Since $n$ is finite, there exists $T_0<\infty$ such that the preceding
equality holds simultaneously for every $i\in\{1,\ldots,n\}$ and every
$T\geq T_0$. Consequently, for every $T\geq T_0$,
$
    \widehat F_{C,n,T}(u)
    =
    \frac{1}{n}
    \sum_{i=1}^n
    \ind_{\{\widehat C_{i,T}^{\calC}\leq u\}}
    =
    \frac{1}{n}
    \sum_{i=1}^n
    \ind_{\{C_i\leq u\}}
    =
    F_{C,n}(u).
$
Therefore, on $\Omega_n$, the equality is eventually exact at every
continuity point $u$ of $F_{C,n}$. This proves
\eqref{eq:cdf-first-stage}.

\medskip
\noindent
\textbf{Step 3:}
Finally, we prove the uniform convergence of the feasible Bernstein
estimator for fixed $m$ and $n$. Fix $m,n\geq1$. By Assumptions~\ref{ass:population}
and~\ref{ass:support}, each $C_i$ has density $f_C$ on $[0,1]$. Its
distribution is therefore non-atomic, so that
$
    \mathbb P(C_i=x)=0
$
for every fixed $x\in[0,1]$. In particular, for each
$
    (i,k)
    \in
    \{1,\ldots,n\}\times\{0,\ldots,m\},
$
one has
$
    \mathbb P\left(C_i=\frac{k}{m}\right)=0.
$
Since there are only finitely many such pairs, the event
$
    \Omega_{n,m}
    :=
    \bigcap_{i=1}^n
    \bigcap_{k=0}^m
    \left\{
        C_i\neq\frac{k}{m}
    \right\}
$
satisfies
$
    \mathbb P(\Omega_{n,m})=1.
$
On $\Omega_n\cap\Omega_{n,m}$, every grid point $k/m$,
$k=0,\ldots,m$, is a continuity point of $F_{C,n}$. By Step~2,
$
    \widehat F_{C,n,T}\left(\frac{k}{m}\right)
    \xrightarrow[T\to\infty]{}
    F_{C,n}\left(\frac{k}{m}\right),
    \; k=0,\ldots,m.
$
For $k=0,\ldots,m$, define
$
    \Delta_{k,T}
    :=
    \widehat F_{C,n,T}\left(\frac{k}{m}\right)
    -
    F_{C,n}\left(\frac{k}{m}\right).
$
Since the number of grid points is finite,
$
    \max_{0\leq k\leq m}
    |\Delta_{k,T}|
    \xrightarrow[T\to\infty]{}
    0.
$
Using \eqref{eq:fC-bernstein-estimator} and
\eqref{eq:ideal-bernstein}, we obtain, for every $u\in[0,1]$,
$
    \widehat f_{C,m,n,T}(u)
    -
    \widetilde f_{C,m,n}(u)
    =
    m\sum_{k=0}^{m-1}
    \left(
        \Delta_{k+1,T}-\Delta_{k,T}
    \right)
    p_k(m-1,u).
$
Since
$
    p_k(m-1,u)\geq0
    \;\text{and}\;
    \sum_{k=0}^{m-1}p_k(m-1,u)=1,
$
it follows that
$
\sup_{u\in[0,1]}
    \left|
        \widehat f_{C,m,n,T}(u)
        -
        \widetilde f_{C,m,n}(u)
    \right|
    \leq
    2m
    \max_{0\leq k\leq m}
    |\Delta_{k,T}|
    \xrightarrow[T\to\infty]{}
    0.
$
This proves \eqref{eq:bernstein-first-stage} and completes the proof.
\end{proof}
The following two lemmas provide the technical ingredients needed for the
proof of Theorem~\ref{Bias-variance-MISE}. The first establishes a
Voronovskaya-type bias expansion for the ideal Bernstein density estimator,
whereas the second establishes its variance expansion. Their proofs are
deferred to Appendix~\ref{app:proofs-auxiliary-lemmas}.
\begin{lemm}
\label{lem:ideal-bernstein-bias}
Under Assumption~\ref{ass:population}, assume that the common density
$f_C$ is supported on $[0,1]$ and belongs to $C^2([0,1])$. Define
\[
    \widetilde f_{C,m}(u)
    :=
    \mathbb E\bigl[\widetilde f_{C,m,n}(u)\bigr]
    =
    m\sum_{k=0}^{m-1}
    \left[
        F_C\left(\frac{k+1}{m}\right)
        -
        F_C\left(\frac{k}{m}\right)
    \right]
    p_k(m-1,u),
\]
where $F_C$ is the distribution function associated with $f_C$. Then
\begin{equation}
    \sup_{u\in[0,1]}
    \left|
        m\bigl(\widetilde f_{C,m}(u)-f_C(u)\bigr)
        -
        \frac{1}{2}
        \left[
            (1-2u)f_C'(u)
            +
            u(1-u)f_C''(u)
        \right]
    \right|
    \xrightarrow[m\to\infty]{}
    0.
    \label{eq:lemma-uniform-bias-expansion}
\end{equation}
In particular, for every $u\in[0,1]$,
$
    \widetilde f_{C,m}(u)-f_C(u)
    =
    \frac{1}{2m}
    \left[
        (1-2u)f_C'(u)
        +
        u(1-u)f_C''(u)
    \right]
    +
    o(m^{-1}).
$
Moreover,
\begin{equation}
    \int_0^1
    \bigl(\widetilde f_{C,m}(u)-f_C(u)\bigr)^2
    \dd u
    =
    B_{f_C}m^{-2}
    +
    o(m^{-2}),
    \label{eq:lemma-integrated-bias}
\end{equation}
where
$
    B_{f_C}
    :=
    \frac{1}{4}
    \int_0^1
    \left[
        (1-2u)f_C'(u)
        +
        u(1-u)f_C''(u)
    \right]^2
    \dd u.
$
\end{lemm}

\begin{lemm}
\label{lem:ideal-bernstein-variance}

Under Assumption~\ref{ass:population}, assume that the common density
$f_C$ is supported on $[0,1]$ and belongs to $C([0,1])$. Then, for every
fixed $u\in(0,1)$, as $m\to\infty$,
\begin{equation}
    \operatorname{Var}
    \bigl(\widetilde f_{C,m,n}(u)\bigr)
    =
    \frac{m^{1/2}}{n}
    f_C(u)\psi(u)
    +
    o\left(
        \frac{m^{1/2}}{n}
    \right),
    \label{eq:lemma-variance-expansion}
\end{equation}
where
$
    \psi(u)
    :=
    \frac{1}{\sqrt{4\pi u(1-u)}}.
$
Moreover,
\begin{equation}
    \int_0^1
    \operatorname{Var}
    \bigl(\widetilde f_{C,m,n}(u)\bigr)
    \dd u
    =
    \frac{m^{1/2}}{n}
    A_{f_C}
    +
    o\left(
        \frac{m^{1/2}}{n}
    \right),
    \label{eq:lemma-integrated-variance}
\end{equation}
where
$
    A_{f_C}
    :=
    \int_0^1
    f_C(u)\psi(u)\dd u
    <
    \infty.
$
The remainder terms in
\eqref{eq:lemma-variance-expansion} and
\eqref{eq:lemma-integrated-variance} are uniform with respect to
$n\geq1$.
\end{lemm}

The following theorem uses the asymptotic equivalence between the feasible
and ideal Bernstein estimators, together with the bias and variance
expansions of the ideal estimator. It establishes
the asymptotic bias, variance, and mean integrated squared error of the
feasible estimator under the following sequential asymptotic regime:
$T\to\infty$ first, for fixed $n$ and $m$, and then $n\to\infty$ with
$m=m_n\to\infty$.
\begin{Theo}
\seclabel{Bias-variance-MISE}
\label{thm:bias-var-mise}
Under Assumptions~\ref{ass:population} and~\ref{ass:support}, for every
fixed $u\in(0,1)$,
\begin{align}
    \lim_{T\to\infty}
    \operatorname{Bias}
    \bigl(\widehat f_{C,m,n,T}(u)\bigr)
    &=
    \frac{1}{2m}
    \left[
        (1-2u)f_C'(u)
        +
        u(1-u)f_C''(u)
    \right]
    +
    o(m^{-1}),
    \label{eq:bias-g}
    \\[1mm]
    \lim_{T\to\infty}
    \operatorname{Var}
    \bigl(\widehat f_{C,m,n,T}(u)\bigr)
    &=
    \frac{m^{1/2}}{n}
    f_C(u)\psi(u)
    +
    o\left(
        \frac{m^{1/2}}{n}
    \right),
    \label{eq:var-g}
\end{align}
where
$
    \psi(u)
    :=
    \frac{1}{\sqrt{4\pi u(1-u)}}.
$
The remainder terms in \eqref{eq:bias-g} and \eqref{eq:var-g} are
understood in the second stage, as $n\to\infty$ and
$m=m_n\to\infty$.
Moreover,
\begin{equation}
    \lim_{T\to\infty}
    \operatorname{MISE}
    \bigl(\widehat f_{C,m,n,T}\bigr)
    =
    A_{f_C}\frac{m^{1/2}}{n}
    +
    B_{f_C}m^{-2}
    +
    o\left(\frac{m^{1/2}}{n}\right)
    +
    o(m^{-2}),
    \label{eq:mise-g}
\end{equation}
where
$
    A_{f_C}
    :=
    \int_0^1f_C(u)\psi(u)\dd u
$
and
$
    B_{f_C}
    :=
    \frac{1}{4}
    \int_0^1
    \left[
        (1-2u)f_C'(u)
        +
        u(1-u)f_C''(u)
    \right]^2
    \dd u.
$
In particular, if
\begin{equation}
   m_n\xrightarrow[n\to\infty]{}\infty,
    \;
    \frac{\sqrt{m_n}}{n}
    \xrightarrow[n\to\infty]{}
    0,
    \label{eq:mise-consistency-regime}
\end{equation}
then
\begin{equation}
    \lim_{n\to\infty}
    \left[
        \lim_{T\to\infty}
        \operatorname{MISE}
        \bigl(\widehat f_{C,m_n,n,T}\bigr)
    \right]
    =
    0.
    \label{eq:mise-consistency}
\end{equation}
If $A_{f_C}>0$ and $B_{f_C}>0$, balancing the leading variance and
squared-bias terms in \eqref{eq:mise-g} shows that the asymptotically
optimal Bernstein order is of order $n^{2/5}$ and that the corresponding
MISE is of order $n^{-4/5}$.
\end{Theo}
\begin{proof}[\bf Proof of Theorem~\ref{thm:bias-var-mise}]

For any random density estimator $g$, write
$
    \operatorname{Bias}\bigl(g(u)\bigr)
    :=
    \mathbb E\bigl[g(u)\bigr]-f_C(u)
$
and
$
    \operatorname{MISE}(g)
    :=
    \int_0^1
    \mathbb E\left[
        \bigl(g(u)-f_C(u)\bigr)^2
    \right]
    \dd u.
$
We first reduce the analysis of the feasible estimator to that of the
ideal estimator. Fix $m,n\geq1$. Since
$
    p_k(m-1,u)\in[0,1],
    \;
    \sum_{k=0}^{m-1}p_k(m-1,u)=1,
$
and every increment
$
    \widehat F_{C,n,T}\left(\frac{k+1}{m}\right)
    -
    \widehat F_{C,n,T}\left(\frac{k}{m}\right)
$
belongs to $[0,1]$, one has
$
    0
    \leq
    \widehat f_{C,m,n,T}(u)
    \leq
    m,
    \; u\in[0,1].
$
The same bound holds for $\widetilde f_{C,m,n}(u)$. By
Theorem~\ref{thm:first-stage-reduction},
$
    \widehat f_{C,m,n,T}(u)
    \xrightarrow[T\to\infty]{\mathrm{a.s.}}
    \widetilde f_{C,m,n}(u)
$
for fixed $m$ and $n$, uniformly in $u\in[0,1]$.
Bounded convergence, applied to the first and second moments, gives
$
    \lim_{T\to\infty}
    \operatorname{Bias}
    \bigl(\widehat f_{C,m,n,T}(u)\bigr)
    =
    \operatorname{Bias}
    \bigl(\widetilde f_{C,m,n}(u)\bigr)
$
and
$
    \lim_{T\to\infty}
    \operatorname{Var}
    \bigl(\widehat f_{C,m,n,T}(u)\bigr)
    =
    \operatorname{Var}
    \bigl(\widetilde f_{C,m,n}(u)\bigr).
$
Moreover,
$    \left|
        \widehat f_{C,m,n,T}(u)-f_C(u)
    \right|^2
    \leq
    \left(
        m+\lVert f_C\rVert_\infty
    \right)^2.
$
Therefore, the dominated convergence theorem on
$\Omega\times[0,1]$ yields
$
    \lim_{T\to\infty}
    \operatorname{MISE}
    \bigl(\widehat f_{C,m,n,T}\bigr)
    =
    \operatorname{MISE}
    \bigl(\widetilde f_{C,m,n}\bigr).
$
It is thus sufficient to study the ideal estimator constructed from the
unobserved i.i.d. sample $C_1,\ldots,C_n$.
By Lemma~\ref{lem:ideal-bernstein-bias}, for every fixed
$u\in(0,1)$,
$
    \operatorname{Bias}
    \bigl(\widetilde f_{C,m,n}(u)\bigr)
    =
    \frac{1}{2m}
    \left[
        (1-2u)f_C'(u)
        +
        u(1-u)f_C''(u)
    \right]
    +
    o(m^{-1}).
$
Combining this expansion with the first-stage convergence of the
expectations established above proves \eqref{eq:bias-g}. Similarly,
Lemma~\ref{lem:ideal-bernstein-variance} gives
$
    \operatorname{Var}
    \bigl(\widetilde f_{C,m,n}(u)\bigr)
    =
    \frac{m^{1/2}}{n}
    f_C(u)\psi(u)
    +
    o\left(\frac{m^{1/2}}{n}\right),
$
which proves \eqref{eq:var-g}.
We next derive the integrated expansion. By the bias--variance
decomposition,
\begin{align*}
    \operatorname{MISE}
    \bigl(\widetilde f_{C,m,n}\bigr)
    &=
    \int_0^1
    \operatorname{Bias}
    \bigl(\widetilde f_{C,m,n}(u)\bigr)^2
    \dd u
    +
    \int_0^1
    \operatorname{Var}
    \bigl(\widetilde f_{C,m,n}(u)\bigr)
    \dd u.
\end{align*}
Using the integrated conclusions of
Lemmas~\ref{lem:ideal-bernstein-bias}
and~\ref{lem:ideal-bernstein-variance}, we obtain
$
    \operatorname{MISE}
    \bigl(\widetilde f_{C,m,n}\bigr)
    =
    A_{f_C}\frac{m^{1/2}}{n}
    +
    B_{f_C}m^{-2}
    +
    o\left(\frac{m^{1/2}}{n}\right)
    +
    o(m^{-2}).
$
Combining this expansion with the first-stage MISE convergence established
above proves \eqref{eq:mise-g}.

Under \eqref{eq:mise-consistency-regime},
$
    m_n^{-2}
    \xrightarrow[n\to\infty]{}
    0,
    \;
    \frac{\sqrt{m_n}}{n}
    \xrightarrow[n\to\infty]{}
    0.
$
The corresponding remainder terms tend to zero as well. Therefore,
$
    \lim_{n\to\infty}
    \lim_{T\to\infty}
    \operatorname{MISE}
    \bigl(\widehat f_{C,m,n,T}\bigr)
    =
    0,
$
which proves \eqref{eq:mise-consistency}.
It remains to minimize the leading MISE term
$
    h_n(m)
    :=
    A_{f_C}\frac{m^{1/2}}{n}
    +
    B_{f_C}m^{-2}.
$
Assume that $A_{f_C}>0$ and $B_{f_C}>0$. For fixed $n$, consider the
continuous relaxation of $h_n$ over $m>0$. Its derivative is
$
    h_n'(m)
    =
    \frac{A_{f_C}}{2n}m^{-1/2}
    -
    2B_{f_C}m^{-3}.
$
Thus, $h_n'(m)=0$ if and only if
$
    m^{5/2}
    =
    \frac{4B_{f_C}}{A_{f_C}}\,n.
$
The continuous relaxation therefore attains its unique minimum at
$
    m_{\mathrm{opt}}
    :=
    \left(
        \frac{4B_{f_C}}{A_{f_C}}
    \right)^{2/5}
    n^{2/5}.
$
Indeed,
$
    h_n'(m)<0
    \;\text{for }0<m<m_{\mathrm{opt}},
    \;
    h_n'(m)>0
    \;\text{for }m>m_{\mathrm{opt}}.
$
Since the Bernstein order must be a positive integer, one may choose an
integer nearest to $m_{\mathrm{opt}}$. This choice is of order $n^{2/5}$.
Substitution into the leading MISE term shows that both the variance and
squared-bias contributions are of order $n^{-4/5}$. This completes the
proof.
\end{proof}
The following lemma provides the triangular-array argument needed for
Theorem~\ref{Asymptotic-Normality-gtilde}. It establishes a pointwise
central limit theorem for the ideal Bernstein estimator, that is,
the estimator constructed from the unobserved random effects
$C_1,\ldots,C_n$.  Its proof is deferred to
Appendix~\ref{app:proofs-auxiliary-lemmas}.

\begin{lemm}
\label{lem:ideal-bernstein-pointwise-clt}
Let $C_1,C_2,\ldots$ be i.i.d. random variables with density $f_C$
supported on $[0,1]$. Assume that $f_C\in C([0,1])$, and fix
$u\in(0,1)$ such that $f_C(u)>0$.
If
$
    m_n
    \xrightarrow[n\to\infty]{}
    \infty,
    \;
    \frac{m_n}{n}
    \xrightarrow[n\to\infty]{}
    0,
$
then
\begin{equation}
    n^{1/2}m^{-1/4}
    \left[
        \widetilde f_{C,m,n}(u)
        -
        \mathbb E\bigl[\widetilde f_{C,m,n}(u)\bigr]
    \right]
    \xrightarrow[n\to\infty]{\mathcal D}
    \mathcal N\bigl(0,f_C(u)\psi(u)\bigr),
    \label{eq:ideal-pointwise-clt-lemma}
\end{equation}
where
$
    \psi(u)
    :=
    \frac{1}{\sqrt{4\pi u(1-u)}}.
$
\end{lemm}
The following theorem combines the asymptotic equivalence between the
feasible and ideal Bernstein estimators with the pointwise central limit
theorem for the ideal estimator. It
establishes the pointwise asymptotic normality of the feasible estimator
under the following sequential asymptotic regime: $T\to\infty$ first, for
fixed $n$ and $m$, and then $n\to\infty$ with $m=m_n\to\infty$.
\begin{Theo}
\seclabel{Asymptotic-Normality-gtilde}
\label{thm:pointwise-normality}
Under Assumptions~\ref{ass:population} and~\ref{ass:support}, fix
$u\in(0,1)$ such that $f_C(u)>0$. Let $m=m_n\to\infty$ satisfy
$
    \frac{m_n}{n}
    \xrightarrow[n\to\infty]{}
    0.
$
Then
\begin{equation}
    n^{1/2}m^{-1/4}
    \left[
        \widehat f_{C,m,n,T}(u)
        -
        \mathbb E\bigl[\widehat f_{C,m,n,T}(u)\bigr]
    \right]
    \xrightarrow[
        T\to\infty,\, n\to\infty
    ]{\mathcal D}
    \mathcal N\bigl(0,f_C(u)\psi(u)\bigr),
    \label{eq:pointwise-normality}
\end{equation}
where
$
    \psi(u)
    :=
    \frac{1}{\sqrt{4\pi u(1-u)}}.
$
If, in addition,
$
    \frac{n}{m_n^{5/2}}
    \xrightarrow[n\to\infty]{}
    0,
    \label{eq:pointwise-undersmoothing}
$
then the expectation in \eqref{eq:pointwise-normality} may be replaced by
$f_C(u)$. More precisely,
\begin{equation}
    n^{1/2}m^{-1/4}
    \left[
        \widehat f_{C,m,n,T}(u)
        -
        f_C(u)
    \right]
    \xrightarrow[
       T\to\infty,\, n\to\infty
    ]{\mathcal D}
    \mathcal N\bigl(0,f_C(u)\psi(u)\bigr).
    \label{eq:pointwise-normality-recentered}
\end{equation}
\end{Theo}

\begin{proof}[\bf Proof of Theorem~\ref{thm:pointwise-normality}]

Fix $u\in(0,1)$ such that $f_C(u)>0$. We first establish the result
centered by the expectation. By
Lemma~\ref{lem:ideal-bernstein-pointwise-clt}, the ideal Bernstein
estimator satisfies
\begin{equation}
    n^{1/2}m^{-1/4}
    \left[
        \widetilde f_{C,m,n}(u)
        -
        \mathbb E\bigl[\widetilde f_{C,m,n}(u)\bigr]
    \right]
    \xrightarrow[n\to\infty]{\mathcal D}
    \mathcal N\bigl(0,f_C(u)\psi(u)\bigr),
    \label{eq:ideal-pointwise-clt}
\end{equation}
as $m=m_n\to\infty$ with $m_n/n\to0$.
We next transfer this limit from the ideal estimator to the feasible
estimator. Fix $n$ and $m$. Since
$
    p_k(m-1,u)\geq0,
    \;
    \sum_{k=0}^{m-1}p_k(m-1,u)=1,
$
and the increments of the empirical distribution functions belong to
$[0,1]$, one has
$
    0
    \leq
    \widehat f_{C,m,n,T}(u)
    \leq
    m
$
and
$
    0
    \leq
    \widetilde f_{C,m,n}(u)
    \leq
    m.
$
By Theorem~\ref{thm:first-stage-reduction},
$
    \widehat f_{C,m,n,T}(u)
    \xrightarrow[T\to\infty]{\mathrm{a.s.}}
    \widetilde f_{C,m,n}(u).
$
The dominated convergence theorem yields
$
    \mathbb E\bigl[\widehat f_{C,m,n,T}(u)\bigr]
    \xrightarrow[T\to\infty]{}
    \mathbb E\bigl[\widetilde f_{C,m,n}(u)\bigr].
$
Consequently, for fixed $n$ and $m$,
\begin{align*}
    &n^{1/2}m^{-1/4}
    \left[
        \widehat f_{C,m,n,T}(u)
        -
        \mathbb E\bigl[\widehat f_{C,m,n,T}(u)\bigr]
    \right]
    \\
    &\qquad\xrightarrow[T\to\infty]{\mathrm{a.s.}}
    n^{1/2}m^{-1/4}
    \left[
        \widetilde f_{C,m,n}(u)
        -
        \mathbb E\bigl[\widetilde f_{C,m,n}(u)\bigr]
    \right].
\end{align*}
Taking next $n\to\infty$ with $m=m_n\to\infty$ and $m_n/n\to0$, and using
\eqref{eq:ideal-pointwise-clt}, yields
\[
    n^{1/2}m^{-1/4}
    \left[
        \widehat f_{C,m,n,T}(u)
        -
        \mathbb E\bigl[\widehat f_{C,m,n,T}(u)\bigr]
    \right]
    \xrightarrow[
       T\to\infty,\, n\to\infty
    ]{\mathcal D}
    \mathcal N\bigl(0,f_C(u)\psi(u)\bigr).
\]
This proves \eqref{eq:pointwise-normality}.
We finally show that the estimator may be centered at $f_C(u)$. For the
ideal estimator, write
\begin{align*}
    &n^{1/2}m^{-1/4}
    \left[
        \widetilde f_{C,m,n}(u)-f_C(u)
    \right]
    \\
    &\qquad=
    n^{1/2}m^{-1/4}
    \left[
        \widetilde f_{C,m,n}(u)
        -
        \mathbb E\bigl[\widetilde f_{C,m,n}(u)\bigr]
    \right]
    +
    n^{1/2}m^{-1/4}
    \left[
        \widetilde f_{C,m}(u)-f_C(u)
    \right],
\end{align*}
where
$
    \widetilde f_{C,m}(u)
    =
    \mathbb E\bigl[\widetilde f_{C,m,n}(u)\bigr].
$
By Lemma~\ref{lem:ideal-bernstein-bias},
$
    \widetilde f_{C,m}(u)-f_C(u)
    =
    O(m^{-1}).
$
Hence,
\begin{align*}
    n^{1/2}m^{-1/4}
    \left|
        \widetilde f_{C,m}(u)-f_C(u)
    \right|
    &=
    O\left(n^{1/2}m^{-5/4}\right)
    \\
    &=
    O\left(
        \sqrt{\frac{n}{m^{5/2}}}
    \right)
    \xrightarrow[n\to\infty]{}
    0,
\end{align*}
because $n/m_n^{5/2}\to0$. Slutsky's theorem and
\eqref{eq:ideal-pointwise-clt} therefore give
\[
    n^{1/2}m^{-1/4}
    \left[
        \widetilde f_{C,m,n}(u)-f_C(u)
    \right]
    \xrightarrow[n\to\infty]{\mathcal D}
    \mathcal N\bigl(0,f_C(u)\psi(u)\bigr).
\]
For fixed $n$ and $m$, Theorem~\ref{thm:first-stage-reduction} also
implies that
\begin{align*}
    &n^{1/2}m^{-1/4}
    \left[
        \widehat f_{C,m,n,T}(u)-f_C(u)
    \right]
    \\
    &\qquad\xrightarrow[T\to\infty]{\mathrm{a.s.}}
    n^{1/2}m^{-1/4}
    \left[
        \widetilde f_{C,m,n}(u)-f_C(u)
    \right].
\end{align*}
Taking next $n\to\infty$ with $m=m_n\to\infty$, $m_n/n\to0$, and
$n/m_n^{5/2}\to0$, we conclude that
\[
    n^{1/2}m^{-1/4}
    \left[
        \widehat f_{C,m,n,T}(u)-f_C(u)
    \right]
    \xrightarrow[
        T\to\infty,\, n\to\infty
    ]{\mathcal D}
    \mathcal N\bigl(0,f_C(u)\psi(u)\bigr).
\]
This proves \eqref{eq:pointwise-normality-recentered} and completes the
proof.
\end{proof}

\begin{rema}
\begin{enumerate}
\item
The MISE-optimal Bernstein order derived in
Theorem~\ref{thm:bias-var-mise} is obtained by minimizing the leading
expression
$
    A_{f_C}\frac{m^{1/2}}{n}
    +
    B_{f_C}m^{-2}.
$
When $B_{f_C}>0$, its continuous minimizer is
$
    m_{\mathrm{opt}}
    =
    \left(
        \frac{4B_{f_C}}{A_{f_C}}
    \right)^{2/5}
    n^{2/5}.
$
By contrast, replacing
$\mathbb E[\widehat f_{C,m,n,T}(u)]$ by $f_C(u)$ in
Theorem~\ref{thm:pointwise-normality} requires
$
    \frac{n}{m^{5/2}}
    \xrightarrow[n\to\infty]{}
    0.
$
Equivalently, $m$ must grow faster than $n^{2/5}$. At the MISE-optimal
order,
$
    \frac{n}{m_{\mathrm{opt}}^{5/2}}
    =
    \frac{A_{f_C}}{4B_{f_C}}
    >
    0,
$
so the preceding condition is not satisfied.
To describe the effect, define
$
    \beta_{f_C}(u)
    :=
    \frac{1}{2}
    \left[
        (1-2u)f_C'(u)
        +
        u(1-u)f_C''(u)
    \right].
$
By Lemma~\ref{lem:ideal-bernstein-bias},
$
    \widetilde f_{C,m}(u)-f_C(u)
    =
    \frac{\beta_{f_C}(u)}{m}
    +
    o(m^{-1}).
$
Consequently, at the MISE-optimal order,
$
    n^{1/2}m_{\mathrm{opt}}^{-1/4}
    \left[
        \widetilde f_{C,m_{\mathrm{opt}}}(u)-f_C(u)
    \right]
    \xrightarrow[n\to\infty]{}
    \beta_{f_C}(u)
    \left(
        \frac{A_{f_C}}{4B_{f_C}}
    \right)^{1/2}.
$

Thus, unless $\beta_{f_C}(u)=0$, the normalized limiting distribution
centered at $f_C(u)$ has a nonzero mean shift. This illustrates the usual
tension in nonparametric density estimation between the smoothing order
that minimizes the MISE and the stronger growth condition required for an
uncorrected pointwise central limit theorem centered at the true density.
Consequently, the MISE-optimal order should not generally be used for
pointwise confidence intervals based directly on
\eqref{eq:pointwise-normality-recentered} without an explicit bias
correction. One possible uncorrected choice is
$
    m_n
    =
    \left\lfloor n^\alpha\right\rfloor,
    \;
    \alpha\in\left(\frac{2}{5},1\right).
$
Indeed, this choice satisfies
$
    \frac{m_n}{n}
    \xrightarrow[n\to\infty]{}
    0
    \;\text{and}\;
    \frac{n}{m_n^{5/2}}
    \xrightarrow[n\to\infty]{}
    0.
$
The price of this stronger growth condition is a slower MISE convergence
rate than the optimal order $n^{-4/5}$.

\item
The condition
$
    \frac{m_n}{n}
    \xrightarrow[n\to\infty]{}
    0
$
imposed in Lemma~\ref{lem:ideal-bernstein-pointwise-clt} is a convenient
sufficient condition for verifying the Lindeberg condition. Indeed, the
proof uses the bound
$
    \frac{|Y_{1,n}|}{\varepsilon\sqrt n\,s_n}
    =
    O\left(
        \left(\frac{m_n}{n^2}\right)^{1/4}
    \right).
$
Therefore, the weaker condition
$
    \frac{m_n}{n^2}
    \xrightarrow[n\to\infty]{}
    0
$
already suffices for the triangular-array Lindeberg--Feller argument.
This observation does not affect the validity of
Theorem~\ref{thm:pointwise-normality}, but it shows that a wider range of
Bernstein orders is available for the expectation-centered limit
\eqref{eq:pointwise-normality}.

\end{enumerate}
\end{rema}

\begin{rema}
The sequential asymptotic framework separates the estimation of the
individual random effects from the estimation of their common density.
For a finite observation horizon, however, the estimator
satisfies
$
    \widehat C_{i,T}-C_i
    =
    \frac{N_T^{(i)}-C_iT}{T},
$
and, conditionally on $C_i$,
$
    \operatorname{Var}
    \bigl(\widehat C_{i,T}\mid C_i\bigr)
    =
    \frac{C_i}{T}.
$
Thus, the first-stage estimation error may be substantial when $T$ is
small. 
Although the conditional variance of $\widehat C_{i,T}$ is smaller
for smaller values of $C_i$, its conditional relative variance is
$
    \operatorname{Var}
    \left(
        \left.
        \frac{\widehat C_{i,T}-C_i}{C_i}
        \,\right|\,C_i
    \right)
    =
    \frac{1}{C_iT},
$
which becomes large when $C_iT$ is small. In particular, small expected
jump counts lead to substantial relative uncertainty and a non-negligible
probability of observing no jump, since
$
    \mathbb P
    \bigl(N_T^{(i)}=0\mid C_i\bigr)
    =
    e^{-C_iT}.
$

This observation does not prevent the use of the feasible estimator for a
finite observation horizon, as in the numerical study. It only indicates
that the resulting Bernstein density estimate retains some first-stage
estimation uncertainty, particularly when the expected jump counts
$C_iT$ are small.

A simultaneous asymptotic theory, in which $T$, $n$, and $m$ tend to
infinity jointly, would require suitable growth conditions ensuring that
the first-stage Poisson estimation error is negligible relative to the
Bernstein smoothing error.
\end{rema}

\section{Specialization to the Gamma jump-size family}
\label{sec:applications}

This section specializes the general framework of
Section~\ref{sec:general-framework} to a parametric Gamma family of
jump-size densities with a fixed and known shape parameter. Rather than
treating the exponential and Erlang cases separately, we use a single
Gamma family in which the exponential model corresponds to $k=1$, whereas
the Erlang model of order two corresponds to $k=2$.

The key point is that the first-stage random-effect estimator
$
    \widehat C_{i,T}
    =
    \frac{N_T^{(i)}}{T}
$
depends only on the observed jump count. Hence, it is invariant with
respect to the jump-size distribution. 
The specification of the jump-size
distribution affects the jump-size likelihood, the estimation of the common
structural parameter $\lambda$, the immigration mechanism, and the
stationary moments, but it does not alter the Bernstein estimator
$\widehat f_{C,m,n,T}$ of the random-intensity density.

Let $k\geq1$ be a fixed and known integer. Conditionally on the random
intensity $C_i$, the jumps of the $i$-th trajectory occur at rate $C_i$,
and their sizes are independent and identically distributed positive
random variables with Gamma density
\begin{equation}
    q_\lambda^{(k)}(z)
    =
    \frac{\lambda^k}{\Gamma(k)}
    z^{k-1}e^{-\lambda z}
    \mathds{1}_{(0,\infty)}(z),
    \;
    \lambda>0,
    \label{eq:gamma-k-density}
\end{equation}
where $\lambda$ is the rate parameter. Under this specification, a jump
size $Z\sim q_\lambda^{(k)}$ satisfies
$
    \mathbb E_\lambda^{(k)}(Z)
    =
    \frac{k}{\lambda},
    \;
    \operatorname{Var}_\lambda^{(k)}(Z)
    =
    \frac{k}{\lambda^2},
    \;
    \mathbb E_\lambda^{(k)}(Z^2)
    =
    \frac{k(k+1)}{\lambda^2}.
$
In particular, the Gamma density satisfies the first-moment condition
\eqref{eq:q-first-moment} and has finite moments of every order.
Conditionally on $C_i=c$, the generator of the corresponding BAJD
trajectory is
$
    \mathcal L_c^{(k)}f(y)
    =
    (a-by)f'(y)
    +
    \frac{\sigma^2}{2}y f''(y)
    +
    c\int_0^\infty
    \bigl[f(y+z)-f(y)\bigr]
    q_\lambda^{(k)}(z)\dd z.
$
Equivalently, its immigration mechanism is
$
    F_c^{(k)}(u)
    =
    au
    +
    c
    \left[
        \left(\frac{\lambda}{\lambda-u}\right)^k
        -1
    \right],
    \;
    u<\lambda.
$
In the ergodic case $b>0$, the conditional stationary mean is
$
    m_\infty^{(k)}(c)
    :=
    \lim_{t\to\infty}
    \mathbb E\bigl(Y_t^{(i)}\mid C_i=c\bigr)
    =
    \frac{a+kc/\lambda}{b}.
$
Consequently, the stationary mean at the population level is
\begin{equation}
    \mathbb E(Y_\infty)
    =
    \frac{1}{b}
    \left(
        a+\frac{k}{\lambda}\mathbb E(C_1)
    \right).
    \label{eq:population-stationary-mean-gamma-k}
\end{equation}

\subsection{Pooled estimation of the common jump-size parameter}

For each trajectory, define, for $t\geq0$,
$
    J_t^{(i)}
    :=
    \sum_{\ell=1}^{N_t^{(i)}}Z_\ell^{(i)}.
$
Thus, $J_T^{(i)}$ is the total jump amplitude observed over $[0,T]$.
We also define
$
    N_{n,t}
    :=
    \sum_{i=1}^nN_t^{(i)},
    \;
    J_{n,t}
    :=
    \sum_{i=1}^nJ_t^{(i)}.
$
In particular, $N_{n,T}$ and $J_{n,T}$ denote, respectively, the pooled
jump count and the pooled jump amplitude at the observation horizon $T$.
Under the Gamma$(k,\lambda)$ jump-size specification, the joint jump
log-likelihood is
\begin{equation}
\begin{aligned}
    \ell_{n,T}^{J,(k)}
    (C_1,\ldots,C_n,\lambda)
    &=
    \sum_{i=1}^n
    \left[
        N_T^{(i)}\log C_i-C_iT
    \right]
    +
    kN_{n,T}\log\lambda
    -
    \lambda J_{n,T}
    \\
    &\quad
    +
    (k-1)
    \sum_{i=1}^n
    \sum_{\ell=1}^{N_T^{(i)}}
    \log Z_\ell^{(i)}
    -
    N_{n,T}\log\Gamma(k).
\end{aligned}
\label{eq:gamma-k-loglik}
\end{equation}
Since $k$ is fixed and known, the last two terms depend only on the
observed jump sizes and counts and not on the unknown parameters
$(C_1,\ldots,C_n,\lambda)$.
Maximizing the component depending on $C_i$ over the extended parameter
space $\mathbb R_+$ gives
$
    \widehat C_{i,T}
    =
    \frac{N_T^{(i)}}{T},
    \;
    i=1,\ldots,n.
$
If the parameter space is restricted to $[0,1]$, the corresponding
maximizer is
$
    \widehat C_{i,T}^{\mathcal C}
    =
    \min\left\{
        \frac{N_T^{(i)}}{T},1
    \right\}.
$
The component of the log-likelihood depending on $\lambda$ is unaffected
by the restriction of the individual intensities to $[0,1]$. The
following proposition identifies the pooled Gamma-rate MLE and establishes
its consistency and asymptotic normality.

\begin{prop}
\label{prop:lambda-gamma-k}
Under Assumptions~\ref{ass:population} and~\ref{ass:support}, fix
$n\geq1$, and suppose that the jump-size density is given by
\eqref{eq:gamma-k-density}, where $k\geq1$ is a fixed and known integer
and $\lambda\in\real_{++}$.
 On the event $\{N_{n,T}>0\}$, the pooled jump-size likelihood, regarded
as a function of $\lambda\in\real_{++}$, attains its unique maximum at
$
    \widehat\lambda_{n,T}^{(k)}
    =
    \frac{kN_{n,T}}{J_{n,T}}.
$
 On the event $\{N_{n,T}=0\}$, the pooled jump-size likelihood does not
contain information about $\lambda$. Accordingly, define
\begin{equation}
    \widehat\lambda_{n,T}^{(k)}
    :=
    \begin{cases}
        \displaystyle
        \frac{kN_{n,T}}{J_{n,T}},
        & N_{n,T}>0,\\[2mm]
        0,
        & N_{n,T}=0.
    \end{cases}
    \label{eq:lambda-mle-gamma-k}
\end{equation}
One has
\begin{equation}
    \widehat\lambda_{n,T}^{(k)}
    \xrightarrow[T\to\infty]{\mathrm{a.s.}}
    \lambda.
    \label{eq:lambda-gamma-k-consistency}
\end{equation}
Moreover,
\begin{equation}
    \sqrt{N_{n,T}}
    \left(
        \widehat\lambda_{n,T}^{(k)}-\lambda
    \right)
    \xrightarrow[T\to\infty]{\mathcal D}
    \mathcal N
    \left(
        0,\frac{\lambda^2}{k}
    \right).
    \label{eq:lambda-gamma-k-clt}
\end{equation}
\end{prop}

\begin{proof}[\bf Proof of Proposition~\ref{prop:lambda-gamma-k}]

On the event $\{N_{n,T}>0\}$, one also has $J_{n,T}>0$, since all jump
sizes are strictly positive. 
It follows from \eqref{eq:gamma-k-loglik} that the only terms depending
on $\lambda$ are
$
    \ell_{n,T}^{\lambda,(k)}(\lambda)
    :=
    kN_{n,T}\log\lambda
    -
    \lambda J_{n,T},
    \;
    \lambda\in\real_{++}.
$
Its first two derivatives are
$
    \frac{\partial}{\partial\lambda}
    \ell_{n,T}^{\lambda,(k)}(\lambda)
    =
    \frac{kN_{n,T}}{\lambda}-J_{n,T}
$
and
$
    \frac{\partial^2}{\partial\lambda^2}
    \ell_{n,T}^{\lambda,(k)}(\lambda)
    =
    -\frac{kN_{n,T}}{\lambda^2}
    <0.
$
Thus, on the event $\{N_{n,T}>0\}$, the log-likelihood is strictly
concave on $\real_{++}$, and its unique maximizer is
$
    \widehat\lambda_{n,T}^{(k)}
    =
    \frac{kN_{n,T}}{J_{n,T}}.
$
On the event $\{N_{n,T}=0\}$, one has $J_{n,T}=0$, and the jump-size
likelihood is independent of $\lambda$. This proves the finite-sample
claims.
Set
$
    \mathcal G_C
    :=
    \sigma(C_1,\ldots,C_n),
    \;
    S_C
    :=
    \sum_{i=1}^nC_i.
$
By Assumption~\ref{ass:support},
$
    S_C>0
    \;\text{a.s.}
$
Conditionally on $\mathcal G_C$, the counting processes
$N^{(1)},\ldots,N^{(n)}$ are independent Poisson processes with
respective intensities $C_1,\ldots,C_n$. Consequently, the aggregate
counting process
$
    N_n=(N_{n,t})_{t\geq0}
$
is conditionally a Poisson process with intensity $S_C$. The strong law
of large numbers for Poisson processes therefore gives
$
    \mathbb P\left(
        \left.
        \lim_{T\to\infty}
        \frac{N_{n,T}}{T}
        =
        S_C
        \,\right|\,
        \mathcal G_C
    \right)
    =
    1
    \;\text{a.s.}
$
Taking expectations and using the tower property yields
\begin{equation}
    \frac{N_{n,T}}{T}
    \xrightarrow[T\to\infty]{\mathrm{a.s.}}
    S_C.
    \label{eq:aggregate-count-slln}
\end{equation}
In particular,
$
    N_{n,T}
    \xrightarrow[T\to\infty]{\mathrm{a.s.}}
    \infty.
$
Consequently, almost surely, $N_{n,T}>0$ and $J_{n,T}>0$ for all
sufficiently large $T$.
Under the Gamma$(k,\lambda)$ specification,
$
    \mathbb E_\lambda^{(k)}(Z)
    =
    \frac{k}{\lambda}.
$
The strong law of large numbers for compound Poisson processes yields
\begin{equation}
    \frac{J_{n,T}}{T}
    \xrightarrow[T\to\infty]{\mathrm{a.s.}}
    \frac{k}{\lambda}S_C.
    \label{eq:aggregate-jump-slln}
\end{equation}
Since $S_C>0$ almost surely, it follows from
\eqref{eq:aggregate-count-slln} and
\eqref{eq:aggregate-jump-slln} that
$
    \widehat\lambda_{n,T}^{(k)}
    =
    k\frac{N_{n,T}/T}{J_{n,T}/T}\xrightarrow[T\to\infty]{\mathrm{a.s.}}
    k\frac{S_C}{(k/\lambda)S_C}
    =
    \lambda.
$
This proves \eqref{eq:lambda-gamma-k-consistency}.
We now establish the central limit theorem by applying the martingale GCLT
stated in Theorem~\ref{GCLTtouati} of Appendix~B. Let
$\mu_i(\dd t,\dd z)$ denote the jump measure of the $i$-th
trajectory. With respect to the filtration initially enlarged by
$\mathcal G_C$, its compensator is
$
    \nu_i(\dd t,\dd z)
    =
    C_iq_\lambda^{(k)}(z)\dd t\,\dd z.
$
Define the process $R^{(k)}=(R_t^{(k)})_{t\geq0}$ by
\begin{equation}
    R_t^{(k)}
    :=
    \sum_{i=1}^n
    \int_0^t\int_0^\infty
    \left(
        z-\frac{k}{\lambda}
    \right)
    \bigl(\mu_i-\nu_i\bigr)(\dd s,\dd z),
    \;
    t\geq0.
    \label{eq:mark-martingale-gamma-k}
\end{equation}
This process is a purely discontinuous, square-integrable,
quasi-left-continuous martingale. Indeed, its jump times are totally
inaccessible. Since
$
    \int_0^\infty
    \left(
        z-\frac{k}{\lambda}
    \right)
    q_\lambda^{(k)}(z)\dd z
    =
    0,
$
one has
$
    R_t^{(k)}
    =
    J_{n,t}
    -
    \frac{k}{\lambda}N_{n,t},
    \;
    t\geq0.
$
Its predictable quadratic variation is
\begin{align}
    \left\langle R^{(k)}\right\rangle_T
    &=
    \sum_{i=1}^n
    C_iT
    \int_0^\infty
    \left(
        z-\frac{k}{\lambda}
    \right)^2
    q_\lambda^{(k)}(z)\dd z
    =
    \frac{k}{\lambda^2}S_CT.
    \label{eq:mark-martingale-bracket}
\end{align}
Therefore,
$
    \frac{1}{T}
    \left\langle R^{(k)}\right\rangle_T
    =
    \frac{k}{\lambda^2}S_C.
$
For every $\varepsilon>0$, the Lindeberg term associated with the
normalization $\sqrt T$ is
\begin{align*}
    \Delta_T^\varepsilon
    &:=
    \frac{1}{T}
    \sum_{i=1}^n
    \int_0^T\int_0^\infty
    \left(
        z-\frac{k}{\lambda}
    \right)^2
    \ind_{
        \{
            |z-k/\lambda|>\varepsilon\sqrt T
        \}
    }
    \nu_i(\dd t,\dd z)\\
    &=
    S_C
    \int_0^\infty
    \left(
        z-\frac{k}{\lambda}
    \right)^2
    \ind_{
        \{
            |z-k/\lambda|>\varepsilon\sqrt T
        \}
    }
    q_\lambda^{(k)}(z)\dd z.
\end{align*}
Since the Gamma distribution has a finite second moment, the dominated
convergence theorem gives
$
    \Delta_T^\varepsilon
    \xrightarrow[T\to\infty]{\mathrm{a.s.}}
    0.
$
Thus, the bracket and Lindeberg conditions of
Theorem~\ref{GCLTtouati} are satisfied. Consequently, stably with respect
to $\mathcal G_C$,
$
    \frac{R_T^{(k)}}{\sqrt T}
    \xrightarrow[T\to\infty]{\mathcal D,\mathrm{st}}
    \frac{\sqrt{kS_C}}{\lambda}Z_R,
$
where $Z_R\sim\mathcal N(0,1)$ is defined on an extension of the original
probability space and is independent of $\mathcal G_C$.
By \eqref{eq:aggregate-count-slln},
$
    \sqrt{\frac{T}{N_{n,T}}}
    \xrightarrow[T\to\infty]{\mathrm{a.s.}}
    \frac{1}{\sqrt{S_C}}.
$
Hence, the stable version of Slutsky's theorem gives
\begin{equation}
    \frac{R_T^{(k)}}{\sqrt{N_{n,T}}}
    =
    \frac{R_T^{(k)}}{\sqrt T}
    \sqrt{\frac{T}{N_{n,T}}}
    \xrightarrow[T\to\infty]{\mathcal D,\mathrm{st}}
    \frac{\sqrt{k}}{\lambda}Z_R.
    \label{eq:mark-martingale-clt-N}
\end{equation}
In particular,
$
    \frac{R_T^{(k)}}{\sqrt{N_{n,T}}}
    \xrightarrow[T\to\infty]{\mathcal D}
    \mathcal N
    \left(
        0,\frac{k}{\lambda^2}
    \right).
$
Almost surely, for all sufficiently large $T$,
$
    \widehat\lambda_{n,T}^{(k)}-\lambda
    =
    \frac{kN_{n,T}}{J_{n,T}}-\lambda=
    -\lambda\frac{R_T^{(k)}}{J_{n,T}}.
$
Therefore,
$
    \sqrt{N_{n,T}}
    \left(
        \widehat\lambda_{n,T}^{(k)}-\lambda
    \right)
    =
    -\lambda
    \frac{N_{n,T}}{J_{n,T}}
    \frac{R_T^{(k)}}{\sqrt{N_{n,T}}}.
$
Furthermore, by
\eqref{eq:aggregate-count-slln}--\eqref{eq:aggregate-jump-slln},
$
    \frac{N_{n,T}}{J_{n,T}}
    \xrightarrow[T\to\infty]{\mathrm{a.s.}}
    \frac{\lambda}{k}.
$
Slutsky's theorem and \eqref{eq:mark-martingale-clt-N} then yield
$
    \sqrt{N_{n,T}}
    \left(
        \widehat\lambda_{n,T}^{(k)}-\lambda
    \right)
    \xrightarrow[T\to\infty]{\mathcal D}
    -\frac{\lambda}{\sqrt{k}}Z_R.
$
Since
$
    -\frac{\lambda}{\sqrt{k}}Z_R
    \sim
    \mathcal N
    \left(
        0,\frac{\lambda^2}{k}
    \right),
$
this proves \eqref{eq:lambda-gamma-k-clt}.
\end{proof}

\begin{rema}
The unified Gamma$(k,\lambda)$ specification leaves the first-stage
estimator
$
    \widehat C_{i,T}
    =
    \frac{N_T^{(i)}}{T}
$
unchanged. Consequently, the projected observations
$\widehat C_{i,T}^{\mathcal C}$ and the feasible Bernstein estimator
$\widehat f_{C,m,n,T}$ do not depend on either the shape parameter $k$ or
the common rate parameter $\lambda$.
Under Assumptions~\ref{ass:population} and~\ref{ass:support}, the
first-stage convergence result in
Theorem~\ref{thm:first-stage-reduction}, the bias--variance--MISE
expansion in Theorem~\ref{thm:bias-var-mise}, and the pointwise
asymptotic normality result in Theorem~\ref{thm:pointwise-normality}
apply without modification to $\widehat f_{C,m,n,T}$ under the
Gamma$(k,\lambda)$ jump-size family
\eqref{eq:gamma-k-density}.
Indeed, the estimator $\widehat C_{i,T}$ depends only on the counting
process $N^{(i)}$. Hence, $\widehat F_{C,n,T}$ and
$\widehat f_{C,m,n,T}$ are constructed in the same way,
regardless of the jump-size density $q_\lambda^{(k)}$. 
This invariance makes explicit the separation between frequency
heterogeneity and jump-size calibration. The density $f_C$ is estimated
from the individual jump counts, whereas the common rate parameter
$\lambda$ is estimated from the pooled jump counts and jump amplitudes
through \eqref{eq:lambda-mle-gamma-k}.
\end{rema}

\subsection{Composite estimation of the stationary mean}
\label{subsec:composite-stationary-mean}
Equation~\eqref{eq:population-stationary-mean-gamma-k} suggests a
composite plug-in estimator of the stationary population mean. Throughout
this subsection, the common parameters $a$ and $b$ are treated as known.
In practice, they must therefore be supplied by an external calibration,
for example, through a likelihood-based or moment-based procedure applied
to a pilot trajectory, a topic beyond the scope of this paper.
Define
$
    \widehat\mu_{C,m,n,T}
    :=
    \int_0^1
    c\,\widehat f_{C,m,n,T}(c)\dd c.
    \label{eq:estimated-mean-random-C}
$
On the event $\{\widehat\lambda_{n,T}^{(k)}>0\}$, define
\begin{equation}
    \widehat{\mathbb E}_{n,T}^{(k)}(Y_\infty)
    :=
    \frac{1}{b}
    \left(
        a+
        \frac{k}{\widehat\lambda_{n,T}^{(k)}}
        \widehat\mu_{C,m,n,T}
    \right).
    \label{eq:composite-stationary-mean-estimator}
\end{equation}
For definiteness, its value may be set equal to zero on
$\{\widehat\lambda_{n,T}^{(k)}=0\}$. This arbitrary value does not affect 
the asymptotic result below.

The following proposition establishes the consistency of the proposed
estimator of the population stationary mean.
We consider the sequential asymptotic regime in which $T\to\infty$ first,
for fixed $n$ and $m$, and then $n\to\infty$ with $m=m_n$ satisfying
\begin{equation}
    m_n
    \xrightarrow[n\to\infty]{}
    \infty,
\;
    \frac{\sqrt{m_n}}{n}
    \xrightarrow[n\to\infty]{}
    0.
    \label{eq:stationary-mean-asymptotic-regime}
\end{equation}

\begin{prop}
\label{prop:stationary-mean-consistency}

Under Assumptions~\ref{ass:population} and~\ref{ass:support}, suppose that
$a\in\real_{++}$ and $b\in\real_{++}$ are known. Let the jump-size
density be the Gamma$(k,\lambda)$ density
\eqref{eq:gamma-k-density}, where $k\geq1$ is a fixed and known integer
and $\lambda\in\real_{++}$.
For every fixed $n\geq1$,
$\widehat\lambda_{n,T}^{(k)}>0$ almost surely for all sufficiently large
$T$. Moreover, under the sequential asymptotic regime described above
and the growth conditions
\eqref{eq:stationary-mean-asymptotic-regime}, for every $\varepsilon>0$,
\begin{equation}
    \lim_{n\to\infty}
    \left[
        \lim_{T\to\infty}
        \mathbb P
        \left(
            \left|
                \widehat{\mathbb E}_{n,T}^{(k)}(Y_\infty)
                -
                \mathbb E(Y_\infty)
            \right|
            >
            \varepsilon
        \right)
    \right]
    =
    0,
    \label{eq:composite-stationary-mean-consistency}
\end{equation}
where
$
    \mathbb E(Y_\infty)
    =
    \frac{1}{b}
    \left(
        a+\frac{k}{\lambda}\mathbb E(C_1)
    \right).
$
\end{prop}
\begin{proof}[\bf Proof of Proposition~\ref{prop:stationary-mean-consistency}]
Fix $n,m\geq1$. With
$
    S_C
    :=
    \sum_{i=1}^nC_i,
$
Assumption~\ref{ass:support} implies that $S_C>0$ almost surely. By
\eqref{eq:aggregate-count-slln},
$
    \frac{N_{n,T}}{T}
    \xrightarrow[T\to\infty]{\mathrm{a.s.}}
    S_C>0.
$
Consequently, almost surely, $N_{n,T}>0$ for all sufficiently large $T$.
Since the Gamma jump-sizes are strictly positive, this also implies that
$J_{n,T}>0$. Hence,
$    \widehat\lambda_{n,T}^{(k)}
    =
    \frac{kN_{n,T}}{J_{n,T}}
    >0
$
almost surely for all sufficiently large $T$. Moreover,
Proposition~\ref{prop:lambda-gamma-k} gives
\begin{equation}
    \widehat\lambda_{n,T}^{(k)}
    \xrightarrow[T\to\infty]{\mathrm{a.s.}}
    \lambda.
    \label{eq:lambda-consistency-stationary-mean-proof}
\end{equation}
Thus, the arbitrary value assigned to
$\widehat{\mathbb E}_{n,T}^{(k)}(Y_\infty)$ on
$\{\widehat\lambda_{n,T}^{(k)}=0\}$ has no effect on its asymptotic
behavior. Since
$
    \mathbb E(C_1)
    =
    \int_0^1c f_C(c)\dd c,
$
the Cauchy--Schwarz inequality yields
\begin{align*}
    \left|
        \widehat\mu_{C,m,n,T}
        -
        \mathbb E(C_1)
    \right|
    &\leq
    \left(
        \int_0^1c^2\dd c
    \right)^{1/2}
    \left(
        \int_0^1
        \bigl(
            \widehat f_{C,m,n,T}(c)-f_C(c)
        \bigr)^2
        \dd c
    \right)^{1/2}\\
    &=
    \frac{1}{\sqrt3}
    \left(
        \int_0^1
        \bigl(
            \widehat f_{C,m,n,T}(c)-f_C(c)
        \bigr)^2
        \dd c
    \right)^{1/2}.
\end{align*}
Therefore, for every $\varepsilon>0$, Markov's inequality gives
$\mathbb P
    \left(
        \left|
            \widehat\mu_{C,m,n,T}
            -
            \mathbb E(C_1)
        \right|
        >
        \varepsilon
    \right)\leq
    \frac{1}{3\varepsilon^2}
    \operatorname{MISE}
    \bigl(\widehat f_{C,m,n,T}\bigr).
$
Under the growth conditions
\eqref{eq:stationary-mean-asymptotic-regime},
Theorem~\ref{thm:bias-var-mise}, and in particular
\eqref{eq:mise-consistency}, implies that
$$
     \lim_{n\to\infty}
    \left[
        \lim_{T\to\infty}
    \mathbb P
    \left(
        \left|
            \widehat\mu_{C,m_n,n,T}
            -
            \mathbb E(C_1)
        \right|
        >
        \varepsilon
    \right)\right]
    =
    0.
$$
Thus, under the stated sequential asymptotic regime,
\begin{equation}
    \widehat\mu_{C,m_n,n,T}
    \xrightarrow{\mathbb P}
    \mathbb E(C_1).
    \label{eq:mu-C-consistency-stationary-mean-proof}
\end{equation}
Combining
\eqref{eq:lambda-consistency-stationary-mean-proof} and
\eqref{eq:mu-C-consistency-stationary-mean-proof}, and applying the
continuous mapping theorem to
$
    (x,y)
    \longmapsto
    \frac{1}{b}
    \left(
        a+\frac{ky}{x}
    \right)
$
at
$
    \bigl(\lambda,\mathbb E(C_1)\bigr)
    \in
    (0,\infty)\times[0,1],
$
we obtain
$
    \widehat{\mathbb E}_{n,T}^{(k)}(Y_\infty)
    \xrightarrow{\mathbb P}
    \frac{1}{b}
    \left(
        a+\frac{k}{\lambda}\mathbb E(C_1)
    \right)
    =
    \mathbb E(Y_\infty).
$
This proves
\eqref{eq:composite-stationary-mean-consistency}.
\end{proof}

The following proposition formalizes the first-order separation between
the frequency estimation problem, driven by the compensated counting
martingales, and the jump-size calibration problem, driven by the centered
jump-size martingale.
This separation is particularly useful for composite
estimators such as \eqref{eq:composite-stationary-mean-estimator}, in which
the frequency and amplitude contributions enter simultaneously.

\begin{prop}
\label{prop:orthogonality-asymptotic-independence}
Under Assumptions~\ref{ass:population} and~\ref{ass:support}, fix
$n\geq1$, and suppose that the jump-size density is the
Gamma$(k,\lambda)$ density \eqref{eq:gamma-k-density}, where $k\geq1$ is
a fixed and known integer and $\lambda>0$. Set
$
    \mathcal G_C
    :=
    \sigma(C_1,\ldots,C_n),
    \;
    S_C
    :=
    \sum_{i=1}^nC_i,
$
and define
$
    M_{i,t}
    :=
    N_t^{(i)}-C_it,
    \;
    t\geq0,
    \;
    i=1,\ldots,n.
$

Let $R^{(k)}$ be the centered jump-size martingale defined in
\eqref{eq:mark-martingale-gamma-k}.
Then, for every
$i\in\{1,\ldots,n\}$ and every $T\in \real_{++}$,
\begin{equation}
    \left\langle M_i,R^{(k)}\right\rangle_T
    =
    0.
    \label{eq:Mi-R-orthogonality}
\end{equation}

Moreover, as $T\to\infty$,
\begin{align}
    &\left(
        \sqrt T\bigl(\widehat C_{1,T}-C_1\bigr),
        \ldots,
        \sqrt T\bigl(\widehat C_{n,T}-C_n\bigr),
        \sqrt T
        \bigl(\widehat\lambda_{n,T}^{(k)}-\lambda\bigr)
    \right)
    \notag\\
    &\qquad
    \xrightarrow[T\to\infty]{\mathcal D,\mathrm{st}}
    \left(
        \sqrt{C_1}Z_1,
        \ldots,
        \sqrt{C_n}Z_n,
        \frac{\lambda}{\sqrt{kS_C}}Z_\lambda
    \right)
    \label{eq:joint-Chat-lambda-stable}
\end{align}
stably with respect to $\mathcal G_C$, where
$Z_1,\ldots,Z_n,Z_\lambda$ are independent standard normal random
variables, defined on an extension of the original probability space and
independent of $\mathcal G_C$.
Consequently, conditionally on the random effects, the individual
intensity estimators and the pooled Gamma-rate estimator are
asymptotically independent at first order.
\end{prop}
\begin{proof}[\bf Proof of Proposition~\ref{prop:orthogonality-asymptotic-independence}]
Consider the filtration initially enlarged by $\mathcal G_C$, defined by
$  \mathcal G_t
    :=
    \mathcal F_t\vee\mathcal G_C
    =
    \sigma\bigl(\mathcal F_t\cup\mathcal G_C\bigr),
    \;
    t\geq0.
$
Conditionally on $\mathcal G_C$, the counting processes
$N^{(1)},\ldots,N^{(n)}$ are independent Poisson processes with respective
intensities $C_1,\ldots,C_n$. Hence, the compensated counting martingales
$M_1,\ldots,M_n$ are mutually orthogonal and satisfy
$
    \left\langle M_i,M_j\right\rangle_T
    =
    \delta_{ij}C_iT,
    \;
    i,j\in\{1,\ldots,n\}.
$
The compensator of the jump measure $\mu_i$ associated with the $i$-th
trajectory is
$
    \nu_i(\dd t,\dd z)
    =
    C_iq_\lambda^{(k)}(z)\dd t\,\dd z.
$
Since $q_\lambda^{(k)}$ is a probability density, one has
$
    M_{i,t}
    =
    \int_0^t\int_0^\infty
    \bigl(\mu_i-\nu_i\bigr)(\dd s,\dd z).
$
Using the definition of $R^{(k)}$, the predictable covariation between
$M_i$ and $R^{(k)}$ is
$$
    \left\langle M_i,R^{(k)}\right\rangle_T
    =
    \int_0^T\int_0^\infty
    \left(
        z-\frac{k}{\lambda}
    \right)
    \nu_i(\dd t,\dd z)=
    C_iT
    \int_0^\infty
    \left(
        z-\frac{k}{\lambda}
    \right)
    q_\lambda^{(k)}(z)\dd z.
$$
Since
$
    \int_0^\infty
    zq_\lambda^{(k)}(z)\dd z
    =
    \frac{k}{\lambda},
$
it follows that
$
    \int_0^\infty
    \left(
        z-\frac{k}{\lambda}
    \right)
    q_\lambda^{(k)}(z)\dd z
    =
    0.
$
Therefore,
$
    \left\langle M_i,R^{(k)}\right\rangle_T
    =
    0,
$
which proves \eqref{eq:Mi-R-orthogonality}.
By \eqref{eq:mark-martingale-bracket},
$
    \left\langle R^{(k)}\right\rangle_T
    =
    \frac{k}{\lambda^2}S_CT.
$
Consequently, the predictable covariance matrix of the
$(n+1)$-dimensional martingale
$
    \mathbb M_t
    :=
    \left(
        M_{1,t},\ldots,M_{n,t},R_t^{(k)}
    \right)^\top
$
satisfies
$
    \frac{1}{T}
    \left\langle\mathbb M\right\rangle_T
    =
    \operatorname{diag}
    \left(
        C_1,\ldots,C_n,
        \frac{k}{\lambda^2}S_C
    \right).
$

We next verify the Lindeberg condition. At a jump time of the counting
process $N^{(i)}$, with corresponding jump-size $z$, the jump of
$\mathbb M$ is
$
    \xi_i(z)
    :=
    \left(
        0,\ldots,0,1,0,\ldots,0,
        z-\frac{k}{\lambda}
    \right)^\top,
$
where the entry $1$ occupies the $i$-th coordinate. Hence,
$
    \|\xi_i(z)\|^2
    =
    1+
    \left(
        z-\frac{k}{\lambda}
    \right)^2.
$
For every $\varepsilon>0$, the Lindeberg term associated with the
normalization $\sqrt T\,I_{n+1}$ is therefore
\begin{align*}
    \Delta_T^\varepsilon
    &:=
    \frac{1}{T}
    \sum_{i=1}^n
    \int_0^T\int_0^\infty
    \|\xi_i(z)\|^2
    \ind_{\{\|\xi_i(z)\|>\varepsilon\sqrt T\}}
    \nu_i(\dd t,\dd z)\\
    &=
    S_C
    \int_0^\infty
    \left[
        1+
        \left(
            z-\frac{k}{\lambda}
        \right)^2
    \right]
    \ind_{
        \left\{
            1+(z-k/\lambda)^2>\varepsilon^2T
        \right\}
    }
    q_\lambda^{(k)}(z)\dd z.
\end{align*}
The Gamma distribution has a finite second moment, and hence
$
    \int_0^\infty
    \left[
        1+
        \left(
            z-\frac{k}{\lambda}
        \right)^2
    \right]
    q_\lambda^{(k)}(z)\dd z
    <
    \infty.
$
Moreover, the indicator in the preceding integral converges pointwise to
zero as $T\to\infty$. The dominated convergence theorem therefore gives
$
    \Delta_T^\varepsilon
    \xrightarrow[T\to\infty]{\mathrm{a.s.}}
    0.
$
The martingale $\mathbb M$ is purely discontinuous and square-integrable.
Moreover, its jump times are totally inaccessible. Hence, for every
predictable stopping time $\tau$,
$
    \Delta\mathbb M_\tau=0
    \;
    \text{a.s. on }\{\tau<\infty\},
$
so that $\mathbb M$ is quasi-left-continuous.
Thus, the bracket and Lindeberg conditions of
Theorem~\ref{GCLTtouati} in Appendix~B are satisfied. Applying this
theorem with normalization $\sqrt T\,I_{n+1}$ yields, stably with respect
to $\mathcal G_C$,
\begin{align}
    &\left(
        \frac{M_{1,T}}{\sqrt T},
        \ldots,
        \frac{M_{n,T}}{\sqrt T},
        \frac{R_T^{(k)}}{\sqrt T}
    \right)
    \notag\\
    &\qquad
    \xrightarrow[T\to\infty]{\mathcal D,\mathrm{st}}
    \left(
        \sqrt{C_1}Z_1,
        \ldots,
        \sqrt{C_n}Z_n,
        \frac{\sqrt{kS_C}}{\lambda}Z_R
    \right),
    \label{eq:joint-martingale-stable}
\end{align}
where $Z_1,\ldots,Z_n,Z_R$ are independent standard normal random
variables, defined on an extension of the original probability space and
independent of $\mathcal G_C$.
For every $i\in\{1,\ldots,n\}$,
$
    \sqrt T
    \bigl(
        \widehat C_{i,T}-C_i
    \bigr)
    =
    \frac{M_{i,T}}{\sqrt T}.
$
Furthermore, by \eqref{eq:aggregate-count-slln},
$
    N_{n,T}
    \xrightarrow[T\to\infty]{\mathrm{a.s.}}
    \infty.
$
Thus, almost surely, $N_{n,T}>0$ and $J_{n,T}>0$ for all sufficiently
large $T$. For such $T$,
$
    \widehat\lambda_{n,T}^{(k)}-\lambda
    =
    -\lambda\frac{R_T^{(k)}}{J_{n,T}},
$
and hence
$
    \sqrt T
    \bigl(
        \widehat\lambda_{n,T}^{(k)}-\lambda
    \bigr)
    =
    -\lambda
    \frac{T}{J_{n,T}}
    \frac{R_T^{(k)}}{\sqrt T}.
$
By \eqref{eq:aggregate-jump-slln},
$
    \frac{J_{n,T}}{T}
    \xrightarrow[T\to\infty]{\mathrm{a.s.}}
    \frac{k}{\lambda}S_C.
$
Since $S_C>0$ almost surely, it follows that
$
    \frac{T}{J_{n,T}}
    \xrightarrow[T\to\infty]{\mathrm{a.s.}}
    \frac{\lambda}{kS_C}.
$
The stable version of Slutsky's theorem, applied jointly with
\eqref{eq:joint-martingale-stable}, now gives
\begin{align*}
    &\left(
        \sqrt T\bigl(\widehat C_{1,T}-C_1\bigr),
        \ldots,
        \sqrt T\bigl(\widehat C_{n,T}-C_n\bigr),
        \sqrt T
        \bigl(\widehat\lambda_{n,T}^{(k)}-\lambda\bigr)
    \right)\\
    &\qquad
    \xrightarrow[T\to\infty]{\mathcal D,\mathrm{st}}
    \left(
        \sqrt{C_1}Z_1,
        \ldots,
        \sqrt{C_n}Z_n,
        -\frac{\lambda}{\sqrt{kS_C}}Z_R
    \right).
\end{align*}
Setting
$
    Z_\lambda
    :=
    -Z_R,
$
we obtain \eqref{eq:joint-Chat-lambda-stable}. Since
$Z_\lambda=-Z_R$, the random variables
$Z_1,\ldots,Z_n,Z_\lambda$ remain independent standard normal random
variables and are independent of $\mathcal G_C$.
Finally, conditionally on $\mathcal G_C$, the limiting vector is Gaussian
with diagonal covariance matrix
$
    \operatorname{diag}
    \left(
        C_1,\ldots,C_n,
        \frac{\lambda^2}{kS_C}
    \right).
$
Its components are therefore conditionally independent. This proves the
conditional first-order asymptotic independence.
\end{proof}

\subsection{Special cases and a Beta benchmark}
\label{subsec:special-cases-beta-benchmark}

\paragraph{Exponential and Erlang special cases.}

The following corollary records the exponential and Erlang special cases
of the Gamma$(k,\lambda)$ family. For each $k\in\{1,2\}$, the density
$q_\lambda^{(k)}$ defined in \eqref{eq:gamma-k-density} is supported on
$(0,\infty)$ and satisfies the first-moment condition
\eqref{eq:q-first-moment}. The two cases follow directly by setting $k=1$
and $k=2$, respectively, in the general formulas established above.

\begin{coro}
\label{cor:exponential-jumps}
Under Assumptions~\ref{ass:population} and~\ref{ass:support}, fix
$n\geq1$, and let
$ a\in\real_{++},
    \;
    b\in\real_{++},
    \;
    \lambda\in\real_{++}.
$
\begin{enumerate}
    \item For $k=1$, the jump-size density is the exponential density
    $
        q_\lambda^{(1)}(z)
        =
        \lambda e^{-\lambda z}
        \ind_{(0,\infty)}(z).
    $
    On the event $\{N_{n,T}>0\}$, the pooled Gamma-rate MLE is
    $        \widehat\lambda_{n,T}^{(1)}
        =
        \frac{N_{n,T}}{J_{n,T}},
    $
    and
    $        \sqrt{N_{n,T}}
        \left(
            \widehat\lambda_{n,T}^{(1)}-\lambda
        \right)
        \xrightarrow[T\to\infty]{\mathcal D}
        \mathcal N(0,\lambda^2).
    $
    Moreover, for every $i\in\{1,\ldots,n\}$ and every
    $c\in(0,1]$, conditionally on $C_i=c$, the immigration mechanism is
    $
        F_c^{(1)}(u)
        =
        au
        +
        c\left(
            \frac{\lambda}{\lambda-u}-1
        \right),
        \;
        u<\lambda,
    $
    and the conditional stationary mean is
    $
        m_\infty^{(1)}(c)
        =
        \frac{a+c/\lambda}{b}.
    $
    \item For $k=2$, the jump-size density is the Erlang density of order
    two,
    $
        q_\lambda^{(2)}(z)
        =
        \lambda^2z e^{-\lambda z}
        \ind_{(0,\infty)}(z).
    $
    On the event $\{N_{n,T}>0\}$, the pooled Gamma-rate MLE is
    $        \widehat\lambda_{n,T}^{(2)}
        =
        \frac{2N_{n,T}}{J_{n,T}},
    $
    and
    $
        \sqrt{N_{n,T}}
        \left(         \widehat\lambda_{n,T}^{(2)}-\lambda
        \right)
        \xrightarrow[T\to\infty]{\mathcal D}
        \mathcal N
        \left(
            0,\frac{\lambda^2}{2}
        \right).
    $
    Moreover, for every $i\in\{1,\ldots,n\}$ and every
    $c\in(0,1]$, conditionally on $C_i=c$, the immigration mechanism is
    $
        F_c^{(2)}(u)
        =
        au
        +
        c\left[
            \left(
                \frac{\lambda}{\lambda-u}
            \right)^2
            -
            1
        \right],
        \;
        u<\lambda,
    $
    and the conditional stationary mean is
    $
        m_\infty^{(2)}(c)
        =
        \frac{a+2c/\lambda}{b}.
    $
\end{enumerate}
\end{coro}

\paragraph{A Beta$(2,2)$ benchmark.}
To illustrate the Bernstein density-estimation results explicitly, we
consider a Beta$(2,2)$ distribution for the random effects on the
canonical support $[0,1]$:
\begin{equation}
    f_C(c)
    =
    6c(1-c)\mathds{1}_{[0,1]}(c).
    \label{eq:beta22-rescaled-C}
\end{equation}
The constants appearing in Theorem~\ref{thm:bias-var-mise} are therefore
$
    A_{f_C}
    :=
    \int_0^1 f_C(u)\psi(u)\dd u
    =
    \frac{3\sqrt{\pi}}{8},
$
and
$
    B_{f_C}
    :=
    \frac14
    \int_0^1
    \left[
        (1-2u)f_C'(u)
        +
        u(1-u)f_C''(u)
    \right]^2
    \dd u
    =
    \frac95.
$
Consequently, minimizing the continuous relaxation of the leading MISE
expression gives
\begin{equation}
    m_{\mathrm{opt}}
    :=
    \left(
        \frac{4B_{f_C}}{A_{f_C}}
    \right)^{2/5}
    n^{2/5}=
    \left(
        \frac{96}{5\sqrt{\pi}}
    \right)^{2/5}
    n^{2/5}.
\label{eq:mopt-beta22}
\end{equation}
Since the Bernstein order must be a positive integer, one may choose $m$
as the nearest positive integer to $m_{\mathrm{opt}}$.
This distribution provides a natural numerical benchmark. One may
simulate $n$ independent BAJD trajectories with random intensities $C_i$
drawn from \eqref{eq:beta22-rescaled-C}, compute $\widehat C_{i,T}$ and
its projected version $\widehat C_{i,T}^{\mathcal C}$, estimate the common
Gamma rate through $\widehat\lambda_{n,T}^{(k)}$, and construct
$\widehat f_{C,m,n,T}$. The resulting density estimate can then be
compared with the true density $f_C$.
The same benchmark can be implemented for both $k=1$ and $k=2$. It
therefore provides a direct numerical illustration of the three
components of the proposed procedure: estimation of the individual
random intensities, Bernstein estimation of their common density, and
pooled calibration of the Gamma jump-size parameter. This benchmark is
examined in the following numerical section.
\section{Numerical study and empirical illustration}
\label{sec:numerical-illustration}

This section provides a numerical illustration and an empirical
application of the estimation procedure developed in
Sections~\ref{sec:general-framework} and~\ref{sec:applications}.
We first simulate the jump counts and total jump amplitudes associated
with independent BAJD trajectories, compute the first-stage estimators
$\widehat C_{i,T}$, estimate the common Gamma jump-size rate
$\widehat\lambda_{n,T}^{(k)}$, and reconstruct the random-effect density
using the feasible Bernstein estimator. Representative discretized BAJD
trajectories are generated separately for graphical illustration.

We then present a reproducible empirical application based on financial
realized-volatility proxies. While the theoretical framework is formulated
under continuous-time observation, its empirical implementation requires
a discrete-time adaptation because financial data are observed at discrete
sampling times. We therefore adopt a discrete observation scheme and
replace the exact jump statistics available under continuous observation
by empirical counterparts constructed from discrete increments. Large
positive increments are identified as jump events, from which the
individual jump frequencies and the corresponding Bernstein density
estimate are computed. This adaptation makes it possible to implement the
proposed two-stage estimation procedure on a panel of positive financial
time series while accounting for the discrete nature of the available
data. The empirical implementation should therefore be interpreted as a
discrete-time adaptation of the continuous-observation methodology.

The simulated density reconstruction below uses a large panel and a long
observation horizon. This choice is motivated by the fact that the feasible
Bernstein estimator is constructed from
\[
    \widehat C_{i,T}
    =
    \frac{N_T^{(i)}}{T},
\]
rather than from the latent random effects $C_i$. A long observation
horizon reduces the first-stage estimation error and ensures that the
simulation results primarily illustrate the boundary behavior of the
Bernstein density reconstruction, rather than the noise arising from the
estimation of the individual random effects.

\subsection{Simulation design}
\label{subsec:numerical-design}

For numerical readability, we retain the canonical Beta$(2,2)$ benchmark
introduced in Subsection~\ref{subsec:special-cases-beta-benchmark}, but
apply an affine transformation mapping its support from $[0,1]$ onto an
interior interval $[c_-,c_+]\subset[0,1]$. More precisely, we generate
$
    U_i\sim\operatorname{Beta}(2,2),
    \;
    C_i=c_-+(c_+-c_-)U_i,
    \;
    i=1,\ldots,n.
$
Thus, $U_i$ follows the canonical Beta$(2,2)$ distribution, whereas $C_i$
follows its affine rescaling to $[c_-,c_+]$. Equivalently,
$
    U_i
    =
    \frac{C_i-c_-}{c_+-c_-}
    \in[0,1].
$
The corresponding density of $C_i$ is
$
    f_C(c)
    =
    \frac{6(c-c_-)(c_+-c)}
         {(c_+-c_-)^3}
    \mathds{1}_{[c_-,c_+]}(c).
$
For each $k\in\{1,2\}$, the jump sizes follow the Gamma$(k,\lambda)$
density \eqref{eq:gamma-k-density}, while, conditionally on $C_i$, the
counting process $N^{(i)}$ is a Poisson process with intensity $C_i$.
Since the continuous-observation framework makes the jump times and jump
sizes exactly observable,
the estimation experiment can be conducted directly at the level of the
sufficient jump statistics. Conditionally on $C_i$,
$
    N_T^{(i)}
    \sim
    \operatorname{Poisson}(C_iT).
$
Moreover, conditionally on $N_T^{(i)}$,
$
    J_T^{(i)}
    \sim
    \operatorname{Gamma}
    \left(
        kN_T^{(i)},\lambda
    \right),
$
where $\lambda$ is the rate parameter and $J_T^{(i)}=0$ when
$N_T^{(i)}=0$.
The corresponding estimators are
$
    \widehat C_{i,T}
    =
    \frac{N_T^{(i)}}{T},
    \;
    \widehat\lambda_{n,T}^{(k)}
    =
    \frac{
        k\sum_{i=1}^nN_T^{(i)}
    }{
        \sum_{i=1}^nJ_T^{(i)}
    }.
$
For the Bernstein reconstruction, the first-stage estimates are normalized
according to
\begin{equation}
    \widehat U_{i,T}
    =
    \frac{\widehat C_{i,T}-c_-}{c_+-c_-},
    \;
    \widehat U_{i,T}^{\mathcal C}
    =
    \Pi_{[0,1]}
    \bigl(\widehat U_{i,T}\bigr).
    \label{eq:simulation-affine-Uhat}
\end{equation}
Thus, values of $\widehat C_{i,T}$ that fall outside $[c_-,c_+]$ because
of finite-horizon Poisson noise are projected onto the normalized interval
$[0,1]$.
To ensure that observations projected onto the endpoints are retained, we
partition $[0,1]$ into the subintervals
$$
    I_{j,m}
    :=
    \begin{cases}
        \left[
            \dfrac{j}{m},
            \dfrac{j+1}{m}
        \right),
        & j=0,\ldots,m-2,\\
        \left[
            \dfrac{m-1}{m},
            1
        \right],
        & j=m-1,
    \end{cases}
$$
and define the corresponding empirical frequencies by
$
    \widehat\pi_{j,m,n,T}
    :=
    \frac{1}{n}
    \sum_{i=1}^n
    \mathds{1}_{
        \{\widehat U_{i,T}^{\mathcal C}\in I_{j,m}\}
    },
    \;
    j=0,\ldots,m-1.
$
The feasible Bernstein estimator on the normalized scale is then
$
    \widehat f_{U,m,n,T}(u)
    =
    m\sum_{j=0}^{m-1}
    \widehat\pi_{j,m,n,T}
    p_j(m-1,u),
    \;
    u\in[0,1].
$
In particular,
$
    \int_0^1
    \widehat f_{U,m,n,T}(u)\dd u
    =
    1,
$
because
$
    \sum_{j=0}^{m-1}
    \widehat\pi_{j,m,n,T}
    =
    1$
    and $
    \int_0^1p_j(m-1,u)\dd u
    =
    \frac{1}{m}.
$
Although the Bernstein estimator is first constructed on the normalized
scale $[0,1]$, the density reported in the numerical results and displayed
in the corresponding figure is transformed back to the original intensity
scale $[c_-,c_+]$:
\begin{equation}
    \widehat f_{C,m,n,T}(c)
    =
    \frac{1}{c_+-c_-}
    \widehat f_{U,m,n,T}
    \left(
        \frac{c-c_-}{c_+-c_-}
    \right),
    \;
    c\in[c_-,c_+],
    \label{eq:simulation-density-backtransform}
\end{equation}
and
$
    \widehat f_{C,m,n,T}(c)
    =
    0,
    \;
    c\notin[c_-,c_+].
$
Hence, the affine normalization is used to construct the Bernstein density
estimator on the canonical scale $[0,1]$, as described in
Remark~\ref{rem:support-normalization}, whereas all density comparisons
are performed on the original intensity scale $[c_-,c_+]$. In particular,
the Bernstein regularity conditions and the choice of the polynomial order
are applied to the normalized density
$
    g(u)
    =
    6u(1-u),
    \;
    u\in[0,1].
$
For the density reconstruction, we use
$
    n=10000,
    \;
    T=50000,
    \;
    \lambda=4,
    \;
    c_-=0.10,
    \;
    c_+=0.50,
$
and
$
    m
    =
    \left\lfloor
        \left(
            \frac{96}{5\sqrt{\pi}}
        \right)^{2/5}
        n^{2/5}
        +
        \frac{1}{2}
    \right\rfloor
    =
    103.
$
The value $m=103$ is obtained from \eqref{eq:mopt-beta22} using the
canonical density $g(u)=6u(1-u)$ of the normalized random effect $U$.
The same order is used to construct $\widehat f_{U,m,n,T}$ before applying
the Jacobian correction \eqref{eq:simulation-density-backtransform}, since
an affine change of scale does not alter the Bernstein polynomial order.
The large value of $T$ is used to reduce the first-stage estimation error
and make the boundary behavior of the feasible Bernstein estimator
visually clear.

The illustrative trajectories are generated separately using the following
projected Euler scheme:
\[
    Y_{t+\Delta}^{(i)}
    =
    \left[
        Y_t^{(i)}
        +
        (a-bY_t^{(i)})\Delta
        +
        \sigma\sqrt{(Y_t^{(i)})^+}\,
        \Delta W_t^{(i)}
        +
        \sum_{\ell=1}^{L_{i,t}}
        Z_{\ell,t}^{(i)}
    \right]^+,
\]
where
$
    \Delta W_t^{(i)}
    \sim
    \mathcal N(0,\Delta),
    \;
    L_{i,t}
    \sim
    \operatorname{Poisson}(C_i\Delta),
    \;
    Z_{\ell,t}^{(i)}
    \sim
    \operatorname{Gamma}(k,\lambda),
$
and
$
    x^+
    :=
    \max\{x,0\}.
$
Conditionally on the random intensities, the Brownian increments, jump
counts, and jump sizes are mutually independent and independent across
time steps and trajectories. For these illustrative trajectories, we use
$
    a=1,
    \;
    b=1.2,
    \;
    \sigma=0.35,
    \;
    Y_0=1,
    \;
    \Delta=0.05.
$
These discretized trajectories are included only to provide a visual
comparison of the two jump-size mechanisms.

\subsection{Numerical results}

Table~\ref{tab:section7-numerical-summary} reports the results of the
large-sample density-reconstruction experiment generated using the fixed
random-number seed \texttt{20260829}. The pooled estimator
$\widehat\lambda_{n,T}^{(k)}$ is close to the true value $\lambda=4$ for
both jump-size specifications. The RMSE of the first-stage estimators
$\widehat C_{i,T}$ is small, and the integrated squared error of the
feasible Bernstein estimator is of order $10^{-3}$.

For each jump-size specification, the ISE is evaluated on the original
intensity interval $[0.1,0.5]$ using the trapezoidal rule on an equally
spaced grid of $G=5000$ points. More precisely, let
$
    c_\ell
    :=
    0.1+
    \frac{\ell-1}{G-1}(0.5-0.1),
    \;
    \ell=1,\ldots,G,
$
and
$
    e_\ell
    :=
    \widehat f_{C,m,n,T}(c_\ell)-f_C(c_\ell).
$
The numerical approximation of the ISE is then
\begin{equation}
    \operatorname{ISE}
    \bigl(\widehat f_{C,m,n,T}\bigr)
    \approx
    \sum_{\ell=1}^{G-1}
    \frac{c_{\ell+1}-c_\ell}{2}
    \left(
        e_\ell^2+e_{\ell+1}^2
    \right).
    \label{eq:numerical-ise}
\end{equation}
Doubling the grid size from $G=5000$ to $G=10\,000$ does not change the
reported ISE values at the displayed precision, confirming the numerical
stability of the approximation.

\begin{table}[H]
\centering
\caption{Large-sample density-reconstruction benchmark for the two
Gamma$(k,\lambda)$ jump-size models, generated with
$n=10\,000$, $T=50\,000$, $m=103$, $\lambda=4$, and random-number seed
\texttt{20260829}. The ISE is computed after back-transformation on the
original intensity interval $[0.1,0.5]$, using the trapezoidal rule on an
equally spaced grid of $G=5000$ points.}
\label{tab:section7-numerical-summary}
\begin{tabular}{cccccc}
\toprule
$k$
&
$\widehat\lambda_{n,T}^{(k)}$
&
$N_{n,T}$
&
$n^{-1}\sum_{i=1}^{n}\widehat C_{i,T}$
&
$\operatorname{RMSE}(\widehat C_{i,T})$
&
$\operatorname{ISE}(\widehat f_{C,m,n,T})$
\\
\midrule
$1$
&
$4.0002$
&
$151256033$
&
$0.3025$
&
$0.002446$
&
$0.002678$
\\
$2$
&
$4.0002$
&
$151249607$
&
$0.3025$
&
$0.002446$
&
$0.002653$
\\
\bottomrule
\end{tabular}
\end{table}

Figure~\ref{fig:section7-density-comparison} compares, on the original
intensity interval $[c_-,c_+]$, the true density $f_C$ with the
back-transformed feasible Bernstein estimator
$\widehat f_{C,m,n,T}$ defined in
\eqref{eq:simulation-density-backtransform}. The estimator is constructed
from the projected normalized estimates
$\widehat U_{i,T}^{\mathcal C}$ defined in
\eqref{eq:simulation-affine-Uhat}. It successfully recovers the global
shape of the true density and exhibits satisfactory behavior near both
endpoints of the compact support in the considered simulation setting.

\begin{figure}[H]
\centering
\includegraphics[width=0.98\textwidth]
{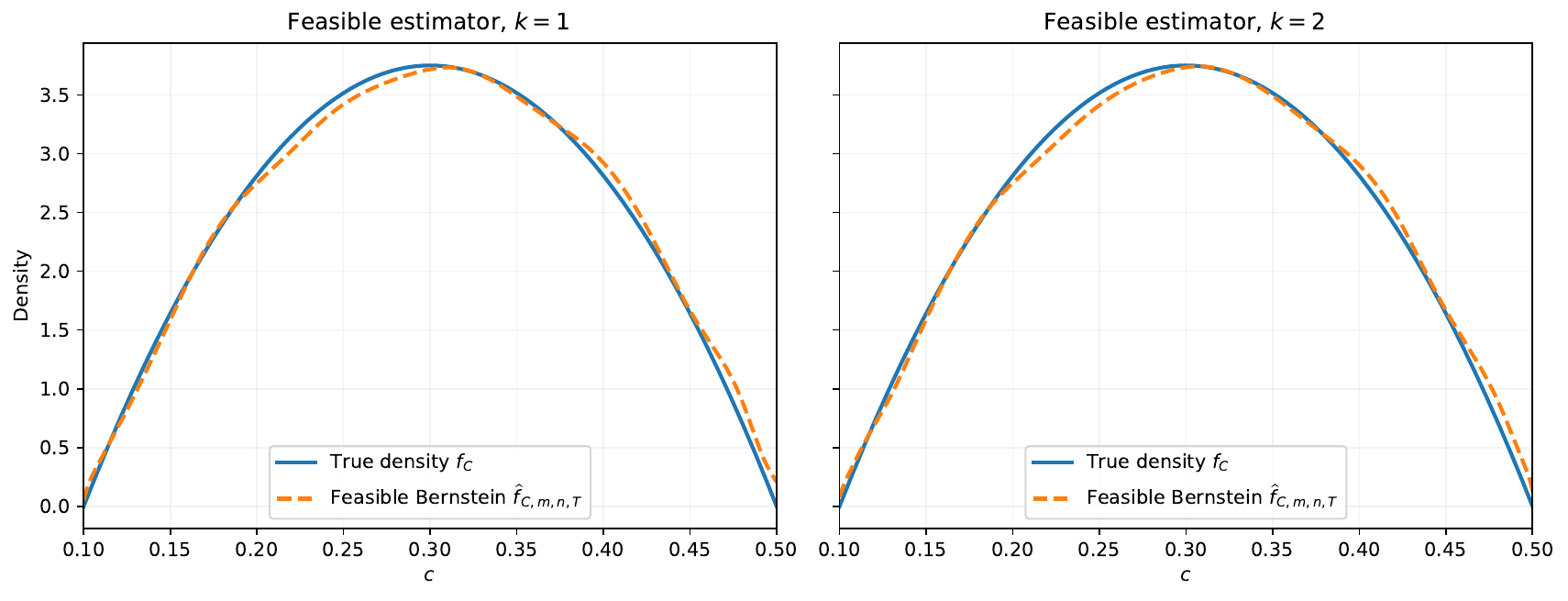}
\caption{True density $f_C$ and back-transformed feasible Bernstein
estimator $\widehat f_{C,m,n,T}$ on the original intensity scale
$[c_-,c_+]=[0.1,0.5]$, constructed from the projected normalized estimates
$\widehat U_{i,T}^{\mathcal C}$, for the two jump-size specifications
$k=1$ and $k=2$.}
\label{fig:section7-density-comparison}
\end{figure}
Figure~\ref{fig:section7-chat-scatter} illustrates the first-stage
estimation error over a finite observation horizon.
It was
generated separately from the large-sample density-reconstruction
experiment, using
$
    n=500,
    \;
    T=100,
$
and the fixed random-number seed \texttt{20260829}. The same realization
of $C_1,\ldots,C_n$ is used in both panels, whereas the conditional
Poisson counts are generated independently for the two jump-size
specifications. The points are concentrated around the diagonal
$\widehat C_{i,T}=C_i$, with visible finite-horizon dispersion. The
moderate values of $n$ and $T$ make this dispersion visually apparent.

\begin{figure}[H]
\centering
\begin{minipage}{0.48\textwidth}
\centering
\includegraphics[width=\textwidth]{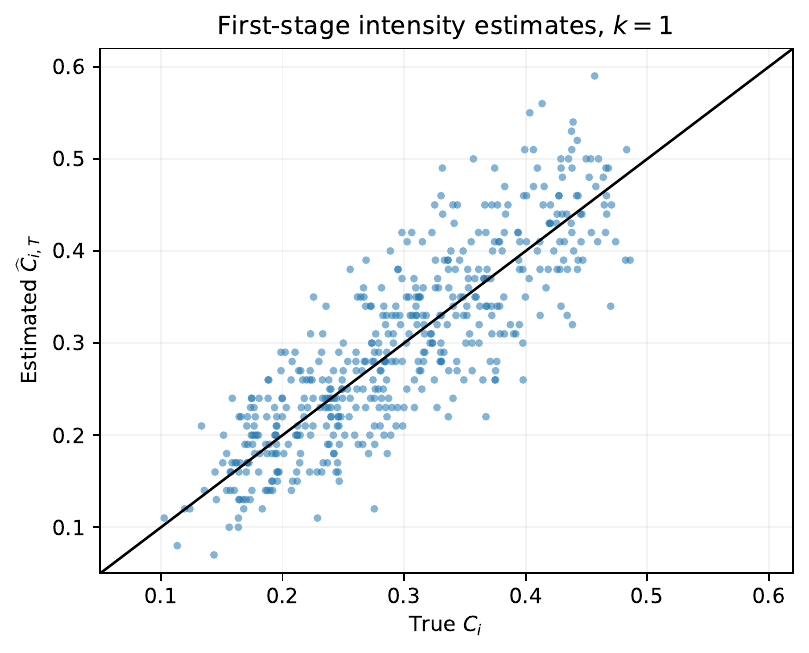}\\[-1mm]
\small (a) Exponential jumps, $k=1$.
\end{minipage}
\hfill
\begin{minipage}{0.48\textwidth}
\centering
\includegraphics[width=\textwidth]{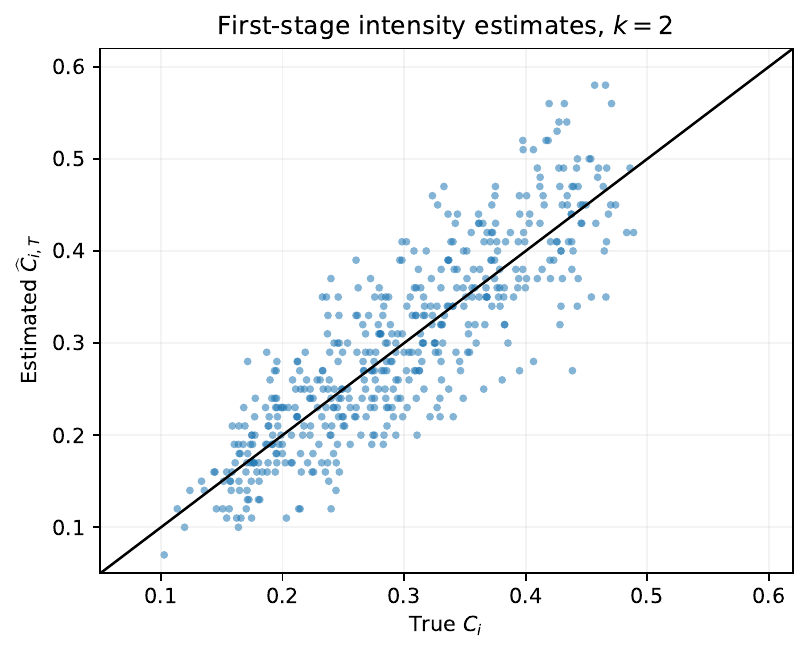}\\[-1mm]
\small (b) Erlang jumps, $k=2$.
\end{minipage}
\caption{Scatter plots of the individual estimators
$\widehat C_{i,T}$ against the true random effects $C_i$, generated with
$n=500$, $T=100$, and random-number seed \texttt{20260829}. The same
latent random effects are used in both panels, whereas the conditional
Poisson counts are generated independently. The diagonal line represents
the identity relation $\widehat C_{i,T}=C_i$.}
\label{fig:section7-chat-scatter}
\end{figure}

The first-stage estimator and the Bernstein reconstruction depend only on
the jump counts and therefore do not depend on the jump-size specification
$k$. The small differences between the results obtained for $k=1$ and
$k=2$ arise solely from the independent generation of the conditional
Poisson counts in the two simulation experiments.

Figure~\ref{fig:section7-sample-paths} displays one representative
simulated BAJD trajectory for each jump-size specification. Panel~(a)
corresponds to exponential jumps with $k=1$, whereas panel~(b) corresponds
to Erlang jumps with $k=2$. To facilitate visual comparison, the two
trajectories are generated using the same intensity $C_i=0.3$, the same
Brownian increments, and the same jump-arrival counts at each time step;
only the jump sizes are generated from different distributions. This
coupling is used solely for graphical comparison and does not alter the
marginal BAJD dynamics represented in either panel.

\begin{figure}[H]
\centering
\begin{minipage}{0.48\textwidth}
\centering
\includegraphics[width=\textwidth]{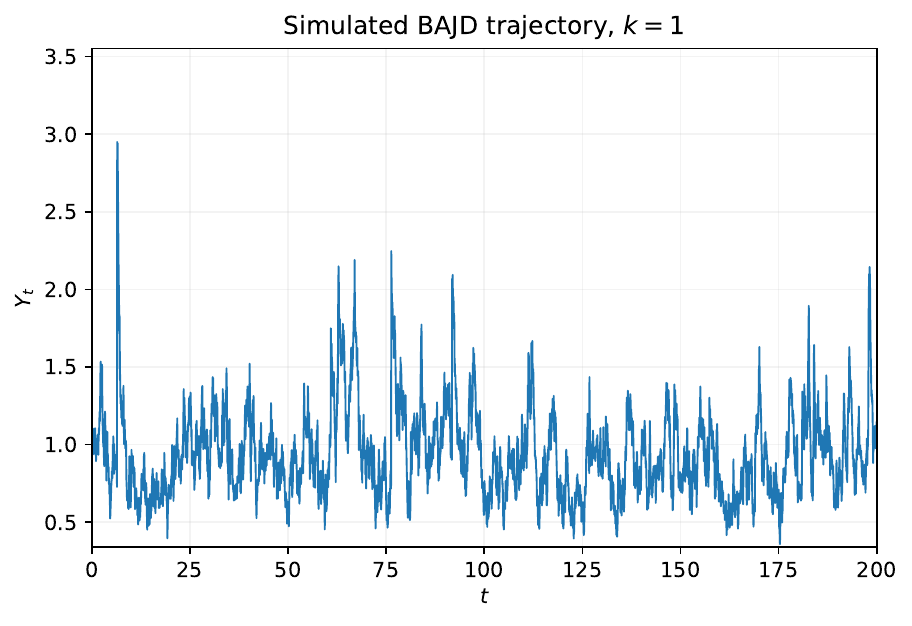}\\[-1mm]
\small (a) Exponential jumps, $k=1$.
\end{minipage}
\hfill
\begin{minipage}{0.48\textwidth}
\centering
\includegraphics[width=\textwidth]{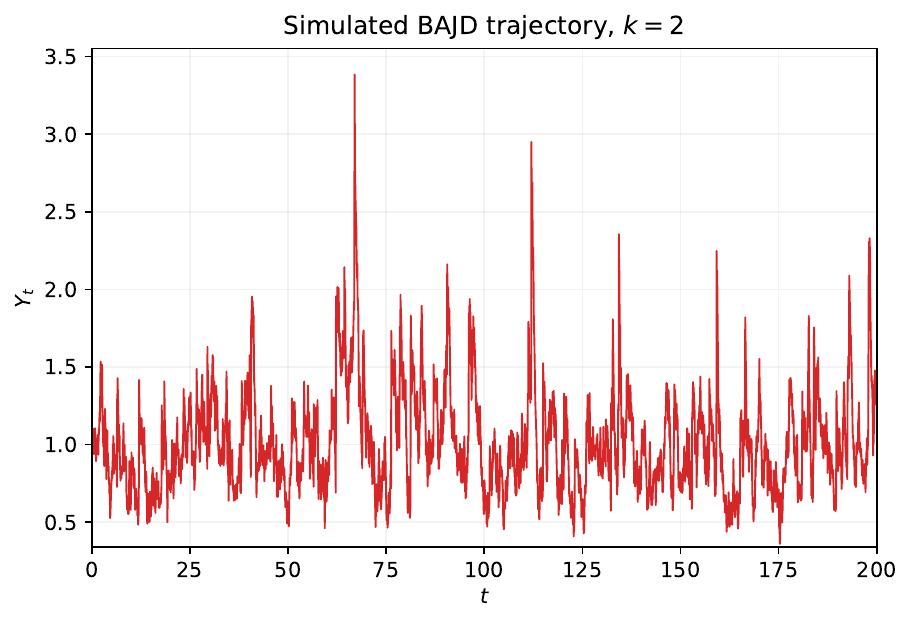}\\[-1mm]
\small (b) Erlang jumps, $k=2$.
\end{minipage}
\caption{One representative simulated BAJD trajectory for each jump-size
specification: exponential jumps with $k=1$ in panel~(a) and Erlang
jumps with $k=2$ in panel~(b). Both trajectories are generated with
$a=1$, $b=1.2$, $\sigma=0.35$, $\lambda=4$, $C_i=0.3$, $Y_0=1$,
$\Delta=0.05$, $T=200$, and random-number seed \texttt{20260829}. The
same Brownian increments and jump-arrival counts are used in both panels
to facilitate comparison of the two jump-size mechanisms.}
\label{fig:section7-sample-paths}
\end{figure}

Taken together, the numerical results illustrate the complete estimation
procedure proposed in the paper. First, the jump counts provide accurate
estimates of the heterogeneous intensities through
$\widehat C_{i,T}=N_T^{(i)}/T$. Second, the pooled jump amplitudes yield
estimates of the common rate parameter
$\widehat\lambda_{n,T}^{(k)}$ that are close to the true value
$\lambda=4$. Third, the feasible Bernstein estimator accurately recovers
the population density $f_C$ under both jump-size specifications and
exhibits satisfactory behavior near the endpoints of the compact support.
\subsection{Empirical illustration on financial realized-volatility data}
\label{subsec:empirical-real-data}
\paragraph{Motivation and data source.}
Finally, we illustrate how the proposed two-stage procedure can be adapted
to real financial data. In the theoretical BAJD random-effects model,
continuous-time observation makes the jumps of each trajectory and the
corresponding count $N_T^{(i)}$ directly observable. For the empirical
application, we replace these exact jump statistics by empirical
counterparts constructed from discretely observed data. More precisely,
for each asset, we construct a positive realized-volatility proxy and
identify its large positive increments as empirical upward volatility
jumps. This allows us to implement the same statistical procedure on a
real-data panel: first estimate the individual jump frequencies and then
reconstruct their cross-sectional density using Bernstein polynomials.

The data consist of daily adjusted closing prices obtained from Yahoo
Finance using the Python package \texttt{yfinance}. They were downloaded
on 27 August 2026. The download request specifies 1 January 2010 as the
start date and 27 August 2026 as the exclusive end date. Consequently,
the last observation included in the sample corresponds to the trading
day of 26 August 2026. These fixed dates replace a moving
``last available trading day'' rule and ensure that the sample period is
precisely defined.

We consider a mixed panel comprising broad-market exchange-traded funds
(ETFs), sector ETFs, and large-cap U.S. equities. Adjusted closing prices
are used to account for standard corporate actions, such as stock splits
and dividend distributions, when computing returns.

\begin{table}[H]
\centering
\scriptsize
\caption{Financial instruments included in the empirical panel. The panel
comprises broad-market ETFs, sector ETFs, and large-cap U.S. equities.}
\label{tab:realdata-tickers}
\begin{tabular}{ll|ll}
\toprule
Ticker & Instrument & Ticker & Instrument \\
\midrule
SPY   & SPDR S\&P 500 ETF Trust
      & QQQ   & Invesco QQQ Trust, NASDAQ-100 ETF \\
IWM   & iShares Russell 2000 ETF
      & DIA   & SPDR Dow Jones Industrial Average ETF Trust \\
XLK   & Technology Select Sector SPDR Fund
      & XLF   & Financial Select Sector SPDR Fund \\
XLE   & Energy Select Sector SPDR Fund
      & XLV   & Health Care Select Sector SPDR Fund \\
XLI   & Industrial Select Sector SPDR Fund
      & XLP   & Consumer Staples Select Sector SPDR Fund \\
XLU   & Utilities Select Sector SPDR Fund
      & XLY   & Consumer Discretionary Select Sector SPDR Fund \\
XLB   & Materials Select Sector SPDR Fund
      & AAPL  & Apple Inc. \\
MSFT  & Microsoft Corporation
      & AMZN  & Amazon.com, Inc. \\
GOOGL & Alphabet Inc., Class A
      & META  & Meta Platforms, Inc. \\
NVDA  & NVIDIA Corporation
      & JPM   & JPMorgan Chase \& Co. \\
XOM   & Exxon Mobil Corporation
      &       & \\
\bottomrule
\end{tabular}
\end{table}
\paragraph{Construction of empirical BAJD-type trajectories.}

Let $S_t^{(i)}$ denote the adjusted closing price of asset $i$ on trading
day $t$. We first compute the daily log-return
$
    r_t^{(i)}
    :=
    \log S_t^{(i)}-\log S_{t-1}^{(i)}.
$
The positive process to which the BAJD methodology is applied is the
rolling realized-variance proxy
$
    Y_t^{(i)}
    :=
    \sum_{j=t-h+1}^{t}
    \left(r_j^{(i)}\right)^2,
    \;
    h=21,
$
where $h=21$ corresponds approximately to one trading month. This proxy
is nonnegative and persistent and may exhibit abrupt upward movements
during periods of market stress. It can therefore be viewed as an
empirical analogue of a positive mean-reverting process with upward
jumps.
Because consecutive rolling windows overlap, the increment of the proxy
admits the exact representation
\begin{equation}
    \Delta Y_t^{(i)}
    :=
    Y_t^{(i)}-Y_{t-1}^{(i)}
    =
    \left(r_t^{(i)}\right)^2
    -
    \left(r_{t-h}^{(i)}\right)^2.
    \label{eq:realdata-overlapping-increment}
\end{equation}
Consequently, these increments should not generally be regarded as
independent observations. In particular, each squared return enters once
with a positive sign and once with a negative sign in two increments
separated by $h$ trading days. In addition, the volatility dependence
present in financial returns may induce further serial dependence in the
increment sequence.

\paragraph{Empirical jump extraction.}

For each asset $i$, we form the increments
$
    \Delta Y_t^{(i)}
    :=
    Y_t^{(i)}-Y_{t-1}^{(i)}.
$
Since the jump component of the BAJD model is positive, only positive
increments are considered as jump candidates. We define the
asset-specific detection threshold by
\begin{equation}
    \vartheta_i
    :=
    \operatorname{Quantile}_{0.95}
    \left(
        \left\{
            \Delta Y_t^{(i)}
            :
            \Delta Y_t^{(i)}>0
        \right\}
    \right).
    \label{eq:realdata-threshold}
\end{equation}
The empirical jump count and cumulative jump amplitude are then defined by
$
    N_{T_i}^{(i)}
    :=
    \sum_t
    \mathds{1}_{
        \{\Delta Y_t^{(i)}>\vartheta_i\}
    },
    \;
    J_{T_i}^{(i)}
    :=
    \sum_t
    \Delta Y_t^{(i)}
    \mathds{1}_{
        \{\Delta Y_t^{(i)}>\vartheta_i\}
    }.
$
Except when several increments are equal to the estimated threshold, this
asset-specific rule retains approximately $5\%$ of the positive increments
of each asset. Consequently, cross-sectional differences in the resulting
counts reflect, to a substantial extent, differences in the numbers of
positive increments available for the different assets.
The estimated frequencies should therefore be
interpreted as frequencies of threshold exceedances rather than as
unrestricted estimates of latent Poisson intensities.
Let $L_i$ denote the number of usable daily increments for asset $i$.
Using the standard convention of $252$ trading days per year, we express
the observation horizon in trading years as
$
    T_i
    :=
    \frac{L_i}{252}.
$
The empirical counterpart of the first-stage intensity estimator,
expressed as the estimated number of detected jump events per trading
year, is
\begin{equation}
    \widehat C_i
    :=
    \frac{N_{T_i}^{(i)}}{T_i}
    =
    \frac{252\,N_{T_i}^{(i)}}{L_i}.
    \label{eq:realdata-Chat}
\end{equation}
Thus, $\widehat C_i$ represents the annualized frequency of the detected
jump events.

This estimator has the same functional form as the estimator in
Proposition~\ref{prop:mle-Ci}, except that the continuously observed jump
count is replaced by a threshold-based empirical count. The corresponding
empirical version of the pooled Gamma-rate estimator would be
\begin{equation}
    \widehat\lambda_{n,\boldsymbol T}^{(k)}
    :=
    \frac{
        k\sum_{i=1}^nN_{T_i}^{(i)}
    }{
        \sum_{i=1}^nJ_{T_i}^{(i)}
    },
    \;
    \boldsymbol T:=(T_1,\ldots,T_n).
    \label{eq:realdata-lambda-hat}
\end{equation}
However, this estimator is not used in the subsequent empirical analysis,
because no Gamma jump-size model is imposed on the detected threshold
exceedances. Consequently,
$\widehat\lambda_{n,\boldsymbol T}^{(k)}$ is not interpreted as an
estimate of a structural parameter. The empirical analysis instead
focuses on the cross-sectional distribution of the detected jump
frequencies.
The Bernstein estimator is applied to the cross-sectional sample
$\widehat C_1,\ldots,\widehat C_n$, with $n=21$ and fixed order
$
    m=9.
$
This is the nearest integer to
$
    \left(
        \frac{96}{5\sqrt{\pi}}
    \right)^{2/5}
    n^{2/5},
$
the Beta$(2,2)$ benchmark order evaluated at $n=21$.

Annualized frequency estimates $\widehat C_i$ are not restricted to
$[0,1]$ and can exceed one. Therefore, rather than truncating them, we
apply an affine transformation from a data-driven compact interval
$[\underline c,\overline c]$ onto the canonical Bernstein interval
$[0,1]$. More precisely, define
$
    \widehat C_{\min}
    :=
    \min_{1\leq i\leq n}\widehat C_i,
    \;
    \widehat C_{\max}
    :=
    \max_{1\leq i\leq n}\widehat C_i,
    \;
    R_{\widehat C}
    :=
    \widehat C_{\max}-\widehat C_{\min}.
$
For the observed sample, $R_{\widehat C}>0$. Using the padding fraction
$\eta=0.05$, set
$
    \underline c
    :=
    \max
    \left\{
        0,\,
        \widehat C_{\min}-\eta R_{\widehat C}
    \right\},
    \;
    \overline c
    :=
    \widehat C_{\max}+\eta R_{\widehat C}.
$
The normalized pseudo-observations are
$
    \widehat U_i
    :=
    \Pi_{[0,1]}
    \left(
        \frac{
            \widehat C_i-\underline c
        }{
            \overline c-\underline c
        }
    \right),
    \;
    i=1,\ldots,n.
$
The Bernstein density estimator $\widehat f_U$ is constructed from
$\widehat U_1,\ldots,\widehat U_n$ on $[0,1]$ and is transformed back to
the original annualized-frequency scale according to
$
    \widehat f_C(c)
    =
    \frac{1}{\overline c-\underline c}
    \widehat f_U
    \left(
        \frac{
            c-\underline c
        }{
            \overline c-\underline c
        }
    \right),
    \;
    c\in[\underline c,\overline c].
$
The interval $[0,1]$ is therefore used only as the canonical domain on
which the Bernstein estimator is constructed; it does not imply that the
annualized empirical frequencies $\widehat C_i$ are bounded by one.
\paragraph{Sensitivity to the detection threshold.}
We repeat the complete jump-extraction and Bernstein-reconstruction
procedure using the asset-specific quantile levels
$
    q\in\{0.90,0.95,0.975\},
$
while keeping $h=21$, $m=9$, and $\eta=0.05$ fixed. For each asset $i$
and each quantile level $q$, define
$
    \vartheta_i(q)
    :=
    \operatorname{Quantile}_{q}
    \left(
        \left\{
            \Delta Y_t^{(i)}
            :
            \Delta Y_t^{(i)}>0
        \right\}
    \right)
$
and
$
    N_i(q)
    :=
    \sum_t
    \mathds{1}_{
        \{\Delta Y_t^{(i)}>\vartheta_i(q)\}
    }.
$
The corresponding annualized frequency estimate is
$
    \widehat C_i(q)
    :=
    \frac{252\,N_i(q)}{L_i}.
$
The padded support is recomputed separately for each value of $q$.
If $M_i^+$ denotes the number of strictly positive increments available
for asset $i$, then the defining property of the empirical quantile
suggests the approximation
$
    N_i(q)
    \simeq
    (1-q)M_i^+.
$
This is an order-of-magnitude approximation rather than an exact identity.
The actual number of exceedances may differ slightly from
$(1-q)M_i^+$ because of observations equal to the estimated threshold and
the convention used to compute the empirical quantile. The approximation
is included only to facilitate the interpretation of
Table~\ref{tab:realdata-threshold-sensitivity}; it is not used in the
estimation procedure.

For each $q\in\{0.90,0.95,0.975\}$, let
$\widehat f_C^{(q)}$ denote the back-transformed Bernstein density
estimator obtained after recomputing the padded support and the normalized
pseudo-observations at quantile level $q$. We measure its discrepancy
from the baseline estimator corresponding to $q=0.95$ by
\begin{equation}
    D_1(q,0.95)
    :=
    \int_{\mathbb R}
    \left|
        \widehat f_C^{(q)}(c)
        -
        \widehat f_C^{(0.95)}(c)
    \right|
    \dd c.
    \label{eq:realdata-threshold-L1}
\end{equation}
The integral is taken over $\mathbb R$ because the padded support is
recomputed separately for each quantile level.

Table~\ref{tab:realdata-threshold-sensitivity} reports the resulting
total jump counts and summary statistics of the annualized frequency
estimates, together with $D_1(q,0.95)$.
\begin{table}[H]
\centering
\scriptsize
\caption{Sensitivity of the detected jump counts and annualized frequency
estimates to the asset-specific positive-increment quantile level. The
distance $D_1(q,0.95)$ is the $L^1$-distance between the Bernstein density
obtained at level $q$ and the baseline density obtained at $q=0.95$. The
number of trading days used to construct each realized-variance
observation, the degree parameter of the Bernstein estimator, and the
support-padding fraction are fixed at $h=21$, $m=9$, and $\eta=0.05$,
respectively.}
\label{tab:realdata-threshold-sensitivity}
\begin{tabular}{crrrrrr}
\toprule
$q$
&
$\sum_iN_i(q)$
&
$\operatorname{Mean}(\widehat C_i(q))$
&
$\operatorname{Median}(\widehat C_i(q))$
&
$\operatorname{Min}(\widehat C_i(q))$
&
$\operatorname{Max}(\widehat C_i(q))$
&
$D_1(q,0.95)$
\\
\midrule
$0.90$  & $4339$ & $12.586$ & $12.585$ & $12.367$ & $12.887$ & $2.000$ \\
$0.95$  & $2175$ & $ 6.309$ & $ 6.292$ & $ 6.219$ & $ 6.474$ & $0.000$ \\
$0.975$ & $1093$ & $ 3.170$ & $ 3.146$ & $ 3.109$ & $ 3.267$ & $2.000$ \\
\bottomrule
\end{tabular}
\end{table}

\begin{figure}[H]
\centering
\includegraphics[width=0.82\textwidth]
{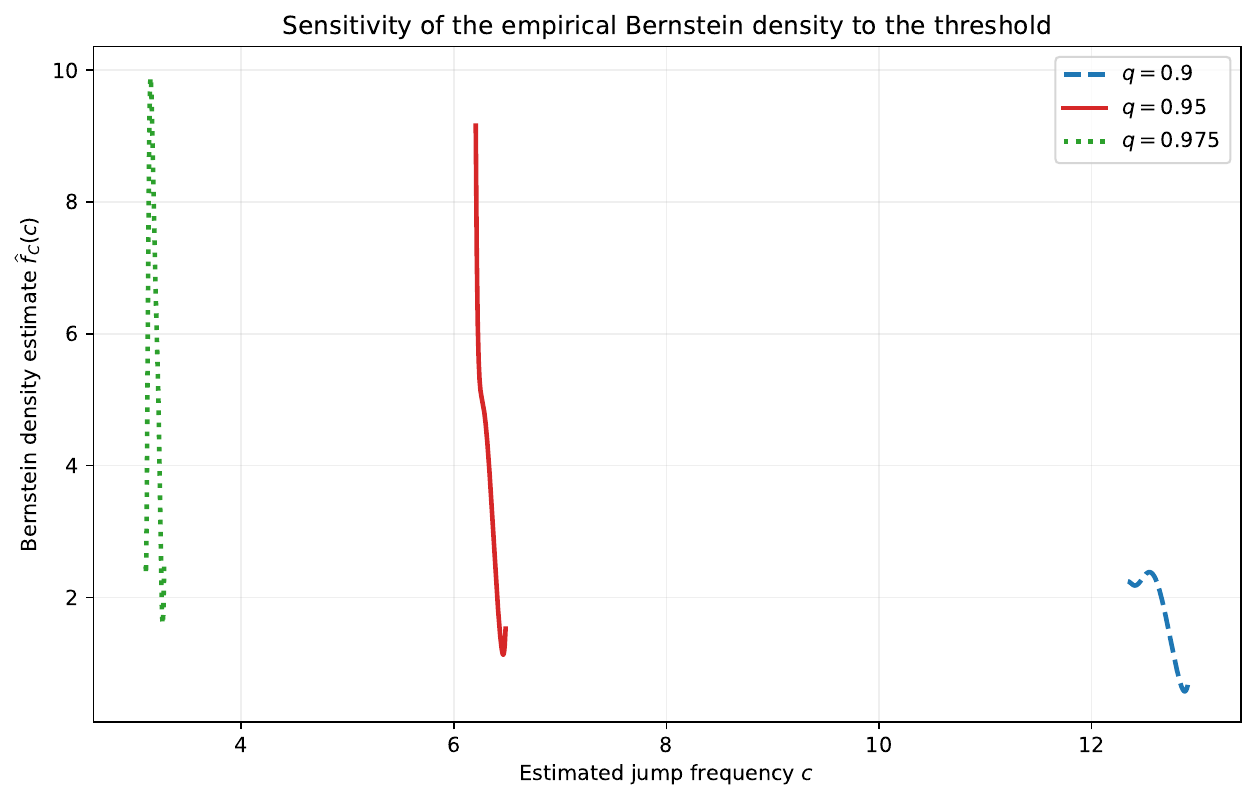}
\caption{Sensitivity of the empirical Bernstein density estimate to the
asset-specific positive-increment quantile level. The number of trading
days used to construct each realized-variance observation, the degree
parameter of the Bernstein estimator, and the support-padding fraction
are fixed at $h=21$, $m=9$, and $\eta=0.05$, respectively. For each
quantile level, the padded support and the normalized pseudo-observations
are recomputed before constructing the Bernstein estimator.}
\label{fig:realdata-threshold-sensitivity}
\end{figure}
Lowering the quantile level from $0.95$ to $0.90$ increases the total
number of detected jumps from $2175$ to $4339$ and nearly doubles the
cross-sectional mean frequency, from $6.309$ to $12.586$. Conversely,
raising the quantile level to $0.975$ reduces the total count to $1093$
and the mean frequency to $3.170$, approximately one half of its baseline
value.

The corresponding padded supports are
$
    [12.341,12.913],
    \;
    [6.206,6.487],
    \;
    [3.101,3.275],
$
for $q=0.90$, $q=0.95$, and $q=0.975$, respectively. These intervals are
pairwise disjoint, which explains the horizontal separation of the three
curves in Figure~\ref{fig:realdata-threshold-sensitivity}. Since each
Bernstein density integrates to one and is zero outside its corresponding
padded support, two estimators with disjoint supports have $L^1$-distance
equal to $2$. This explains the values reported in the last column of
Table~\ref{tab:realdata-threshold-sensitivity}.

The sensitivity analysis therefore shows that the absolute location of
the estimated frequency distribution is largely determined by the chosen
quantile level. Indeed, because approximately a fraction $1-q$ of the
positive increments is retained for each asset, changing $q$ mechanically
changes the scale of the estimated jump frequencies.

The corresponding cross-sectional standard deviations are $0.135$,
$0.070$, and $0.038$ for $q=0.90$, $q=0.95$, and $q=0.975$,
respectively. Together with the narrow ranges reported in
Table~\ref{tab:realdata-threshold-sensitivity}, these values show that the
asset-specific quantile rule produces only limited cross-sectional
dispersion. The three curves should consequently be interpreted as a
threshold-sensitivity diagnostic rather than as three estimates of a
common threshold-free structural intensity. The baseline empirical
results use $q=0.95$.
\paragraph{Interpretation of the empirical figures.}
Figure~\ref{fig:realdata-density} displays the empirical Bernstein
reconstruction of the cross-sectional density of the jump-frequency
estimates $\widehat C_i$. The histogram represents the observed
cross-sectional variation across assets, whereas the Bernstein curve
provides a smooth descriptive representation of the distribution of the
detected jump frequencies.
This construction is the empirical analogue of the feasible estimator
$\widehat f_{C,m,n,T}$ studied in
Theorem~\ref{thm:bias-var-mise}. 
In the theoretical setting, the pseudo-observations are computed from
exactly observed jump counts, whereas in the empirical application they
are computed from threshold-based counts extracted from discrete
financial data. Since the resulting frequency estimates lie in a bounded
range, a compact-support Bernstein representation is convenient for this
descriptive analysis.

Here, $n=21$ and $m=9$, whereas
Theorems~\ref{thm:bias-var-mise} and
\ref{thm:pointwise-normality} provide asymptotic results as
$n\to\infty$ and $m=m_n\to\infty$, subject to the corresponding growth
conditions. In particular, the MISE consistency result requires
$
    \frac{\sqrt m}{n}\longrightarrow0,
$
whereas the pointwise central limit theorem requires
$
    \frac{m}{n}\longrightarrow0.
$
Consequently, at the present sample size, the curve in
Figure~\ref{fig:realdata-density} should be interpreted as a smoothed
descriptive summary of the $21$ pseudo-observations
$\widehat C_1,\ldots,\widehat C_{21}$, rather than as a density estimate
supported by the asymptotic guarantees of these theorems.
\begin{figure}[H]
\centering
\includegraphics[width=0.72\textwidth]
{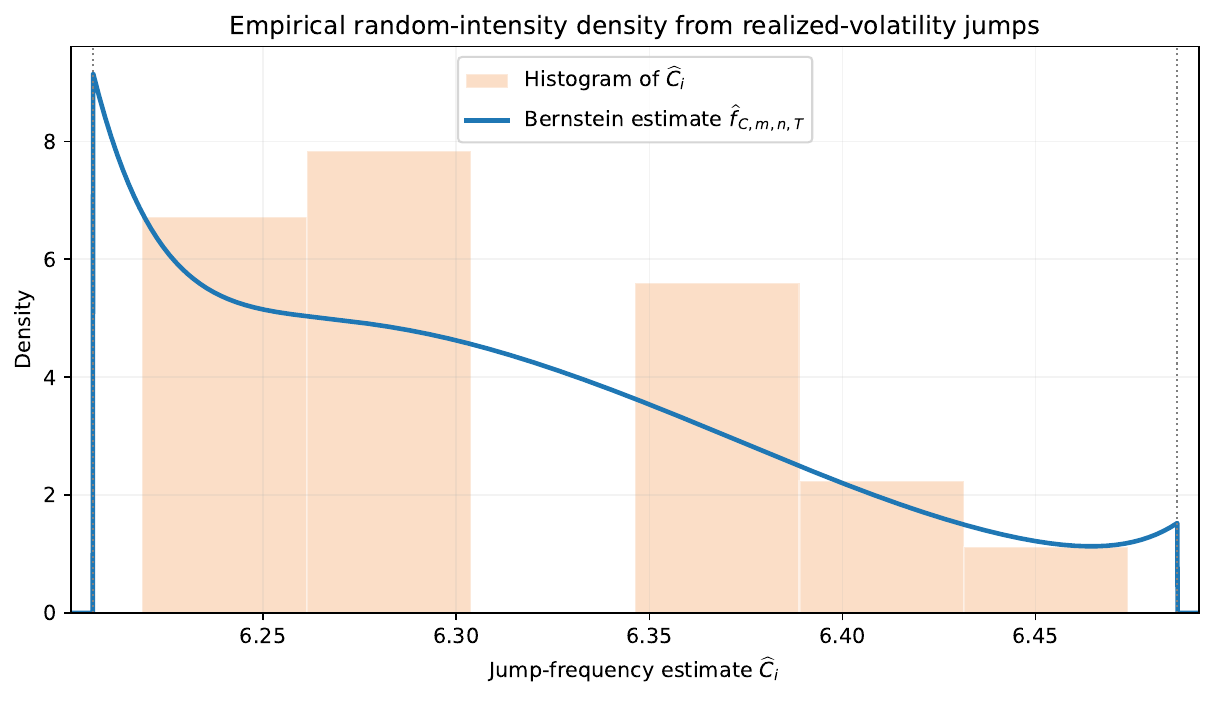}
\caption{Empirical Bernstein reconstruction of the cross-sectional
density of the estimated jump frequencies $\widehat C_i$ obtained from
the financial realized-volatility panel. The histogram represents the
empirical distribution of the estimates, whereas the smooth curve is the
corresponding Bernstein density estimator.}
\label{fig:realdata-density}
\end{figure}
Figure~\ref{fig:realdata-Chat-assets} reports the individual annualized
jump-frequency estimates
$
    \widehat C_i
    =
    \frac{N_{T_i}^{(i)}}{T_i}
$
for each asset. This diagnostic plot provides an empirical representation
of the cross-sectional variation underlying the random-effects
formulation. It reveals limited but nonzero variation in the frequencies
of detected large upward volatility movements under the baseline
quantile level. This variation is the empirical quantity summarized by
the Bernstein reconstruction. Given the asset-specific quantile rule and
the small cross-sectional sample, however, it should not be interpreted
as direct evidence of substantial heterogeneity in latent structural jump
intensities.

\begin{figure}[H]
\centering
\includegraphics[width=0.86\textwidth]
{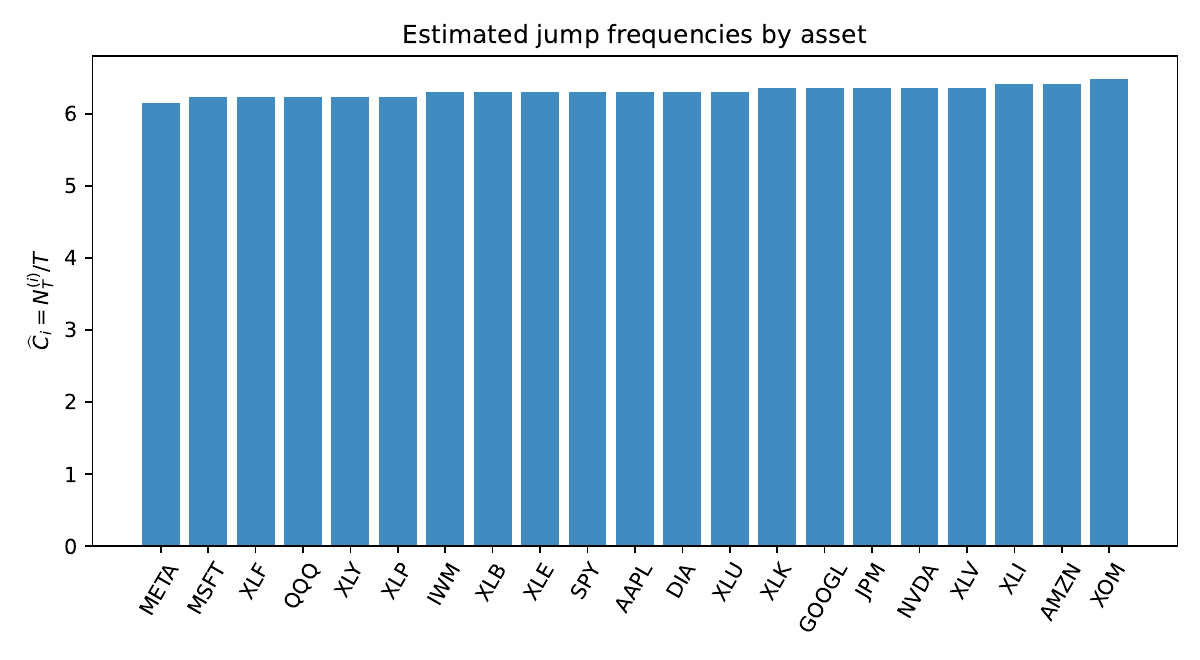}
\caption{Annualized jump-frequency estimates
$\widehat C_i=N_{T_i}^{(i)}/T_i$ by asset. The bars represent the limited
but nonzero cross-sectional variation in the empirical frequencies of
detected positive realized-volatility jumps under the baseline quantile
level $q=0.95$.}
\label{fig:realdata-Chat-assets}
\end{figure}

Figure~\ref{fig:realdata-construction} illustrates the construction of
the empirical BAJD-type sample. Panel~(a) displays several
realized-volatility proxies $Y_t^{(i)}$. These positive trajectories
exhibit pronounced upward movements, particularly during periods of
market stress, which motivates their representation through a positive
jump-type model. Panel~(b) presents the jump-detection diagnostic for SPY.
The horizontal line represents the $95\%$ positive-increment threshold
$\vartheta_{\mathrm{SPY}}(0.95)$, and the highlighted points identify the
detected positive jumps. These observations determine the empirical
statistics $N_{T_i}^{(i)}$ and $J_{T_i}^{(i)}$ used in the first-stage
frequency estimator \eqref{eq:realdata-Chat} and in the formal pooled
Gamma-rate expression \eqref{eq:realdata-lambda-hat}, respectively.

\begin{figure}[H]
\centering
\begin{minipage}{0.48\textwidth}
\centering
\includegraphics[width=\textwidth]
{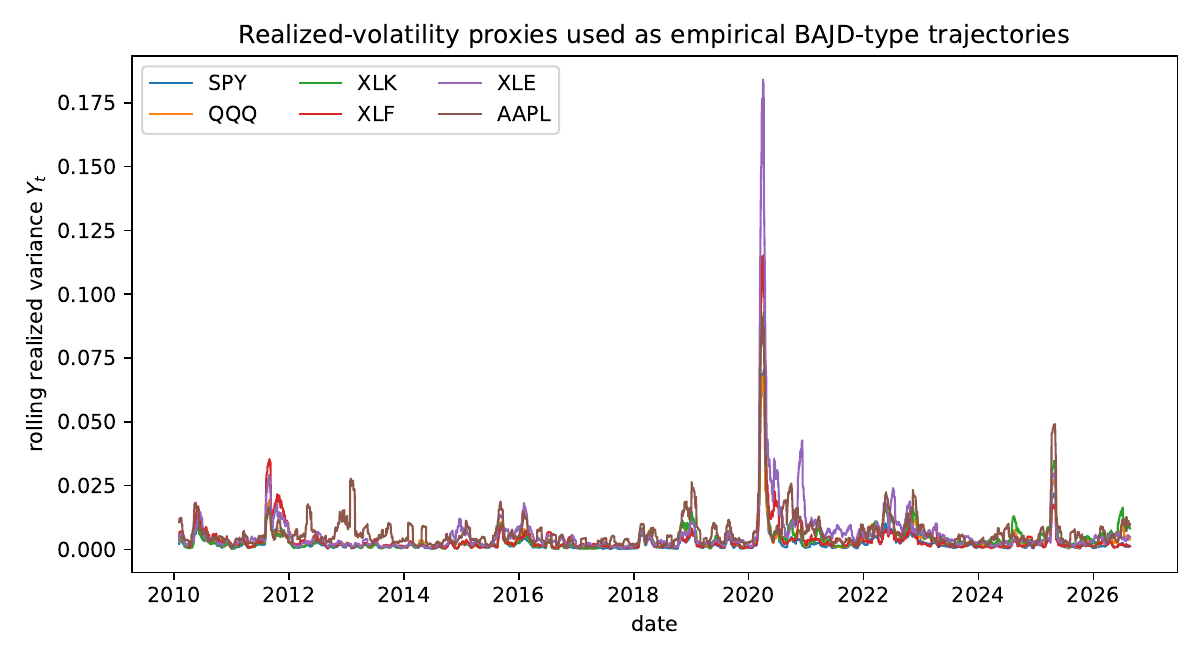}\\[-1mm]
\small (a) Realized-volatility proxies.
\end{minipage}
\hfill
\begin{minipage}{0.48\textwidth}
\centering
\includegraphics[width=\textwidth]
{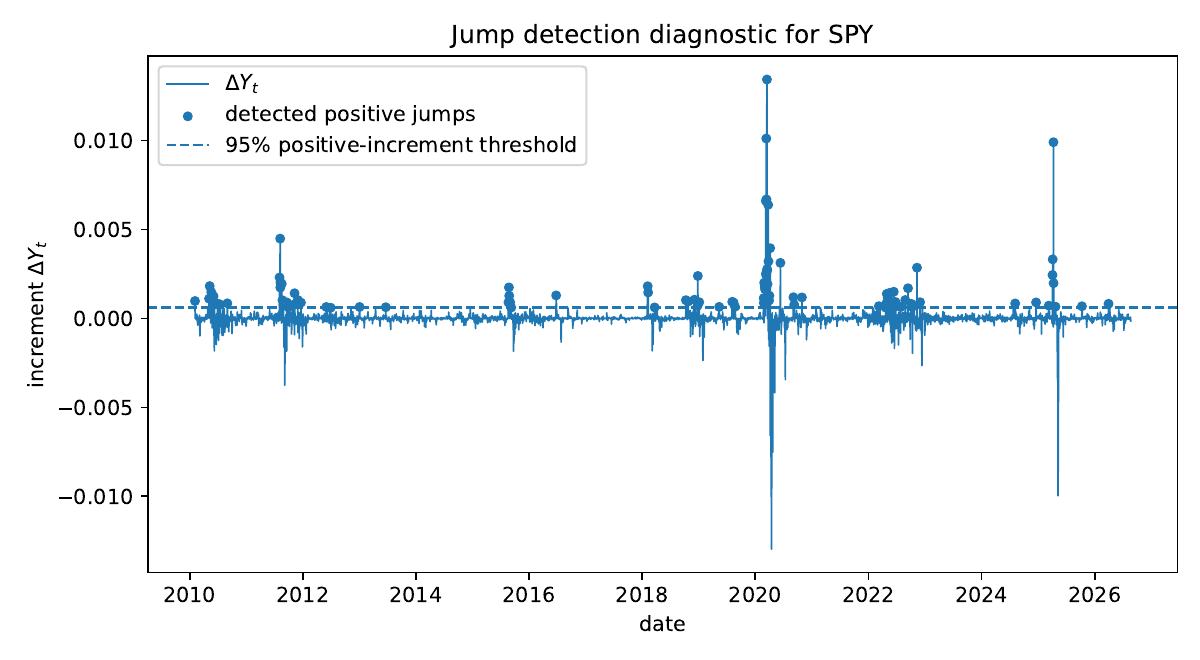}\\[-1mm]
\small (b) Positive-jump detection diagnostic.
\end{minipage}
\caption{Construction of the empirical BAJD-type sample from financial
data. Panel~(a) displays the rolling realized-volatility proxies used as
positive trajectories. Panel~(b) illustrates the quantile-based rule used
to detect empirical positive jumps.}
\label{fig:realdata-construction}
\end{figure}
\paragraph{Connection with the theory and scope of the empirical analysis.}
This empirical subsection complements the simulated BAJD study in three
ways. First, it shows that the estimator
$
    \widehat C_i
    =
    \frac{N_{T_i}^{(i)}}{T_i},
$
which has the same functional form as the likelihood-based estimator in
Proposition~\ref{prop:mle-Ci}, has a natural empirical counterpart once
positive volatility jumps have been identified using a threshold rule.
Second, it illustrates the role of the Bernstein estimator as a
population-level tool: the individual assets provide the
pseudo-observations $\widehat C_i$, which are subsequently used to
estimate the cross-sectional density of the detected jump frequencies.
Third, it illustrates the separation between the frequency and jump-size
components. The empirical analysis focuses on the cross-sectional
distribution of the detected jump frequencies, reconstructed through
$\widehat f_C$. The pooled Gamma-rate formula
$\widehat\lambda_{n,\boldsymbol T}^{(k)}$ is presented only as an
empirical counterpart of the theoretical estimator and is not assigned a
structural interpretation.

The empirical results should nevertheless be interpreted with caution.
A nonparametric density reconstructed from only $n=21$
pseudo-observations is intrinsically sensitive to the Bernstein order,
the support-padding rule, and the composition of the selected asset
universe. It should therefore be regarded as descriptive and exploratory
rather than as a precise estimate of an underlying population density.
More generally, the numerical values of $\widehat C_i$ and $\widehat f_C$
depend on the number $h$ of trading days used to construct each
realized-variance observation, the jump-detection threshold, the sampling
frequency, the support transformation, and the selected asset universe.
The fixed observation period, the choice $m=9$, the
padding rule, and the threshold-sensitivity analysis described above make
the empirical procedure reproducible, but they do not eliminate these
statistical limitations. This empirical subsection should therefore be
viewed as an illustration of the proposed methodology rather than as a
formal validation of the continuous-observation BAJD framework.

\section*{General conclusion}

This paper develops a random-effects extension of the BAJD estimation
problem in which heterogeneity acts through the jump frequency. Each
trajectory has its own positive intensity $C_i$, drawn from an unknown
density $f_C$, whereas the remaining structural parameters are common
across the population.

Under continuous-time observation, the individual jump count
$N_T^{(i)}$ is observed pathwise. The extended MLE of $C_i$ is therefore
the empirical jump frequency $N_T^{(i)}/T$; after the first observed
jump, it coincides with the ordinary MLE on $\real_{++}$. We established
the strong consistency and stable mixed-normal limit of this estimator.
At the population level, the estimated intensities serve as
pseudo-observations for a Bernstein density estimator. Under a sequential
asymptotic regime, we established the asymptotic equivalence between the
feasible and ideal Bernstein estimators, derived pointwise and integrated
bias and variance expansions, obtained the corresponding MISE expansion
and optimal Bernstein order, and proved pointwise asymptotic normality.
These results make explicit the distinct roles of the observation
horizon, the cross-sectional sample size, and the Bernstein order.

The individual intensity estimator remains unchanged for arbitrary
positive jump-size densities in the finite-activity setting, because it
depends only on the jump counts. For the Gamma$(k,\lambda)$ family, we
additionally derived a pooled MLE of the common rate parameter and
established its consistency and asymptotic normality. The counting
martingales and the centered jump-size martingale are orthogonal, yielding
first-order conditional asymptotic independence between the individual
intensity estimators and the pooled rate estimator. The Gamma
specialization also provides explicit immigration mechanisms and
conditional stationary means, together with a consistent plug-in
estimator of the stationary population mean. This composite estimator is
conditional on known or externally calibrated values of the common
parameters $a$ and $b$; their joint estimation and the propagation of
the resulting estimation uncertainty remain topics for future research.
The exponential and Erlang jump-size models are recovered as the cases
$k=1$ and $k=2$, respectively.

The numerical study illustrates the complete estimation procedure. The
controlled simulation experiment compares $\widehat C_{i,T}$,
$\widehat\lambda_{n,T}^{(k)}$, and the feasible Bernstein estimator with
their theoretical targets. The empirical illustration adapts the same
procedure to a panel of financial realized-volatility proxies. Its
threshold-sensitivity analysis shows that the absolute location of the
estimated jump-frequency distribution depends strongly on the detection
threshold and that the asset-specific quantile rule produces only limited
cross-sectional dispersion.
Because the empirical construction uses discretely observed financial
data, each realized-volatility proxy is computed from overlapping sets of
daily returns, and the panel contains only $21$ assets, the resulting
density reconstruction should be interpreted as descriptive and
exploratory rather than as a formal validation of the
continuous-observation BAJD framework.

Several extensions are natural. A first direction is to develop a
simultaneous asymptotic theory in which the number of trajectories $n$,
the observation horizon $T$, and the Bernstein order $m$ tend jointly to
infinity. A second direction is to consider bivariate random effects, such
as $(C_i,\lambda_i)$, to model joint heterogeneity in jump frequency and
jump amplitude. A third direction is to introduce random effects in the
mean-reversion parameter $b$, thereby allowing for heterogeneity in the
speed of convergence toward equilibrium. A fourth direction is to compare
Bernstein, kernel, and deconvolution estimators numerically in order to
clarify their finite-sample behavior, particularly near the boundaries of
the intensity distribution. Finally, motivated by the empirical
financial illustration, an important extension would be to develop the
theoretical results for unbalanced panels in which the observation
horizons $T_i$ vary across trajectories.

\appendix
\section{ Proofs of the auxiliary lemmas}
\label{app:proofs-auxiliary-lemmas}
This appendix collects the proofs of the four auxiliary lemmas used in the
asymptotic analysis of Section~\ref{sec:general-framework}.
Lemma~\ref{lem:stable-poisson-random-intensity} provides the stable
Poisson limit underlying Theorem~\ref{thm:Ci-mle-an}.
Lemmas~\ref{lem:ideal-bernstein-bias} and
\ref{lem:ideal-bernstein-variance} establish the pointwise and integrated
bias and variance expansions used in Theorem~\ref{thm:bias-var-mise}.
Finally, Lemma~\ref{lem:ideal-bernstein-pointwise-clt} supplies the
triangular-array central limit theorem required for
Theorem~\ref{thm:pointwise-normality}. Their proofs are deferred to this
appendix to keep the main exposition focused on the statistical results.

\begin{proof}[\bf Proof of Lemma~\ref{lem:stable-poisson-random-intensity}]

Let $H$ be a bounded $\mathcal G_i$-measurable random variable and let
$u\in\mathbb R$. Using conditional expectation given $\mathcal G_i$ and
the conditional characteristic function of the Poisson distribution, we
obtain
$$\mathbb E\left[
        H\exp\left(
            iu\frac{N_T^{(i)}-C_iT}{\sqrt T}
        \right)
    \right]
    =
    \mathbb E\left[
        H\exp\left\{
            C_iT\left(
                e^{iu/\sqrt T}
                -1
                -\frac{iu}{\sqrt T}
            \right)
        \right\}
    \right].
$$
Since
$
    T\left(
        e^{iu/\sqrt T}
        -1
        -\frac{iu}{\sqrt T}
    \right)
    \xrightarrow[T\to\infty]{}
    -\frac{u^2}{2},
$
and
$$
    \left|
        \exp\left\{
            C_iT\left(
                e^{iu/\sqrt T}
                -1
                -\frac{iu}{\sqrt T}
            \right)
        \right\}
    \right|
    =
    \exp\left\{
        C_iT\left(
            \cos\left(\frac{u}{\sqrt T}\right)-1
        \right)
    \right\}
    \leq 1,
$$
the dominated convergence theorem yields
$
    \mathbb E\left[
        H\exp\left(
            iu\frac{N_T^{(i)}-C_iT}{\sqrt T}
        \right)
    \right]
    \xrightarrow[T\to\infty]{}
    \mathbb E\left[
        H\exp\left(
            -\frac{u^2C_i}{2}
        \right)
    \right].
$
Because $Z_i$ is independent of $\mathcal G_i$,
$
    \mathbb E\left[
        \left.
        e^{iu\sqrt{C_i}Z_i}
        \,\right|\,
        \mathcal G_i
    \right]
    =
    \exp\left(-\frac{u^2C_i}{2}\right).
$
Therefore,
$$
    \mathbb E\left[
        H\exp\left(
            iu\frac{N_T^{(i)}-C_iT}{\sqrt T}
        \right)
    \right]
    \xrightarrow[T\to\infty]{}
    \mathbb E\left[
        He^{iu\sqrt{C_i}Z_i}
    \right].
$$
This is the characteristic-function characterization of stable convergence
with respect to $\mathcal G_i$, and hence proves
\eqref{eq:stable-poisson-random-intensity}.
\end{proof}

\begin{rema}
The conclusion of
Lemma~\ref{lem:stable-poisson-random-intensity} can alternatively be
recovered from the martingale GCLT stated in
Theorem~\ref{GCLTtouati} of Appendix~B. Indeed, fix
$i\in\{1,\ldots,n\}$ and consider the filtration initially enlarged by
$\mathcal G_i=\sigma(C_i)$. Then
$
    M_{i,t}:=N_t^{(i)}-C_it,
    \; t\geq0,
$
is a purely discontinuous, locally square-integrable martingale. Since the
jump times of the conditionally Poisson process $N^{(i)}$ are totally
inaccessible, for every predictable stopping time $\tau$,
$
    \Delta M_{i,\tau}=0
    \; \text{a.s. on }\{\tau<\infty\}.
$
Therefore, $M_i$ is quasi-left-continuous. Moreover, its predictable
quadratic variation satisfies
$
    \frac{1}{T}\langle M_i\rangle_T=C_i.
$
Since the jumps of $M_i$ are of size $1$, its Lindeberg term is, for every
$\varepsilon>0$,
$
    \frac{1}{T}
    \int_0^T\int_{\mathbb R}
        x^2\mathbf 1_{\{|x|>\varepsilon\sqrt T\}}
        \nu^{M_i}(dt,dx)
    =
    C_i\mathbf 1_{\{1>\varepsilon\sqrt T\}}
    \xrightarrow[T\to\infty]{\mathrm{a.s.}}0.
$
Thus, the bracket and Lindeberg conditions of
Theorem~\ref{GCLTtouati} are satisfied. Consequently, stably with respect
to $\mathcal G_i$,
$
    \frac{N_T^{(i)}-C_iT}{\sqrt T}
    =
    \frac{M_{i,T}}{\sqrt T}
    \xrightarrow[T\to\infty]{\mathcal D,\mathrm{st}}
    \sqrt{C_i}\,Z_i,
$
where $Z_i\sim\mathcal N(0,1)$ is independent of $\mathcal G_i$. This
recovers \eqref{eq:stable-poisson-random-intensity}.
\end{rema}

\begin{proof}[\bf Proof of Lemma~\ref{lem:ideal-bernstein-bias}]
By Assumption~\ref{ass:population}, the random variables
$(C_i)_{i\geq1}$ have common distribution function $F_C$ and common
density $f_C$. Since $f_C$ is supported on $[0,1]$, for every
$k=0,\ldots,m-1$,
$
    F_C\left(\frac{k+1}{m}\right)
    -
    F_C\left(\frac{k}{m}\right)
    =    \int_{k/m}^{(k+1)/m}f_C(t)\dd t.
$
Consequently,

$
    \widetilde f_{C,m}(u)
    =
    \mathbb E\bigl[\widetilde f_{C,m,n}(u)\bigr]
    =
    \sum_{k=0}^{m-1}
    \left[        m\int_{k/m}^{(k+1)/m}f_C(t)\dd t
    \right]
    p_k(m-1,u).
$
In particular, $\widetilde f_{C,m}(u)$ does not depend on $n$.
For $k=0,\ldots,m-1$, set
$
    \bar f_{C,k}
    :=
    m\int_{k/m}^{(k+1)/m}f_C(t)\dd t,
$
and let
$
    K_{m,u}
    \sim
    \operatorname{Binomial}(m-1,u).
$
Then
$
    \widetilde f_{C,m}(u)
    =
    \mathbb E\bigl[\bar f_{C,K_{m,u}}\bigr].
$
For the subinterval
$\left[k/m,(k+1)/m\right]$, denote its midpoint by
$
    c_k
    :=
    \frac{k+1/2}{m}.
$
Since $f_C\in C^2([0,1])$, a second-order Taylor expansion of $f_C$ at
$u$ gives, for every $t,u\in[0,1]$,
$
    f_C(t)
    =
    f_C(u)
    +
    f_C'(u)(t-u)
    +
    \frac{1}{2}f_C''(u)(t-u)^2
    +
    R_u(t).
$
Direct integration over the subinterval yields
$
    m\int_{k/m}^{(k+1)/m}(t-u)\dd t
    =
    c_k-u
$
and
$
    m\int_{k/m}^{(k+1)/m}(t-u)^2\dd t
    =
    (c_k-u)^2+\frac{1}{12m^2}.
$
Therefore,
$
    \bar f_{C,k}
    =
    f_C(u)
    +
    f_C'(u)(c_k-u)
    +
    \frac{1}{2}f_C''(u)
    \left[
        (c_k-u)^2+\frac{1}{12m^2}
    \right]
    +
    \bar R_{k,u},
$
where
$
    \bar R_{k,u}
    :=
    m\int_{k/m}^{(k+1)/m}R_u(t)\dd t.
$
Taking expectations with respect to $K_{m,u}$ gives
$
    \widetilde f_{C,m}(u)
    =
    f_C(u)
    +
    f_C'(u)
    \mathbb E\bigl(c_{K_{m,u}}-u\bigr)
   +
    \frac{1}{2}f_C''(u)
    \left[
        \mathbb E\bigl(c_{K_{m,u}}-u\bigr)^2
        +
        \frac{1}{12m^2}
    \right]
    +
    \mathbb E\bigl(\bar R_{K_{m,u},u}\bigr).
$
The binomial moment formulas give, uniformly in $u\in[0,1]$,
$
    \mathbb E\bigl(c_{K_{m,u}}-u\bigr)
    =
    \frac{1-2u}{2m}
$
and
$
    \mathbb E\bigl(c_{K_{m,u}}-u\bigr)^2
    =
    \frac{u(1-u)}{m}
    +
    O(m^{-2}).
$
Hence,
\begin{align}
    \widetilde f_{C,m}(u)-f_C(u)
    =
    \frac{1}{2m}
    \left[
        (1-2u)f_C'(u)
        +
        u(1-u)f_C''(u)
    \right]
    +
    O(m^{-2})
    +
    \mathbb E\bigl(\bar R_{K_{m,u},u}\bigr),
    \label{eq:preliminary-bias-expansion}
\end{align}
where the $O(m^{-2})$ term is uniform in $u\in[0,1]$.
It remains to control the Taylor remainder uniformly, including at the
endpoints. Let $V$ be uniformly distributed on $[0,1]$, independently
of $K_{m,u}$, and define
$
    X_{m,u}
    :=
    \frac{K_{m,u}+V}{m}.
$
Conditioning on $K_{m,u}$ gives
$
    \widetilde f_{C,m}(u)
    =
    \mathbb E\bigl[f_C(X_{m,u})\bigr]
$
and
$
    \mathbb E\bigl(\bar R_{K_{m,u},u}\bigr)
    =
    \mathbb E\bigl[R_u(X_{m,u})\bigr].
$
Since $f_C\in C^2([0,1])$ and $[0,1]$ is compact, the function $f_C''$
is bounded and uniformly continuous on $[0,1]$. Let
$\omega_{f_C''}$ denote its modulus of continuity. Taylor's formula with
integral remainder yields, for every $x,u\in[0,1]$,
\[
    \left|
        f_C(x)
        -
        f_C(u)
        -
        f_C'(u)(x-u)
        -
        \frac{1}{2}f_C''(u)(x-u)^2
    \right|
    \leq
    \frac{1}{2}
    \omega_{f_C''}(|x-u|)(x-u)^2.
\]
Consequently,
$
    \left|
        \mathbb E\bigl(\bar R_{K_{m,u},u}\bigr)
    \right|
    \leq
    \frac{1}{2}
    \mathbb E\left[
        \omega_{f_C''}(|X_{m,u}-u|)
        (X_{m,u}-u)^2
    \right].
$
The binomial moment formulas, together with the boundedness of $V$,
imply, uniformly in $u\in[0,1]$, that
$
    \mathbb E(X_{m,u}-u)^2
    =
    O(m^{-1}),$ $
        \mathbb E(X_{m,u}-u)^4
    =
    O(m^{-2}).
$
For every $\delta>0$, splitting according to
$\{|X_{m,u}-u|\leq\delta\}$ and its complement gives
$$
    m\,
    \mathbb E\left[
        \omega_{f_C''}(|X_{m,u}-u|)
        (X_{m,u}-u)^2
    \right]
  \leq
    K_{\mathrm r}\,
    \omega_{f_C''}(\delta)
    +
    \frac{
        2m\lVert f_C''\rVert_\infty
    }{
        \delta^2
    }
    \mathbb E(X_{m,u}-u)^4
  \leq
    K_{\mathrm r}\,
    \omega_{f_C''}(\delta)
    +
    \frac{K_{\mathrm r}}{m\delta^2},
$$
uniformly in $u\in[0,1]$, where $K_{\mathrm r}>0$ is a deterministic
constant independent of $m$, $u$, and $\delta$.
Letting first $m\to\infty$ and then $\delta\downarrow0$ yields
$
    \sup_{u\in[0,1]}
    m\left|
        \mathbb E\bigl(\bar R_{K_{m,u},u}\bigr)
    \right|
    \xrightarrow[m\to\infty]{}
    0.
$
Combining this result with
\eqref{eq:preliminary-bias-expansion} proves
\eqref{eq:lemma-uniform-bias-expansion}.
Finally, define
$
    \beta_{f_C}(u)
    :=
    \frac{1}{2}
    \left[
        (1-2u)f_C'(u)
        +
        u(1-u)f_C''(u)
    \right].
$
Since $f_C\in C^2([0,1])$, the function $\beta_{f_C}$ is continuous and
bounded on $[0,1]$. By
\eqref{eq:lemma-uniform-bias-expansion},
$
    \sup_{u\in[0,1]}
    \left|
        m\bigl(\widetilde f_{C,m}(u)-f_C(u)\bigr)
        -
        \beta_{f_C}(u)
    \right|
    \xrightarrow[m\to\infty]{}
    0.
$
Squaring this uniform expansion and integrating over $[0,1]$ gives
$$
    \int_0^1
    \bigl(\widetilde f_{C,m}(u)-f_C(u)\bigr)^2
    \dd u
    =
    \frac{1}{m^2}
    \int_0^1\beta_{f_C}(u)^2\dd u
    +
    o(m^{-2})=
    B_{f_C}m^{-2}
    +
    o(m^{-2}),
$$
which proves \eqref{eq:lemma-integrated-bias}. This completes the proof.
\end{proof}

\begin{proof}[\bf Proof of Lemma~\ref{lem:ideal-bernstein-variance}]
Fix $u\in(0,1)$ and set $N:=m-1$. For $v\in[0,1]$, define
$
    \kappa_m(v)
    :=
    \min\{\lfloor mv\rfloor,m-1\}.
$
Since the distribution of each $C_i$ is non-atomic, no $C_i$ coincides
with a grid point $k/m$ almost surely. Therefore,
$
    \widetilde f_{C,m,n}(u)
    =
    \frac{1}{n}\sum_{i=1}^n\xi_i
    \;\text{a.s.},
$
where
$
    \xi_i
    :=
    m\,p_{\kappa_m(C_i)}(m-1,u).
$
By Assumption~\ref{ass:population}, the variables
$\xi_1,\ldots,\xi_n$ are i.i.d. Hence,
\begin{equation}
    \operatorname{Var}\bigl(\widetilde f_{C,m,n}(u)\bigr)
    =
    \frac{1}{n}
    \left[
        \mathbb E(\xi_1^2)
        -
        \bigl(\mathbb E\xi_1\bigr)^2
    \right].
    \label{eq:lemma-var-iid}
\end{equation}

Define
$
    \bar f_{C,k}
    :=
    m\int_{k/m}^{(k+1)/m}f_C(t)\dd t,
    \; k=0,\ldots,m-1.
$
Since
$
    \mathbb P\left(
        C_1\in\left[\frac{k}{m},\frac{k+1}{m}\right)
    \right)
    =
    \frac{\bar f_{C,k}}{m},
$
with the choice of endpoints being irrelevant because the distribution is
non-atomic, we have
$
    \mathbb E(\xi_1^2)
    =
    m\sum_{k=0}^{m-1}
    \bar f_{C,k}p_k(m-1,u)^2.
$
The local limit estimate
\begin{equation}
    \sum_{k=0}^{N}p_k(N,u)^2
    =
    \frac{1}{\sqrt{4\pi Nu(1-u)}}
    +
    o(N^{-1/2})
    =
    \frac{\psi(u)}{\sqrt N}
    +
    o(N^{-1/2})
    \label{eq:lemma-local-limit}
\end{equation}
follows by writing the sum as
$
    \mathbb P(K_1=K_2),
$
where $K_1$ and $K_2$ are independent
$\operatorname{Binomial}(N,u)$ random variables, and applying the
classical local limit theorem to the random variable
$K_1-K_2$.
Since $f_C$ is uniformly continuous on $[0,1]$, a near--far decomposition
around $u$, together with an exponential binomial-tail bound, gives
$
    \sum_{k=0}^{m-1}
    \bar f_{C,k}p_k(m-1,u)^2
    =
    f_C(u)
    \sum_{k=0}^{m-1}p_k(m-1,u)^2
    +
    o(m^{-1/2}).
$
Combining this estimate with \eqref{eq:lemma-local-limit} and using
$N=m-1$, we obtain
$
    \mathbb E(\xi_1^2)
    =
    f_C(u)\psi(u)m^{1/2}
    +
    o(m^{1/2}).
$
Moreover,
$
    \mathbb E\xi_1
    =
    \widetilde f_{C,m}(u)
    =
    \sum_{k=0}^{m-1}
    \bar f_{C,k}p_k(m-1,u).
$
Since
$
    0\leq\bar f_{C,k}\leq\lVert f_C\rVert_\infty$  and $
        \sum_{k=0}^{m-1}p_k(m-1,u)=1,
$
it follows that
$
    0\leq
    \mathbb E\xi_1
    \leq
    \lVert f_C\rVert_\infty.
$
Consequently,
$
    \bigl(\mathbb E\xi_1\bigr)^2
    =
    O(1)
    =
    o(m^{1/2}).
$
Substituting these estimates into \eqref{eq:lemma-var-iid} proves
\eqref{eq:lemma-variance-expansion}.
It remains to establish the integrated variance expansion. 
For $k=0,\ldots,N$, set
$
    I_{m,k}
    :=
    \int_0^1p_k(N,u)^2\dd u,
    \;
    x_{m,k}
    :=
    \frac{k+1/2}{m}.
$
Since $N=m-1$, the Beta-function identity gives
\begin{align*}
    I_{m,k}
    &=
    \binom{N}{k}^{\!2}
    B(2k+1,2N-2k+1)=
    \frac{1}{2N+1}
    \frac{\binom{N}{k}^{\!2}}
         {\binom{2N}{2k}},
    \;
    k=0,\ldots,N.
\end{align*}
We now justify the uniform bound used below. For every integer $r\geq1$,
the elementary two-sided Stirling bounds give
$
    \sqrt{2\pi}\,
    r^{\,r+1/2}e^{-r}
    \leq
    r!
    \leq
    3\,r^{\,r+1/2}e^{-r}.
$
Applying these inequalities to the factorial representation
$
    I_{m,k}
    =
    \frac{
        (N!)^2(2k)!(2N-2k)!
    }{
        (k!)^2((N-k)!)^2(2N+1)!
    }
$
shows that, for $1\leq k\leq N-1$,
$
    I_{m,k}
    \leq
    \frac{K}{
        N^{3/2}
        \sqrt{(k/N)(1-k/N)}
    },
$
where $K>0$ is a deterministic constant independent of $m$ and $k$.
Since $N=m-1$, there exist deterministic constants $d_1,d_2>0$,
independent of $m$ and $k$, such that
$
    d_1x_{m,k}(1-x_{m,k})
    \leq
    \frac{k}{N}
    \left(
        1-\frac{k}{N}
    \right)
    \leq
    d_2x_{m,k}(1-x_{m,k}),
$
for every $1\leq k\leq N-1$. Moreover,
$
    \frac{m}{N}
    =
    \frac{m}{m-1}
    \leq2,
    \;
    m\geq2.
$
Consequently, there exists a deterministic constant
$K_{\mathrm{int}}>0$, independent of $m$ and $k$, such that
\[
    m^{3/2}I_{m,k}
    \leq
    \frac{K_{\mathrm{int}}}{
        \sqrt{x_{m,k}(1-x_{m,k})}
    },
    \;
    k=1,\ldots,N-1.
\]
For the boundary indices, the Beta-function identity gives explicitly
$
    I_{m,0}
    =
    I_{m,N}
    =
    \frac{1}{2m-1}.
$
Since
$
    x_{m,0}
    =
    \frac{1}{2m},
    \;
    x_{m,N}
    =
    1-\frac{1}{2m},
$
one has
$
    m^{3/2}I_{m,0}
    =
    m^{3/2}I_{m,N}
    \leq
    \frac{1}{
        \sqrt{x_{m,0}(1-x_{m,0})}
    }
    =
    \frac{1}{
        \sqrt{x_{m,N}(1-x_{m,N})}
    }.
$
Therefore, setting
$
    K_0
    :=
    \max\{K_{\mathrm{int}},1\},
$
we obtain the uniform bound
\begin{equation}
    m^{3/2}I_{m,k}
    \leq
    \frac{K_0}{
        \sqrt{x_{m,k}(1-x_{m,k})}
    },
    \;
    k=0,\ldots,N.
    \label{eq:integrated-binomial-bound}
\end{equation}
Moreover, the corresponding Stirling expansion is uniform when
$x_{m,k}$ remains in a fixed compact subset of $(0,1)$. Therefore, for
every fixed $\varepsilon\in(0,1/2)$,
\begin{equation}
    \sup_{\substack{
        0\leq k\leq m-1\\
        x_{m,k}\in[\varepsilon,1-\varepsilon]
    }}
    \left|
        m^{3/2}I_{m,k}
        -
        \psi(x_{m,k})
    \right|
    \xrightarrow[m\to\infty]{}
    0.
    \label{eq:integrated-binomial-interior}
\end{equation}

By Tonelli's theorem,
$
    \int_0^1\mathbb E\bigl(\xi_1^2\bigr)\dd u
    =
    m\sum_{k=0}^{m-1}\bar f_{C,k}I_{m,k}.
$
Since $f_C$ is uniformly continuous on $[0,1]$,
$
    \max_{0\leq k\leq m-1}
    \left|
        \bar f_{C,k}-f_C(x_{m,k})
    \right|
    \xrightarrow[m\to\infty]{}
    0.
$
On every fixed interior interval
$[\varepsilon,1-\varepsilon]$, where
$\varepsilon\in(0,1/2)$,
\eqref{eq:integrated-binomial-interior} and a Riemann-sum argument yield
\begin{align}
    &\frac{1}{\sqrt m}\,
    m
    \sum_{\substack{
        0\leq k\leq m-1\\
        x_{m,k}\in[\varepsilon,1-\varepsilon]
    }}
    \bar f_{C,k}I_{m,k}
    =
    \frac{1}{m}
    \sum_{\substack{
        0\leq k\leq m-1\\
        x_{m,k}\in[\varepsilon,1-\varepsilon]
    }}
    \bar f_{C,k}
    \bigl(m^{3/2}I_{m,k}\bigr)
    \xrightarrow[m\to\infty]{}
    \int_\varepsilon^{1-\varepsilon}
    f_C(u)\psi(u)\dd u.
    \label{eq:integrated-variance-interior}
\end{align}

It remains to control the two boundary regions. To obtain a disjoint
decomposition of $[0,1]$, define
$
    \mathcal B_\varepsilon
    :=
    [0,\varepsilon)\cup(1-\varepsilon,1].
$
By \eqref{eq:integrated-binomial-bound},
\begin{align}
    &\frac{1}{\sqrt m}\,
    m
    \sum_{\substack{
        0\leq k\leq m-1\\
        x_{m,k}\in\mathcal B_\varepsilon
    }}
    \bar f_{C,k}I_{m,k}
    \leq
    \frac{K_0\lVert f_C\rVert_\infty}{m}
    \sum_{\substack{
        0\leq k\leq m-1\\
        x_{m,k}\in\mathcal B_\varepsilon
    }}
    \frac{1}{
        \sqrt{x_{m,k}(1-x_{m,k})}
    }
    \leq
    K_1\lVert f_C\rVert_\infty\sqrt{\varepsilon},
    \label{eq:integrated-variance-boundary}
\end{align}
uniformly in $m$, for some deterministic constant $K_1>0$ independent
of $m$ and $\varepsilon$.
Combining \eqref{eq:integrated-variance-interior} and
\eqref{eq:integrated-variance-boundary}, then letting first
$m\to\infty$ and subsequently $\varepsilon\downarrow0$, yields
$    \frac{1}{\sqrt m}
    \int_0^1
    \mathbb E\bigl(\xi_1^2\bigr)\dd u
    \xrightarrow[m\to\infty]{}
    \int_0^1
    f_C(u)\psi(u)\dd u
    =
    A_{f_C}.
$
Finally, since
$
    0\leq
    \mathbb E\xi_1
    \leq
    \lVert f_C\rVert_\infty
$
uniformly in $u\in[0,1]$ and $m$,
$
    \int_0^1
    \bigl(\mathbb E\xi_1\bigr)^2
    \dd u
    =
    O(1)
    =
    o(m^{1/2}).
$
Integrating \eqref{eq:lemma-var-iid} over $[0,1]$ therefore yields
$
    \int_0^1
    \operatorname{Var}\bigl(\widetilde f_{C,m,n}(u)\bigr)
    \dd u
    =
    \frac{m^{1/2}}{n}A_{f_C}
    +
    o\left(\frac{m^{1/2}}{n}\right).
$
This proves \eqref{eq:lemma-integrated-variance} and completes the proof.
\end{proof}
\begin{proof}[\bf Proof of Lemma~\ref{lem:ideal-bernstein-pointwise-clt}]
Fix $u\in(0,1)$ such that $f_C(u)>0$. For $v\in[0,1]$, define
$
    \kappa_m(v)
    :=
    \min\{\lfloor mv\rfloor,m-1\}.
$
Thus, $\kappa_m(v)\in\{0,\ldots,m-1\}$ identifies the subinterval of the
regular partition of $[0,1]$ containing $v$. Since each $C_i$ has a
density,
$
    \mathbb P\left(
        C_i=\frac{k}{m}
    \right)
    =
    0,
    \;
    k=0,\ldots,m.
$
Hence, the assignment of the grid points to adjacent subintervals does
not affect the estimator almost surely.

The ideal Bernstein estimator admits the representation
$
    \widetilde f_{C,m,n}(u)
    =
    \frac{1}{n}
    \sum_{i=1}^n\xi_i^{(m)}
    \;\text{a.s.},
$
where
$
    \xi_i^{(m)}
    :=
    m\,p_{\kappa_m(C_i)}(m-1,u).
$
For each fixed $m$, the random variables
$\xi_1^{(m)},\ldots,\xi_n^{(m)}$ are independent and identically
distributed. Set
$
    \mu_m
    :=
    \mathbb E\bigl[\xi_1^{(m)}\bigr],
    \;
    \sigma_m^2
    :=
    \operatorname{Var}\bigl(\xi_1^{(m)}\bigr).
$
By Lemma~\ref{lem:ideal-bernstein-variance}, the variance admits the
following asymptotic expansion as $m\to\infty$:
\begin{equation}
    \sigma_m^2
    =
    f_C(u)\psi(u)m^{1/2}
    +
    o(m^{1/2}).
    \label{eq:sigma-m-asymptotic-lemma}
\end{equation}
Since $f_C(u)\psi(u)>0$, one has $\sigma_m^2>0$ for all sufficiently
large $m$.

We next establish a bound for the summands that is uniform with respect
to $i$. A standard upper bound for binomial probabilities gives
$
    \max_{0\leq k\leq m-1}
    p_k(m-1,u)
    \leq
    K_u m^{-1/2}
$
for all sufficiently large $m$, where $K_u>0$ depends only on $u$.
Consequently,
\begin{equation}
    0
    \leq
    \xi_i^{(m)}
    \leq
    K_u\sqrt m,
    \;
    i=1,\ldots,n.
    \label{eq:xi-bound-clt}
\end{equation}
To control the centering term without imposing differentiability on
$f_C$, define
$$
    \bar f_{C,k}
    :=
    m\int_{k/m}^{(k+1)/m}f_C(v)\dd v,
    \;
    k=0,\ldots,m-1,
$$
and let
$
    K_m
    \sim
    \operatorname{Binomial}(m-1,u).
$
Then
$
    \mu_m
    =
    \sum_{k=0}^{m-1}
    \bar f_{C,k}p_k(m-1,u)
    =
    \mathbb E\bigl[\bar f_{C,K_m}\bigr].
$
The binomial moment formulas give
$
    \mathbb E\left(\frac{K_m}{m}\right)
    =
    \frac{m-1}{m}u
    \xrightarrow[m\to\infty]{}
    u
$
and
$
    \operatorname{Var}\left(\frac{K_m}{m}\right)
    =
    \frac{(m-1)u(1-u)}{m^2}
    \xrightarrow[m\to\infty]{}
    0.
$
Therefore,
$
    \frac{K_m}{m}
    \xrightarrow[m\to\infty]{\mathbb P}
    u.
$
Since $f_C\in C([0,1])$, the Heine--Cantor
theorem implies that $f_C$ is uniformly continuous on $[0,1]$.
Moreover,
\[
    \left|
        \bar f_{C,K_m}-f_C(u)
    \right|
    \leq
    \sup_{
        v\in[K_m/m,(K_m+1)/m]
    }
    \left|
        f_C(v)-f_C(u)
    \right|.
\]
Since $K_m/m\to u$ in probability and the length of the corresponding
subinterval is $1/m$, uniform continuity yields
$
    \bar f_{C,K_m}
    \xrightarrow[m\to\infty]{\mathbb P}
    f_C(u).
$
Furthermore,
$
    0
    \leq
    \bar f_{C,K_m}
    \leq
    \lVert f_C\rVert_\infty.
$
Thus, the sequence $(\bar f_{C,K_m})_m$ is uniformly integrable, and
convergence in probability implies convergence in $L^1$. Consequently,
\begin{equation}
    \mu_m
    =
    \mathbb E\bigl[\bar f_{C,K_m}\bigr]
    \xrightarrow[m\to\infty]{}
    f_C(u).
    \label{eq:mu-m-convergence}
\end{equation}

In particular, there exists a constant $K_\mu>0$ such that
$
    |\mu_m|
    \leq
    K_\mu
$
for all sufficiently large $m$. Combining this bound with
\eqref{eq:xi-bound-clt}, and setting
$
    \widetilde K_u
    :=
    K_u+K_\mu,
$
we obtain
\begin{equation}
    \left|
        \xi_i^{(m)}-\mu_m
    \right|
    \leq
    \widetilde K_u\sqrt m,
    \;
    i=1,\ldots,n,
    \label{eq:centered-xi-bound-clt}
\end{equation}
for all sufficiently large $m$.
Let $m=m_n$ and define the triangular array
$
    Y_{i,n}
    :=
    \xi_i^{(m_n)}-\mu_{m_n},
    \;
    i=1,\ldots,n.
$
For each $n$, the random variables
$Y_{1,n},\ldots,Y_{n,n}$ are independent and identically distributed,
centered, and have common variance
$
    s_n^2
    :=
    \sigma_{m_n}^2.
$
By \eqref{eq:sigma-m-asymptotic-lemma},
$
    s_n
    =
    \bigl[f_C(u)\psi(u)\bigr]^{1/2}
    m_n^{1/4}
    \bigl(1+o(1)\bigr).
$

For every $\varepsilon>0$, \eqref{eq:centered-xi-bound-clt} gives
\[
    \max_{1\leq i\leq n}
    \frac{|Y_{i,n}|}
         {\varepsilon\sqrt n\,s_n}
    \leq
    \frac{\widetilde K_u\sqrt{m_n}}
         {\varepsilon\sqrt n\,s_n}
    =
    O\left(
        \left(
            \frac{m_n}{n^2}
        \right)^{1/4}
    \right)
    \xrightarrow[n\to\infty]{}
    0.
\]
Indeed,
$
    \frac{m_n}{n^2}
    =
    \frac{1}{n}
    \frac{m_n}{n}
    \xrightarrow[n\to\infty]{}
    0,
$
because $m_n/n\to0$. Hence, for every $\varepsilon>0$ and all
sufficiently large $n$,
$
    \left\{
        |Y_{i,n}|
        >
        \varepsilon\sqrt n\,s_n
    \right\}
    =
    \varnothing,
    \;
    i=1,\ldots,n.
$
It follows that
\[
    \frac{1}{ns_n^2}
    \sum_{i=1}^n
    \mathbb E\left[
        Y_{i,n}^2
        \mathds{1}_{
            \{
                |Y_{i,n}|
                >
                \varepsilon\sqrt n\,s_n
            \}
        }
    \right]
    \xrightarrow[n\to\infty]{}
    0.
\]
Thus, the Lindeberg condition is satisfied. The Lindeberg--Feller central
limit theorem therefore gives
\begin{equation}
    \frac{
        \sum_{i=1}^nY_{i,n}
    }{
        \sqrt n\,s_n
    }
    \xrightarrow[n\to\infty]{\mathcal D}
    \mathcal N(0,1).
    \label{eq:lindeberg-feller-ideal-clt}
\end{equation}
Finally,
$
n^{1/2}m_n^{-1/4}
    \left[
        \widetilde f_{C,m_n,n}(u)
        -
        \mathbb E\bigl[
            \widetilde f_{C,m_n,n}(u)
        \bigr]
    \right]
    =
    \frac{
        \sum_{i=1}^nY_{i,n}
    }{
        \sqrt n\,s_n
    }
    \frac{s_n}{m_n^{1/4}}.
$
Since
$
    \frac{s_n}{m_n^{1/4}}
    \xrightarrow[n\to\infty]{}
    \bigl[f_C(u)\psi(u)\bigr]^{1/2},
$
Slutsky's theorem and
\eqref{eq:lindeberg-feller-ideal-clt} yield
\eqref{eq:ideal-pointwise-clt-lemma}. This completes the proof.
\end{proof}
\section{Generalized central limit theorem for martingales}
\label{app:martingale-gclt}
This appendix collects the martingale limit results used throughout the
paper. The compensated counting martingales and the centered jump-size
martingales appearing in the main proofs are purely discontinuous.
Accordingly, the principal result below
is a generalized central limit theorem for quasi-left-continuous, locally
square-integrable martingales that allows both a continuous martingale part and
jumps. Its classical bracket and Lindeberg conditions are the ones verified in
the proofs of the stable limit theorems for the jump-intensity and Gamma-rate
estimators.
The following theorem, due to Touati \cite{Touati1991}, provides a generalized central limit theorem (GCLT) for martingales based on characteristic function techniques rather than the classical Lindeberg condition. 
Let $M=(M_t)_{t\geq 0}$ be a $d$-dimensional quasi-left continuous local martingale, locally square integrable, defined on a filtered probability space
$(\Omega,\mathcal{F},(\mathcal{F}_t)_{t\geq 0},\mathbb{P})$; (see Jacod and Shiryaev \cite{JS2003}). Let $V=(V_t)_{t\geq 0}$ be a deterministic family of non-singular $d\times d$ matrices.
For $u \in \mathbb{R}^d$, define the characteristic function
\begin{align*}
	\phi_t(u)
	:=
	\exp\left(
	-\frac{1}{2}u^*\langle M^c\rangle_t u
	+
	\int_0^t \int_{\mathbb{R}^d}
	\big(e^{i\langle u,x\rangle}-1-i\langle u,x\rangle\big)
	\,\nu^M(ds,dx)
	\right),
\end{align*}
where $\nu^M$ denotes the compensator of the jump measure associated with $M$.
Recall that any local martingale $M$ admits the decomposition
$
M = M_0 + M^c + M^d,
$
where $M^c$ is a continuous local martingale and $M^d$ is a purely discontinuous local martingale. Its quadratic variation is given by
$
[M]_t = \langle M^c\rangle_t + \sum_{0<s\leq t} \Delta M_s \Delta M_s^*,
$
and its predictable compensator is denoted by $\langle M\rangle_t$.
The classical martingale central limit theorem is obtained under the assumptions
$
(\mathcal{H}_1)\quad
V_t^{-1}\langle M\rangle_t (V_t^*)^{-1}
\stackrel{a.s.}{\longrightarrow} C,
\, (t\to\infty),
$
together with the Lindeberg condition
$
(\mathcal{H}')
\quad
\forall \delta>0,\quad
\int_0^t \int_{\mathbb{R}^d}
\|V_t^{-1}x\|^2
\mathbf{1}_{\{\|V_t^{-1}x\|>\delta\}}
\,\nu^M(ds,dx)
\stackrel{a.s.}{\longrightarrow} 0.
$
These conditions imply the characteristic function condition $(\mathcal{H})$ with
$
\eta = C^{1/2},
\,
\Phi_\infty(\eta,u)
=
\exp\left(-\frac{1}{2}u^*Cu\right).
$
\begin{Theo}[Generalized central limit theorem for martingales]
\label{GCLTtouati}
	Let $M=(M_{t})_{t\in \real_{+}}$ be a d-dimensional quasi-left
	continuous local martingale with $M_{0}=0$ and $V=(V_{t})_{t\in \real_{+}}$ a
	deterministic family of non-singular matrices. We define a
	probability $\mathcal{Q}$ on the space
	$\mathrm{C}(\mathcal{X},\mathbb{R}^{d})$ of continuous functions
	from $\mathcal{X}$ to $\mathbb{R}^{d}$ (where $\mathcal{X}$
	indicates a vector space of finite dimension). If
	the couple $(M,V)$ satisfies the following assumption
	$\mathcal{(H)}\, \left\{
	\begin{array}{lll}
		\phi_{t}((V_{t}^{*})^{-1}u)  \stackrel{a.s.}{\tto}  \phi_{\infty}(\eta,u),\quad \mbox{as}\;\; t\tto \infty,& &\\
		\phi_{\infty}(\eta,u)\neq 0 \quad a.s.,&&
	\end{array}
	\right.$ where $\eta$ denotes a r.v., possibly
	degenerated taking values in $\mathcal{X}$ 
    and
	$\phi_{\infty}(z,u)=\dint_{\mathbb{R}^{d}}^{}\exp\{i\langle
	u,\xi\rangle\}\pi(z,d\xi),\quad (z,u)\in \mathcal{X}\times
	\mathbb{R}^{d},$ denotes the Fourier transform of the
	conditional laws $(\pi(x,.),x\in\mathcal{X})$ of
	the probability $\mathcal{Q}$. Then
	$$(GCLT)~~~~~~~~Z_{t}:=V_{t}^{-1}M_{t}\stackrel{\mathcal{D}}{\tto}
	Z_{\infty}:=\Sigma(\eta),\quad(t\tto \infty),$$ in a
	stable
	manner where $(\Sigma(z),z\in\mathcal{X})$ is a $\mathcal{Q}$ law process independent of the r.v. $\eta$.
\end{Theo}
It is worth noting that the classical Lindeberg condition $(\mathcal{H}')$ is generally easier to verify in practice than the characteristic function condition $(\mathcal{H})$; (see Fathallah et al. \cite[Remark~1.1]{FK2012}).

\nocite{JP2012}
\nocite{B1999}
\nocite{L1992}
\nocite{L1994}
\nocite{Sato1999}
\nocite{V2000}
\nocite{Li2011}
\nocite{KS1997}
\nocite{D1989}
\nocite{LL1996}
\nocite{K2014}
\nocite{JM1976}
\nocite{JS2003}
\nocite{S1991}
\nocite{P1972}
\nocite{OR1997}
\nocite{O1998}
\nocite{LS2001}
\nocite{LM2015}
\nocite{JRT2016b}
\nocite{IW1989}
\nocite{HMZ2011}
\nocite{BBKP2019_Heston}
\nocite{BBKP2018}
\nocite{Alfonsi2005}
\nocite{BK2012}
\nocite{BK2013}
\nocite{KS2008}
\nocite{FMS2013}
\nocite{KM2012}
\nocite{BDZP2013}
\nocite{BZP2015}
\nocite{BZP2015'}
\nocite{CIR1985}
\nocite{F1951}
\nocite{FL2010}
\nocite{Touati1991}
\nocite{Touati1993}
\nocite{Kutoyants2004}
\nocite{BDF2025}

\bibliographystyle{plain}
\bibliography{biblioBAJD-Effetalea_clean}

\end{document}